\documentclass[
  reprint,            
  amsmath,amssymb,
  aps,
  superscriptaddress, 
  longbibliography
]{revtex4-2}

\usepackage[english]{babel}
\usepackage{textcomp}
\usepackage{url}
\usepackage{tcolorbox}

\usepackage[T1]{fontenc}
\usepackage{mathpazo}

\usepackage{amsfonts}
\usepackage{amsmath}
\usepackage{amssymb}
\usepackage{latexsym}
\usepackage{mathtools}

\usepackage{algorithm}
\usepackage{algpseudocode}

\usepackage{qcircuit}

\usepackage{graphicx}

\usepackage[caption=false]{subfig}

\usepackage[colorlinks=true,linkcolor=blue,citecolor=blue,urlcolor=blue]{hyperref}
\usepackage[nameinlink]{cleveref}

\hypersetup{pdfpagemode=UseNone}
\usepackage{amsthm}
\newtheorem{theorem}{Theorem}
\newtheorem{lemma}[theorem]{Lemma}
\newtheorem{prop}[theorem]{Proposition}

\theoremstyle{definition}

\newtheorem{problem}{Problem}

\DeclarePairedDelimiter\abs{\lvert}{\rvert}

\newcommand{\ket}[1]{| #1 \rangle}
\newcommand{\bra}[1]{\langle #1 |}

\def\A{\mathcal{A}}

\def\C{\mathcal{C}}
\def\D{\mathcal{D}}
\def\E{\mathcal{E}}

\def\P{\mathcal{P}}

\newcommand{\CC}{\mathbb{C}}
\newcommand{\la}{\langle}
\newcommand{\ra}{\rangle}
\newcommand{\be}{\begin{equation}}
\newcommand{\ee}{\end{equation}}

\usepackage{xcolor}

\begin{document}


\title{Observation of Improved Accuracy over Classical Sparse Ground State Solvers \\ using a Quantum Computer}


\author{William Kirby}
\email{william.kirby@ibm.com}
\affiliation{IBM Quantum, T. J. Watson Research Center, Yorktown Heights, NY 10598, USA}

\author{Bibek Pokharel}
\affiliation{IBM Quantum, T. J. Watson Research Center, Yorktown Heights, NY 10598, USA}

\author{Javier Robledo Moreno}
\affiliation{IBM Quantum, T. J. Watson Research Center, Yorktown Heights, NY 10598, USA}

\author{Kevin C. Smith}
\affiliation{IBM Quantum, IBM Research Cambridge, Cambridge, MA 02142, USA}

\author{Sergey Bravyi}
\affiliation{IBM Quantum, T. J. Watson Research Center, Yorktown Heights, NY 10598, USA}

\author{Abhinav Deshpande}
\affiliation{IBM Quantum, IBM Research Silicon Valley, San Jose, CA 95141, USA}

\author{Constantinos Evangelinos}
\affiliation{IBM Quantum, IBM Research Cambridge, Cambridge, MA 02142, USA}

\author{Bryce Fuller}
\affiliation{IBM Quantum, T. J. Watson Research Center, Yorktown Heights, NY 10598, USA}

\author{James R. Garrison}
\affiliation{IBM Quantum, T. J. Watson Research Center, Yorktown Heights, NY 10598, USA}

\author{Ben Jaderberg}
\affiliation{IBM Quantum, IBM Research Europe, Hursley, Winchester, SO21 2JN, United Kingdom}

\author{Caleb Johnson}
\affiliation{IBM Quantum, T. J. Watson Research Center, Yorktown Heights, NY 10598, USA}

\author{Petar Jurcevic}
\affiliation{IBM Quantum, T. J. Watson Research Center, Yorktown Heights, NY 10598, USA}

\author{Su-un Lee}
\affiliation{Pritzker School of Molecular Engineering, The University of Chicago, Chicago, IL 60637, USA}
\affiliation{IBM Quantum, T. J. Watson Research Center, Yorktown Heights, NY 10598, USA}

\author{Simon Martiel}
\affiliation{IBM Quantum, IBM France Lab, Saclay, France}

\author{Mario Motta}
\affiliation{IBM Quantum, T. J. Watson Research Center, Yorktown Heights, NY 10598, USA}

\author{Seetharami Seelam}
\affiliation{IBM Quantum, T. J. Watson Research Center, Yorktown Heights, NY 10598, USA}

\author{Oles Shtanko}
\affiliation{IBM Quantum, T. J. Watson Research Center, Yorktown Heights, NY 10598, USA}

\author{Kevin J. Sung}
\affiliation{IBM Quantum, T. J. Watson Research Center, Yorktown Heights, NY 10598, USA}

\author{Minh Tran}
\affiliation{IBM Quantum, T. J. Watson Research Center, Yorktown Heights, NY 10598, USA}

\author{Vinay Tripathi}
\affiliation{IBM Quantum, T. J. Watson Research Center, Yorktown Heights, NY 10598, USA}

\author{Kazuhiro Seki}
\affiliation{RIKEN Center for Quantum Computing, Wako, Saitama 351-0198, Japan}

\author{Kazuya Shinjo}
\affiliation{RIKEN Center for Emergent Matter Science, Wako, Saitama 351-0198, Japan}

\author{Han Xu}
\affiliation{RIKEN Center for Computational Science, Kobe, Hyogo 650-0047, Japan}

\author{Lukas Broers}
\affiliation{RIKEN Center for Computational Science, Kobe, Hyogo 650-0047, Japan}

\author{Tomonori Shirakawa}
\affiliation{RIKEN Center for Computational Science, Kobe, Hyogo 650-0047, Japan}
\affiliation{RIKEN Center for Quantum Computing, Wako, Saitama 351-0198, Japan}
\affiliation{RIKEN Pioneering Research Institute, Wako, Saitama 351-0198, Japan}
\affiliation{RIKEN Center for Interdisciplinary Theoretical and Mathematical Sciences, RIKEN, Wako 351-0198, Japan}

\author{Seiji Yunoki}
\affiliation{RIKEN Center for Quantum Computing, Wako, Saitama 351-0198, Japan}
\affiliation{RIKEN Center for Computational Science, Kobe, Hyogo 650-0047, Japan}
\affiliation{RIKEN Center for Emergent Matter Science, Wako, Saitama 351-0198, Japan}
\affiliation{RIKEN Pioneering Research Institute, Wako, Saitama 351-0198, Japan}

\author{Kunal Sharma}
\affiliation{IBM Quantum, T. J. Watson Research Center, Yorktown Heights, NY 10598, USA}

\author{Antonio Mezzacapo}
\affiliation{IBM Quantum, T. J. Watson Research Center, Yorktown Heights, NY 10598, USA}

\begin{abstract}

Demonstrating quantum advantage over classical algorithms for ground state energy problems is an outstanding open problem in quantum computation. We experimentally demonstrate that a quantum algorithm can outperform classical selected configuration interaction (SCI) methods, a key family of techniques used in computational chemistry and condensed matter physics. We construct a class of local Hamiltonian problems with sparse ground states, and show that SCI fails to find the ground state of a 49-qubit instance. We then show that sample-based Krylov quantum diagonalization, run on an IBM Heron R3 processor, succeeds at the same task. While the problem is solvable classically using iterative solvers designed to target our Hamiltonian construction, this work resolves the question of whether a sample-based quantum diagonalization algorithm can outperform standard SCI heuristics.
    
    
\end{abstract}

\maketitle


Quantum computing holds great promise for expanding the frontier of feasible computation across a range of challenging problems.
One such longstanding and promising problem is the diagonalization of high-dimensional Hamiltonian matrices representing energies of many-body quantum systems.
In particular, estimating low energies and low-energy states is a key subroutine in many calculations across physics and chemistry.
The problem can be formulated as the approximation of extremal eigenvalues and eigenvectors of Hamiltonians, whose Hermitian matrix representations have exponential dimensions in the size of the quantum system.

Sample-based algorithms have recently emerged as leading candidates to offer the first quantum advantages over classical diagonalization methods~\cite{kanno2023quantum,robledomoreno2024chemistry,sugisaki2025hamiltonian,mikkelsen2025quantum,yu2025quantum}.
These methods are based on using quantum circuits to sample configurations of the simulated system, then using a classical processor to project and diagonalize the Hamiltonian in the subspace spanned by the sampled configurations.

This approach offers two key advantages over other quantum diagonalization algorithms.
First, since the projection and diagonalization are implemented classically, the resulting approximate energies are \emph{variational} up to the precision of the classical diagonalization, meaning that they upper bound the true ground state energy of the full Hamiltonian effectively exactly.
This means that the energies can be compared to results of classical approximate diagonalization methods, and if the energies from the quantum algorithm are lower then they must also have lower error, even if the true ground state energy is unknown.
Second, the basis of sampled configurations forms a classical witness for the approximate ground state energy.
This means that results can be efficiently validated by a third party without a quantum computer.
These two features together mean that if we can achieve a quantum advantage using a sample-based diagonalization method, it will be both \emph{unconditional} and \emph{classically verifiable}.

\begin{figure*}[t]
    \centering
    \includegraphics[width=0.6\linewidth]{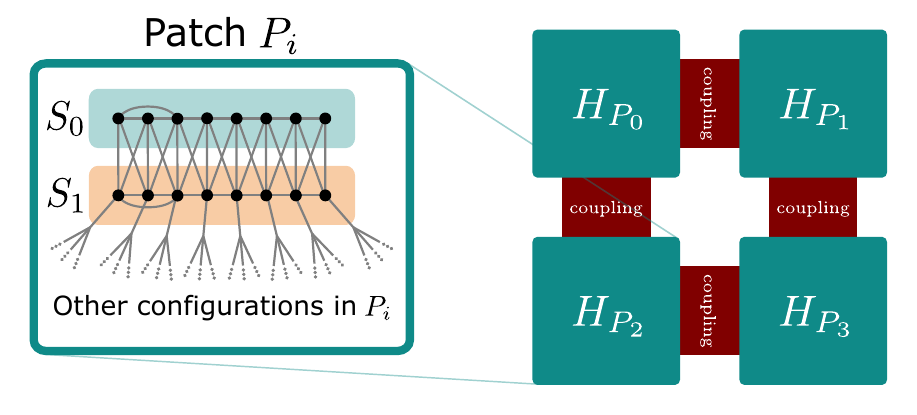}
    \caption{ {\bf Sketch of the Hamiltonian construction}. The qubit layout is partitioned into disjoint patches. Within each patch $P_i$, the set $S_0$ of support configurations is coupled to a set $S_1$ of first-order configurations, which are then coupled to the remaining configurations of the qubits in the patch. This defines the patch Hamiltonian $H_{P_i}$. The patches are themselves coupled by $H_\mathrm{coupling}$.}
    \label{fig:hamiltonian_construction_schematic}
\end{figure*}

In order for the quantum algorithm to achieve advantage, it needs to produce configurations that are of higher quality with respect to its classical counterparts. Well-known classical comparators are selected configuration interaction (SCI)~\cite{Huron_1973_CIPSI,tubman2016deterministic_SCI,holmes2016heat_SCI,holmes2016efficient_SCI,sharma2017semistochastic_SCI,zhang2025TrimCI_SCI} methods, which refer to a family of heuristics based on choosing a basis of configurations using a classical iteration.
The heuristic component of SCI is the selection rule by which configurations are kept or discarded during the iteration; these selection rules are typically based on perturbation theory~\cite{Huron_1973_CIPSI,tubman2016deterministic_SCI,holmes2016heat_SCI,holmes2016efficient_SCI,sharma2017semistochastic_SCI}.
For example, ``Configuration Interaction by Perturbatively Selecting Iteratively'' (CIPSI)~\cite{Huron_1973_CIPSI} employs the cost function
\begin{equation}
\label{eq:perturbative_cf}
    f(x) = \frac{|\langle x|H|\psi_\mathrm{prev}\rangle|}{\langle x|H|x\rangle-E_\mathrm{prev}},
\end{equation}
where $x$ is the new proposed configuration, and $|\psi_\mathrm{prev}\rangle$ and $E_\mathrm{prev}$ are the approximate ground state from the previous iteration and its energy.
This $f(x)$ is simply the first-order perturbative estimate of the coefficient of $x$ in the ground state.
Prior experimental demonstrations of sample-based quantum diagonalization (SQD), although reaching unprecedented problem sizes with quantum computers, have not been able to surpass SCI~\cite{robledomoreno2024chemistry, yu2025quantum}.
We will show that this is no longer the case for a class of quantum Hamiltonians.

Here, we design a family of synthetic Hamiltonians, then experimentally execute the sample-based Krylov quantum diagonalization (SKQD) algorithm~\cite{sugisaki2025hamiltonian,mikkelsen2025quantum,yu2025quantum} alongside classical SCI methods as part of a suite of classical heuristics.
In SKQD, the quantum circuits used for sampling generate a quantum Krylov space, i.e., they are an equally-spaced grid of time-evolved states starting from a single reference state.
This variant of SQD offers provable convergence under certain criteria similar to those required by quantum phase estimation~\cite{kitaev1995quantum}, with the additional requirement that the ground state be approximately sparse.
We engineer our family of Hamiltonians to satisfy these criteria, so that they form a testbed for benchmarking SKQD and classical sparse ground state solvers like SCI.
We then select a 49-qubit Hamiltonian belonging to the family as a test case for our experiments.

We execute the quantum experiments on an IBM Heron r3 quantum processor, which consists of 156 fixed-frequency transmon qubits with tunable couplers in a heavy-hexagon layout.
We compare the resulting energy accuracy to classical results on the same problem, taking advantage of the variationality property discussed above.

The output of SQD is a sparse vector approximation of the ground state, and the output energy is the energy of this sparse vector.
In the physically motivated problems studied so far with SQD~\cite{robledomoreno2024chemistry,yu2025quantum,Liepuoniute2025_SQD,shajan2025towards_SQD,shajan2025proteins_SQD,barison2025quantum_SQD,kaliakin2024accurate_SQD,piccinelli2025SQDRIFT_SQD,shirakawa2025closedloop_SQD,sriluckshmy2025ghostgutzwiller_SQD,smith2025quantumcentricsimulationhydrogenabstraction_SQD,Danilov2025_sqd,D5CP02202A_SQD}, the fact that SCI has also been able to achieve high accuracy indicates to a good approximation the ground state is obtained by expanding perturbatively around some core set of relatively few key configurations, such as the Hartree-Fock configuration in the electronic structure application.
However, all that can be said with certainty from a classical methods perspective is that, for a local Hamiltonian with an exactly sparse ground state, a spectral gap scaling as $1/\text{poly}(n)$ for system size $n$, and a sparse guiding state whose overlap with the ground state also scales as $1/\text{poly}(n)$, a truncated variant of the power method can approximate the ground state formally efficiently~\cite{yuan2013truncated}.
This is not a perturbative method, however, and can perform poorly relative to heuristics like SCI in practice.
There is no fundamental reason that every approximately sparse ground state must possess the perturbative character exhibited by the examples studied to date.

We set out to explore the possibilities of sparse ground states by explicitly designing a family of local Hamiltonians whose ground state sparsity is not due to a perturbative expansion.
In addition to providing test cases for comparing quantum and classical methods, these Hamiltonians are a starting point in the search for models in physics and chemistry that share similar properties and classical hardness.
The construction begins by selecting a qubit lattice to establish a notion of geometric locality.
We then partition this lattice into a small number $n_\text{patch}$ of patches, and build local Hamiltonians for each patch.

To construct these local ``patch Hamiltonians,'' we select a subset $S_0$ of configurations of the patch to form the support of the local ground state.
These \emph{support configurations} are selected to be connected by a chain of terms defined locally on the lattice.
We also select a second set $S_1$ of equally-many configurations that can be similarly ordered and connected to their counterparts in the first set.
Let the \emph{Hamiltonian graph} refer the graph in which nodes correspond to configurations, and edges to nonzero Hamiltonian matrix elements connecting those configurations.
We include in the Hamiltonian local terms such that the subgraph of the Hamiltonian graph on $S_0\cup S_1$ forms two adjacent chains with nearest-neighbor edges as well as next-nearest-neighbor edges that cross between the chains.
The $0$th and $2$nd configurations in each chain are also connected by a next-nearest-neighbor edge.
This graph is shown in \Cref{fig:hamiltonian_construction_schematic}.

This Hamiltonian graph allows us to choose interactions (i.e., edge weights, or equivalently Hamiltonian matrix elements) such that the lowest energy state supported on $S_0$ is the lowest energy state of the entire patch.
Ensuring that the lowest energy state of $S_0$ is an eigenstate of the patch Hamiltonian is accomplished by choosing the interactions connecting $S_0$ and $S_1$ such that they annihilate the lowest energy state of $S_0$.
Forcing it to be the lowest eigenstate of the patch is then a matter of adjusting the local diagonal (potential) terms.

The second key component of the construction is the choice interactions connecting the support configurations themselves.
These are selected such that the lowest energy state obtained by projecting and diagonalizing on any proper subset of $S_0$ is supported only on the first two configurations.
Only when all eight elements of $S_0$ are included does a level-crossing occur such that the full lowest energy state is supported on all eight.
This property breaks the iterative, perturbative approach to finding the ground state provided the zeroth element of $S_0$ is provided as the initial reference state.
For example, in CIPSI each intermediate ground state approximation during the iteration will be supported only on the first two support configurations, so after the third is found the numerator of the cost function \eqref{eq:perturbative_cf} will be zero for all further support configurations.

The final step in the construction is to couple the patches with additional terms in the Hamiltonian, so that the system is interacting.
These additional terms are chosen to be positive semidefinite, and to annihilate the support configurations $S_0$, which together imply that they do not change the ground state.
So in the end, we have an interacting Hamiltonian whose ground state is the tensor product of the ground states of the patches: hence its sparsity is $|S_0|^{n_\text{patch}}$.

In iterative classical solvers in which the iteration is first-order expansion of a configuration pool by the action of the Hamiltonian, like SCI, each iteration can only generate a step along the support configuration path of one patch at a time.
Since each path requires 6 steps to traverse (see \Cref{fig:hamiltonian_construction_schematic}), this sets the minimum number of steps required for any such iterative solver at $6n_\text{patch}$.

For our experimental demonstrations, classical and quantum, we provide three inputs to the solvers:
\begin{enumerate}
    \item A representation of the Hamiltonian instance $H$ in the Pauli basis.
    \item A single initial configuration $|x_0\rangle$.
    \item A circuit defining a single Trotter step under $H$, which defines a possible mapping of $H$ to a two-dimensional qubit lattice.
\end{enumerate}
The specific Hamiltonian instance we employ in our demonstrations is defined on a 49-qubit heavy-hexagon lattice, the native qubit connectivity lattice of IBM Heron quantum processors.
The layout is shown in \Cref{fig:v57_layout}.
The lattice is partitioned into $n_\text{patch}=3$ patches.
The support configurations $S_0$ and the first-order configurations $S_1$ correspond via an encoding to alternating qubits along a 16-qubit path defining each patch, so $|S_0|=8$ and hence the ground state is supported on $8^3=512$ configurations.
The one additional padding qubit is coupled to the remainder and fixed to either $|0\rangle$ or $|1\rangle$ in the ground state, depending on the (randomized) encoding.
Complete details are given in the Supplemental Material.

\begin{figure}[ht]
    \centering
    \includegraphics[width=\linewidth]{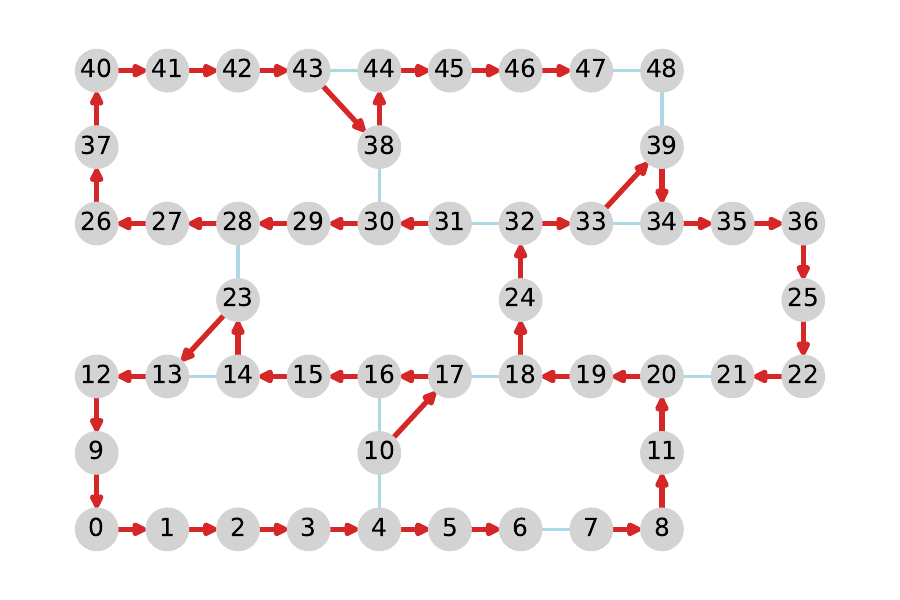}
    \caption{{\bf Patch layout for the problem Hamiltonian on a heavy-hex graph}. The paths defining each patch are shown by the red arrows. Patches are defined by paths with even-indexed qubits corresponding to the support configurations: in this figure, the directed edges point in the direction of increasing index. The edges in the physical qubit layout are shown in light blue wherever they do not overlap with the path edges. Qubit $48$ in the upper right corner is a padding qubit, to complete the sixth hex.}
    \label{fig:v57_layout}
\end{figure}

For our quantum solution, we use the SKQD algorithm as discussed above.
This requires sampling from time evolution circuits corresponding to the Hamiltonian of interest.
We use two main techniques to transpile these.

First, we employ Rustiq~\cite{martiel2025rustiq,debrugiere2024faster}, a software package to compile arbitrary linear combinations of Pauli operators into corresponding Trotter steps constructed from local Pauli rotation gates on a specified layout.
The resulting number of CZ gates (the native entangling gates on an IBM Heron processor) required for a single Trotter step is 324.
This limited us to only a few Trotter steps in practice without significant loss of coherence, so we augmented these circuits with approximate time evolutions using approximate quantum compilation (AQC).
AQC is a strategy for compressing quantum circuits into lower depth and gate count ans\"atze, when tensor network simulation of the circuits is feasible as in our case.
Our full hybrid strategy is to implement each AQC circuit, then append 0, 1, or 2 Trotter steps.
The motivation for this is that it essentially combines two different decays of signal: from the AQC we have a decay due to coherent error (we are preparing the ``wrong'' state), while from the Trotter approximations we have a decay that is well-modeled by Pauli noise channels, i.e., dominated by incoherent error.

We sample from AQC circuits corresponding to Krylov timesteps 1 through 20, each with up to 2 Trotter steps appended, although ultimately only up to 16 timesteps of AQC were required to reach the exact ground state.
In total, around 133 million shots were taken in QPU runtime of about 19 hours (with 97 million shots across AQC timesteps $1,...,16$).
Details of the circuits and shot allocation are given in the Supplemental Material.

Before diagonalizing the Hamiltonian within the corresponding basis, we employ a preliminary filter to reduce the diagonalization dimension.
This filter is motivated by the fact that any configurations in the pool that are not connected by the Hamiltonian to any others in the pool would be eigenstates of the projected Hamiltonian, so it is unnecessary to include them in the diagonalization.
We therefore remove all such configurations, which account for the vast majority of the configurations in the pool, leaving behind only $\sim10.3$ million.
After filtering, we project the Hamiltonian onto the remaining pool of $\sim10.3$ million configurations, then diagonalize the resulting matrix.

On the classical side, we apply a representative collection of off-the-shelf SCI methods, specifically CIPSI~\cite{Huron_1973_CIPSI}, ASCI~\cite{tubman2016deterministic_SCI}, HCI~\cite{holmes2016heat_SCI, sharma2017semistochastic_SCI, holmes2016efficient_SCI}, and TrimCI~\cite{zhang2025TrimCI_SCI}. We test these classical heuristics by sweeping their hyperparameters and selecting the most performant settings in accuracy, but none are able to find the exact ground state energy.

Additionally, we test three non-SCI iterative solvers: truncated power method~\cite{yuan2013truncated}, truncated Arnoldi's method, and diagonal ranking.
The latter two methods are novel (to our knowledge), and were inspired by vulnerabilities that we encountered in the Hamiltonian construction technique.
Finally, we applied the density matrix renormalization group (DMRG).
We similarly performed hyperparameter sweeps for all of these techniques; details of all algorithms and implementations are given in the Supplemental Material.

\begin{figure*}[ht]
    \centering
    \includegraphics[width=1\linewidth]{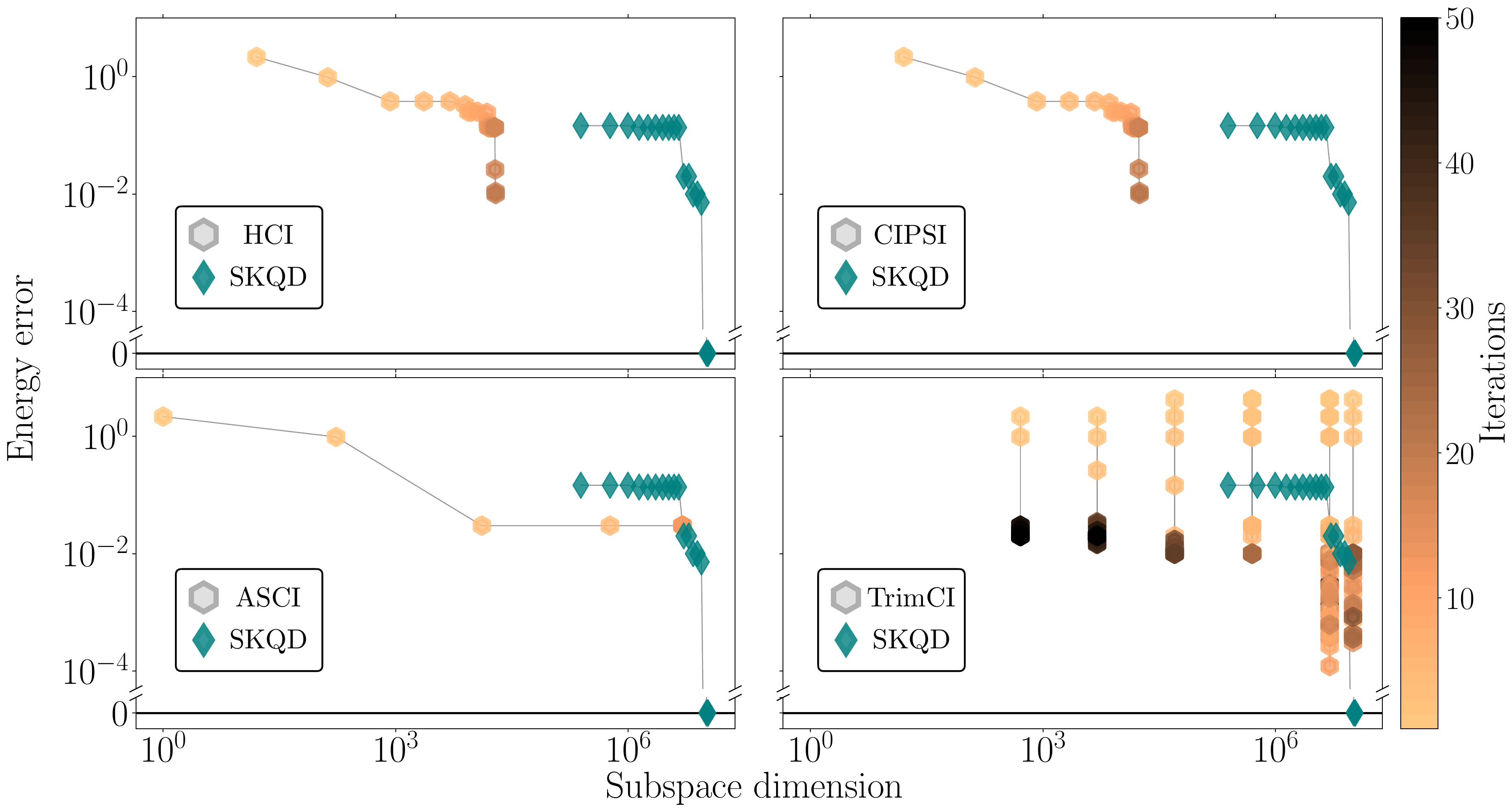}
    \caption{\textbf{Improved accuracy of quantum (SKQD) against off-the-shelf classical SCI methods.} Comparison of the ground state energy reached for a given subspace dimension for different SCI methods, and the SKQD experiment executed on quantum hardware. Different SCI methods are shown in different panels by the hexagonal markers. The different colors show the number of iterations. The results obtained with SKQD are shown by blue diamond markers in each panel. The sequence of SQKD points is with respect to Krylov dimension.}
    \label{fig:results_good}
\end{figure*}

\begin{figure*}[ht]
    \centering
    \includegraphics[width=1\linewidth]{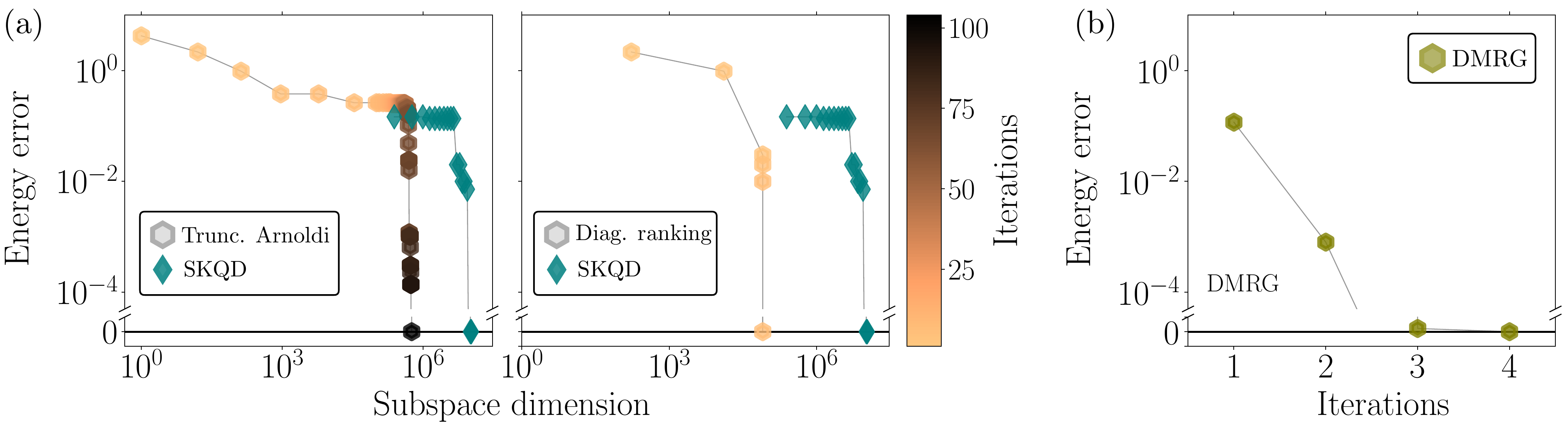}
    \caption{\textbf{Performance of specialized sparse iterative solvers and DMRG.} \textbf{(a)} compares the ground state energy reached for a given subspace dimension for the sparse iterative solvers introduced in this manuscript, as shown by the hexagonal markers. The results obtained with SKQD are shown by blue diamond markers in each panel. \textbf{(b)} Convergence of DMRG as a function of the number of sweeps. }
    \label{fig:results_bad}
\end{figure*}

The results of these experiments, quantum and classical, are shown in \Cref{fig:results_good} and \Cref{fig:results_bad}.
Specifically, the panels in \Cref{fig:results_good} show the off-the-shelf SCI methods that we applied, compared to the results of SKQD.
We only show the best performing hyperparameter settings for the classical methods.
As the plots show, none of the SCI methods are able to find the exact ground state energy.
SKQD successfully finds it with a Krylov dimension of 17, i.e., the initial configuration plus samples from 1 through 16 timesteps (the quantum data is the same in all panels).

CIPSI and HCI were not able to reach diagonalization dimensions comparable to those required in SKQD, since by design they cannot reach arbitrary dimensions, instead terminating when no new configurations are added to the pool from one iteration to the next.
In the cases of ASCI and TrimCI, we have direct control over the diagonalization dimension.
Therefore, we simply swept the dimension until it reached the dimension required by SKQD.
The result is that SKQD obtains an advantage over off-the-shelf SCI strictly in accuracy. 

Results of truncated Arnoldi's method and diagonal ranking, our custom methods, and DMRG are shown in \Cref{fig:results_bad}.
All three are able to find the exact ground state energy.
Truncated power method can as well, but it is a similar type of algorithm to truncated Arnoldi's method, and the latter is more performant in this instance; details can be found in the Supplemental Material.
It is also worth mentioning that diagonalizing in the basis sampled from the AQC circuits alone is able to find the exact ground state energy with a Krylov dimension of 20; however, this is still higher than the dimension required using our hybrid AQC-Trotter evolution circuits.

Separately we leveraged tensor network simulations of the quantum circuits, primarily as a guiding and diagnostic tool for the quantum experiments.
However, such methods also serve as a classical adversary for general quantum advantage, i.e., where the hybrid quantum-classical workflow of SKQD is replaced with a fully classical one via direct simulation and sampling of the quantum circuits.
For the Hamiltonian in our experiments, we found that all support configurations can be found by sampling from AQC target states using MPS methods, albeit requiring higher Krylov dimension than using our hybrid time evolution circuits in the quantum experiment.
We therefore did not run end-to-end simulations of SKQD using 2D tensor networks due to the comparatively higher sampling cost.
Nonetheless, we observe that SKQD converges at modest Krylov dimension such that the circuits remain classically simulable and, using prior knowledge of the support configurations and probability amplitude estimation, estimate that $\sim10^5$ samples are needed from the first 20 Krylov basis states for convergence.

Our results show that sample-based quantum diagonalization algorithms can yield distributions that enable efficient sample coverage of the ground state support, even on a noisy quantum computer.
Even further, they can do so in a regime where the ground state support is not closely connected perturbatively, since SCI fails to find the ground state for our instance.
The only reason the explicitly perturbative SCI heuristics, CIPSI and ASCI, are able to progress at all beyond the first two support configurations in each patch is due to higher-order paths introduced by coupling the patches.
This is the regime we must target in order to achieve quantum advantage using sample-based quantum diagonalization: our experiment provides the first example of successful quantum execution in this regime.

Achieving hardness for DMRG is naturally the next challenge.
One observation is that for any $k$-sparse quantum state, there exists an MPS with bond dimension at most $k$ that exactly represents the state.
Therefore, if we restrict ourselves to polynomial sparsity, as required for any sample-based diagonalization algorithm, we guarantee that the state is efficiently representable by an MPS.
This means that there is no representability barrier for DMRG in principle for a sparse ground state problem.
In practice, the cost of DMRG scales with the cube of the bond dimension, so for example a $10^4$-sparse quantum state with an MPS of bond dimension $10^4$ would be pushing the practical limits, at least based on results in the literature.

A more theoretical approach to achieving hardness for DMRG is via trainability, since DMRG relies on a convergence of an optimization consisting of sweeps over local effective diagonalizations.
This optimization could be broken either by causing it to get stuck in a local minimum, i.e., an excited state, or by causing it to fail to converge entirely.

A natural, forward-looking question is whether these results can be generalized to physically-motivated Hamiltonians such as condensed matter and chemistry use cases.
Our synthetic Hamiltonians are defined locally on 2D lattices, but their interactions are nonuniform yet highly structured, which is the primary way they differ from many physical lattice models.
That said, they suggest general principles we should be looking for to identify advantage candidates among physical Hamiltonians.
In particular, the property we described above of possessing a level crossing that separates partial ground states in configuration subspaces from full ground states, is key in achieving hardness of sparse ground states for perturbation-theory methods.
For the moment, showing that SKQD can outperform the off-the-shelf SCI methods, widely and successfully used in computational chemistry, is a significant milestone in the quest for quantum advantage in ground state problems.

\section*{Software}

The Hamiltonian instance used in the quantum and classical experiments in this work is available at~\footnote{\url{https://github.com/quantum-advantage-tracker/quantum-advantage-tracker.github.io/tree/main/data/variational-problems/hamiltonians/guided_sparse_ground_state_problem/49Q_v1}}.
Tensor network simulations using belief propagation were carried out using \textsc{TensorNetworkQuantumSimulator.jl}~\cite{TindallRudolph_TNQS}.
Time evolution circuits were synthesized using \textsc{Rustiq}~\cite{martiel2025rustiq} and approximately compiled using the \textsc{qiskit-addon-aqc-tensor} package~\cite{qiskit-addon-aqc-tensor}.
DMRG calculations were carried out using the \textsc{block2} package~\cite{block2}.
The CIPSI, HCI, ASCI, TrimCI and diagonal ranking methods were implemented using functionalities included in the \textsc{qiskit-addon-sqd} package~\cite{qiskit-addon-sqd}. 

\section*{Acknowledgments}

The authors thank Zachary Chin, Bill Fefferman, Soumik Ghosh, Abdullah Ash Saki, and Joseph Tindall for valuable conversations.
Research at RIKEN is supported by the New Energy and Industrial Technology Development Organization (NEDO), Japan (Project No. JPNP20017), the Japan Science and Technology Agency (JST) through COI-NEXT (Grant No. JPMJPF2221), and the Ministry of Education, Culture, Sports, Science and Technology (MEXT) through the Program for Promoting Research of the Supercomputer Fugaku (Grant No. MXP1020230411). Additional support is provided by the UTokyo Quantum Initiative, the RIKEN TRIP Initiative (RIKEN Quantum), and the Center of Excellence (COE) Research Grant in Computational Science from Hyogo Prefecture and Kobe City through the Foundation for Computational Science. K. Seki is also supported by Japan Society for the Promotion of Science (JSPS) KAKENHI (Grant No. JP26K06954). K. Shinjo is also supported by the JSPS KAKENHI (Grant No. JP23K13066). T. Shirakawa is also supported by JSPS KAKENHI (Grant No. JP26K06972). S. Yunoki is also supported by JSPS KAKENHI (Grant No. JP21H04446).

\bibliographystyle{ieeetr}
\bibliography{Ref}

\appendix

\section{Sample-based Krylov quantum diagonalization}
\label{sec:skqd_intro}

Many quantum algorithms have been developed for ground state approximation.
Some of these are only suitable for future fault-tolerant quantum computers (e.g.,~\cite{kitaev1995quantum,lin2022heisenberglimited,dong2022groundstate,ding2023even,wang2023quantumalgorithm}), but others are specifically designed with feasibility on noisy quantum hardware in mind.
A well-known example is the variational quantum eigensolver (VQE)~\cite{peruzzo2014variational}, but scaling VQE to large systems has proven challenging due to the high numbers of shots and difficult optimizations required.
More recently, \emph{quantum subspace diagonalization} methods have emerged as promising approaches to estimating spectral properties on pre-fault-tolerant quantum processors.
Among these, the Krylov quantum diagonalization (KQD) method~\cite{parrish2019quantum} has been experimentally demonstrated for quantum many-body systems of up to 56 spins~\cite{yoshioka2024diagonalization}.
Another subspace approach, quantum-selected configuration interaction (QSCI)~\cite{kanno2023quantum}, is based on the idea of adapting SCI to a quantum algorithm.
This technique was later developed into an error-mitigated, scalable workflow called sample-based quantum diagonalization (SQD) for a quantum-centric supercomputer, and experimentally demonstrated for chemistry Hamiltonians~\cite{robledomoreno2024chemistry, Liepuoniute2025_SQD, shajan2025towards_SQD, shajan2025proteins_SQD, barison2025quantum_SQD, kaliakin2024accurate_SQD, piccinelli2025SQDRIFT_SQD, shirakawa2025closedloop_SQD, sriluckshmy2025ghostgutzwiller_SQD, smith2025quantumcentricsimulationhydrogenabstraction_SQD, Danilov2025_sqd, D5CP02202A_SQD}.
Recently, a hybrid approach combining SQD and KQD, termed sample-based Krylov quantum diagonalization (SKQD), was introduced and experimentally validated in~\cite{yu2025quantum}. 

These last two algorithms have the distinctive property that, like SCI, they work by building up a subspace spanned by computational basis states or configurations (we will use the latter term, for brevity).
This means that, unlike other quantum diagonalization algorithms, the projection of the Hamiltonian onto those configurations as well as the subsequent diagonalization can both be carried out classically~--- the quantum computer is only used to identify the configurations in the first place.
This gives SQD and SKQD the pair of properties mentioned in the main text that, together, make them unique among quantum diagonalization algorithms:
\begin{enumerate}
    \item Output energies are classically verifiable, by projecting and diagonalizing on the basis of configurations, which has polynomial size.
    \item Output energies are variational up to the precision of the \emph{classical} diagonalization.
\end{enumerate}
These two properties imply that, if one of these quantum algorithms yields a lower energy for some problem than any classical algorithm is able to achieve with equivalent cost, this will constitute an unconditional and classically verifiable quantum advantage~\cite{lanes2025framework}.

The price that is paid for these advantages is that SQD and SKQD are limited to sparse approximations of ground states, since by definition a state that can be represented in a polynomial-sized subset of the computational basis within exponential total dimension is sparse in that basis.
In spite of this limitation, the fact that they can in principle yield unconditional and verifiable quantum advantage makes them very appealing.

SKQD has an additional property beyond SQD that increases its appeal still further: under a handful of conditions that we discuss in detail in \Cref{sec:skqd_intro}, it can be proven to converge systematically to a sparse ground state~\cite{yu2025quantum}.
SKQD builds upon the rapidly maturing family of quantum Krylov algorithms for approximating ground state energies~\cite{parrish2019quantum,motta2020qite_qlanczos,takeshita2020subspace,stair2020krylov,urbanek2020chemistry,kyriienko2020quantum,YeterAydeniz2020practical,cohn2021filterdiagonalization,yoshioka2021virtualsubspace,epperly2021subspacediagonalization,seki2021powermethod,bespalova2021hamiltonian,baker2021lanczos,baker2021block,cortes2022krylov,klymko2022realtime,tkachenko2024davidson,lee2023sampling,zhang2023measurementefficient,kirby2023exactefficient,shen2023realtimekrylov,yang2023dualgse,ohkura2023leveraging,motta2023subspace,getelina2024quantum,Filip2024variationalphase,anderson2024solving,kirby2024analysis,yoshioka2024diagonalization,byrne2024super,sugisaki2025hamiltonian,yu2025quantum}.
What these algorithms all have in common is that they are based on using a \emph{Krylov space} to approximate extremal eigenvectors.
A Krylov space is a subspace of a larger vector space, constructed by selecting some initial vector, repeatedly multiplying by an operator for some number of iterations, then treating the resulting vectors as a nonorthogonal basis for a subspace.
Lanczos' method and all of its descendants~\cite{lanczos1950iteration,saad2011numerical} derive their power from the remarkable fact that for a Hermitian matrix $M$, the Krylov space generated by $M$ and starting from any vector, contains approximations to the eigenvectors of $M$ with highest and lowest eigenvalues that also overlap with the initial vector (provided those eigenvalues are gapped).
These approximations converge exponentially quickly with respect to the dimension of the Krylov space~\cite{saad1980lanczos}.

Parrish and McMahon first suggested in 2019 that a quantum computer could be used to construct a Krylov space and calculate the projection of a Hamiltonian into it, for the purpose of exploiting the same fast convergence of the ground state energy~\cite{parrish2019quantum}.
In their version, the key difference compared to classical Krylov methods was the proposal to use repetitions of a time-evolution operator rather than the Hamiltonian itself to generate the Krylov space.
Subsequent works have proposed a variety of Krylov spaces that can be generated using a quantum computer, but the Krylov space generated by time-evolutions remains a natural approach since time-evolutions are a relatively cheap and well-studied type of operator to implement.
Furthermore, in 2021 Epperly \emph{et al.} proved that this Krylov space yields convergence at a comparable rate to that generated by powers of the Hamiltonian, even in the presence of noise~\cite{epperly2021subspacediagonalization}.
This algorithm has been used experimentally to diagonalize systems of up to 56 qubits~\cite{yoshioka2024diagonalization}.

The remaining development to reach the SKQD algorithm was to introduce a sampling component~\cite{yu2025quantum}.
While non-sample-based quantum Krylov algorithms require the matrices representing projection of the Hamiltonian into the Krylov space to be estimated statistically using the quantum computer, if one instead samples from the Krylov basis vectors, one can calculate the projection onto the resulting configuration basis classically, as in SQD.
This has the advantage that the resulting matrix is exact to machine precision, yielding the guarantee of variationality that we mentioned in the main text.
It also eliminates the need for a Hadamard test or equivalent, which is required to calculate the matrix elements explicitly on a quantum computer (see e.g.~\cite{cortes2022krylov,yoshioka2024diagonalization}), simplifying the circuits to uncontrolled time evolutions without auxiliary qubits.
The justification for this modification is that one can prove that under the same conditions in which quantum Krylov diagonalization converges efficiently, SKQD does as well, provided the ground state is additionally guaranteed to be approximately sparse.
A vector $|\Psi_0\rangle$ is \emph{sparse} with \emph{sparsity} $s$ if it contains at most $s=\text{poly(n)}$ nonzero entries for $n$ qubits, so that it can be represented efficiently in a classical sparse vector format that stores only the indices and values of the nonzero entries.

To make this precise, the SKQD algorithm is as follows:
\begin{algorithm}[H]
\caption{Sample-based Krylov quantum diagonalization}\label{alg:SKQD}
\begin{algorithmic}[1]
    \Require $n$-qubit Hamiltonian $H$ \Comment{Hamiltonian}
    \Require $U\approx e^{-iH\,\Delta t}$ \Comment{Circuit approximating timestep}
    \Require $|x_0\rangle$ \Comment{Initial guiding state}
    \Require $d$ \Comment{Krylov dimension}
    \Require $M$ \Comment{Samples per basis state}
    \State $B \gets \emptyset$ \Comment{configuration basis (classical)}
    \For {$k = 0,1,2,...,d-1$}
        \For {$l = 0,1,2,...,M-1$}
            \State $b\gets$ sample from $|\phi_k\rangle \coloneqq U^k|x_0\rangle$ \Comment{Quantum}
            \State Add $b$ to $B$ \Comment{Classical}
        \EndFor
    \EndFor
    \State $H_B \gets$ projection of $H$ onto subspace $B$ \Comment{Classical}
    \State $\tilde{E}_0 \gets$ lowest eigenvalue of $H_B$ \Comment{Classical}
    \State\textbf{Return:} $\tilde{E}_0$
\end{algorithmic}
\end{algorithm}
 
SKQD can efficiently approximate the ground state of $H$ provided three conditions are satisfied: we codify these as a problem class that we called the \emph{guided sparse ground state problem}.
\begin{problem}
\label{prob:guided_sparse_ground_state_problem}
    A \emph{guided sparse ground state problem} is defined by a Hamiltonian $H$ and an initial guiding state $|x_0\rangle$, such that...
    \begin{enumerate}
        \item $H$ is gapped;
        \item the ground state $|\Psi_0\rangle$ of $H$ is sparse;
        \item $|x_0\rangle$ has $1/\text{poly}(n)$ overlap with $|\Psi_0\rangle$.
    \end{enumerate}
\end{problem}
In fact, as noted above, SKQD still converges even if the ground state is only approximately sparse in a sense made precise in~\cite{yu2025quantum}.
However, in this paper we will be concerned only with strictly sparse ground states.
The approximation error depends on the initial state overlap ${\gamma_0\coloneqq\langle \Psi_0|x_0\rangle}$, the sparsity of the ground state $|\Psi_0\rangle$, the spectral gap and range, and the number of samples~\cite{yu2025quantum}.
We will refer to the configurations with nonzero amplitudes in our ground states as the \emph{support configurations}.

\section{Classical approaches}
\label{sec:classical_approaches}

\subsection{Sparse Iterative Solvers}
\label{sec:sis_intro}

Sparse iterative solvers aim to discover a subset of configurations $B$ that yields an accurate approximation of the ground state in the form
\begin{equation}
    |\tilde{\psi}_0\rangle = \sum_{x \in B} c_x | x\rangle,
\end{equation}
where $x\in\{0, 1\}^n$. Given $B$, the coefficients $c_x$ that minimize the energy of $|\psi_0\rangle$ can be obtained in closed form as the coefficients of the lowest energy eigenvector of the projection of $H$ into the subspace defined by $B$:
\begin{equation}
\begin{split}
    H_B  |\tilde{\psi}_0\rangle & = \tilde{E}_0  |\tilde{\psi}_0\rangle,\\
    H_B & = \Pi_B H \Pi_B, \\
    \Pi_B & = \sum_{x\in B} |x\rangle \langle x|.
\end{split}
\end{equation}

Starting from an initial set of configurations $B^0$, at each iteration new configurations connected by non-zero Hamiltonian matrix elements are added to the pool of considered configurations. The size of the configuration pool is kept under control by selecting only those that are expected to have the most relevant contributions to the ground state according to some heuristic \textit{selection criterion}.

We break down selection criteria into two broad classes.
The first class relies on performing a subspace diagonalization at each iteration, and is commonly referred to as Selected Configuration Interaction (SCI) algorithms in the electronic structure literature~\cite{holmes2016heat_SCI, zhang2025TrimCI_SCI, holmes2016efficient_SCI, tubman2016deterministic_SCI, sharma2017semistochastic_SCI, Huron_1973_CIPSI}.
A typical selection criterion in this class involves ranking or thresholding the new proposed configurations according to a cost function of the form given in \eqref{eq:perturbative_cf}, which we reproduce here for convenience:
\begin{equation}
\label{eq:perturbative_cf_app}
    f(x) = \frac{\langle x|H|\psi_\mathrm{prev}\rangle}{\langle x|H|x\rangle-E_\mathrm{prev}},
\end{equation}
where $x$ is the new proposed configuration, and $|\psi_\mathrm{prev}\rangle,E_\mathrm{prev}$ are the approximate ground state and energy from the previous iteration.
Eq.~\eqref{eq:perturbative_cf_app} is the amplitude of $|x\rangle$ predicted by first-order perturbation theory.
The SCI heuristics that we consider in this category are CIPSI~\cite{Huron_1973_CIPSI}, HCI~\cite{holmes2016heat_SCI, sharma2017semistochastic_SCI, holmes2016efficient_SCI}, ASCI~\cite{tubman2016deterministic_SCI}, and TrimCI~\cite{zhang2025TrimCI_SCI}.
Detailed descriptions and pseudocode are given in \Cref{ssapp:diagonalization_based}. Within the SCI methods, there are two broad categories: those that select subspaces based on a selection function threshold (like CIPSI and HCI), and those which put a cap on the number of configurations accepted at each iteration (like ASCI and TrimCI). The latter require an additional truncation of configurations from the diagonalization pool before the Hamiltonian application for the production of a new set of configurations. This distinction is important because both CIPSI and HCI grow the diagonalization subspaces as the iterations progress, while for ASCI and TrimCI, the diagonalization subspace size is fixed.

The second class of sparse iterative solvers does not rely on performing subspace diagonalizations at each iteration to evaluate the selection criterion.
We consider three solvers in this class, two of which were specifically designed to target the Hamiltonian construction we present in \Cref{sec:construction}.
The first, which was presented in~\cite{yuan2013truncated}, is a \emph{truncated power method} (TPM).
The approach is to iterate from an initial reference vector by alternating multiplication by $H$ and truncation to the $k$ most significant entries, for some fixed cutoff $k$; see \Cref{app:tpower} for details.
The energy estimate is given as the expectation value of the final vector of the iteration.
We can however extend this to our sparse iterative solver framework by instead taking the support of the final vector (which has size at most $k$) as our basis $B$ for projection and diagonalization: since the final vector is included in the span of this $B$, the error of this modified method will be no larger than the error of the original expectation value method.

The truncated power method is important because it can be proven to converge for a polynomially-gapped, sparse ground state, given a sparse initial reference state with polynomial overlap, with runtime polynomial in all relevant parameters~\cite{yuan2013truncated}.
Under these assumptions, the corresponding sparse eigenvalue problem lies in the complexity class P, which implies that one cannot expect more than a polynomial quantum speedup for this class of problems.
For completeness, we discuss details of TPM including its full convergence proof and runtime analysis in its own \Cref{app:tpower}.
In practice for the finite-sized instances we consider, we find that the TPM does not perform as well as the other classical heuristics we study, including the two that we describe next; see \Cref{app:tpower} for the numerical comparison.

Of the two solvers we designed in this work, the first, which we call \emph{diagonal ranking}, follows the same iteration as SCI, but ranks configurations simply according to their expected energy.
While one would not expect this heuristic to be generally applicable, it can be effective for our construction.
The second, which we call \emph{truncated Arnoldi's method}, is a sparse, matrix-free implementation of Arnoldi's method with truncation of the basis vectors to their heaviest configurations in order to preserve sparsity.
This method may be more broadly applicable.
We give detailed descriptions and pseudocode for both in \Cref{ssapp:non_diagonalization_based}.

\subsubsection{Diagonalization-based sparse iterative solvers}
\label{ssapp:diagonalization_based}

The algorithms for  traditional diagonalization-based sparse iterative solvers can be abstracted as shown in~\Cref{alg:SCI}. The set of inputs to the algorithm are:
\begin{enumerate}
    \item A configuration selection function $f$, whose inputs include the configuration $x$, a set of configurations $\mathcal{S}$, a wave function $|\psi\rangle$, the Hamiltonian $H$, and a selection threshold $\varepsilon$ and/or the maximum number of configurations to be selected $M$. The output is boolean: $1$ if $x$ is selected and $0$ if it is not selected.
    \item A selection threshold $\varepsilon$ and/or the maximum number $M$ of configurations to be selected at each iteration.
    \item The number $C$ of configurations to be kept after the diagonalization at each iteration. $C$ is commonly referred to as the \textit{core} size in the SCI literature. It is common to allow $C$ to be unbounded.
    \item The maximum number of iterations $T$.
    \item An initial set $B^0$ of configurations believed to be contained in the ground state.
\end{enumerate}
Each iteration $\mu$ of the algorithm is decomposed into five steps:
\begin{enumerate}
    \item The approximate ground state in the subspace spanned by $B^\mu$ is obtained by solving the extremal eigenvalue problem in the subspace.
    \item Sort $B^\mu$ such that the corresponding wave function amplitudes are in decreasing order.
    \item Trim $B^\mu$ by keeping the $C$ configurations in $B^\mu$ with largest magnitude in $|\psi^\mu\rangle$, yielding the set of configurations labeled by $\C^\mu$.
    \item Set the coefficients corresponding to configurations outside of $\C^\mu$ to zero in $|\psi^\mu\rangle$. 
    \item Find the set of configurations $\A^{\mu}$ connected by a non-zero Hamiltonian matrix element to any configuration in $\C^\mu$.
    \item Select the set of configurations in $\mathcal{A}^\mu \cup \mathcal{C}^\mu$ that will be used for diagonalization in the next iteration. The configuration selection function $f$ is used.
\end{enumerate}

\begin{algorithm}[H]
\caption{Diagonalization-based iterative sparse solver}\label{alg:SCI}
\begin{algorithmic}[1]
    \Require $f: x, \mathcal{S}, |\psi\rangle, H, \varepsilon, D \mapsto \{0, 1\}$ \Comment{configuration selection function}
    \Require $\varepsilon$, or $D$, or ($\varepsilon$, $D$) \Comment{Selection threshold}
    \Require $1\leq C \leq \infty$ \Comment{Size of set of core configurations}
    \Require $T$ \Comment{Number of iterations}
    \Require $B^0 = \{ x_1, ..., x_{N_0}\}$ \Comment{Initial set of configurations}
    \State $\mu \gets 0$ \Comment{Iteration index}
    \While {$\mu < T$}
    \State $E_0^{\mu}, |\psi^\mu\rangle = \sum_{x\in B^\mu} c^{\mu}_x | x\rangle \gets$ lowest eigenpair of $H_{B^{\mu}}$ \Comment{(1)}
    \State $\sigma \gets$ permutation s.t. $|c^\mu_{x_{\sigma(1)}}|\ge\hdots \ge |c^\mu_{x_{\sigma(|B^\mu|)}}|$ \Comment{(2)}
    \State $\mathcal{C}^\mu = \{x_{\sigma(1)}, \hdots, x_{\sigma(C)}\; | \; x_i \in B^\mu\}$ \Comment{(3)}
    \State $0 \gets c^\mu_x\;\;\forall \; x \notin \mathcal{C}^\mu$ \Comment{(4)}
    \State $\mathcal{A}^\mu = \{x_a \notin \mathcal{C}^\mu \; | \; \langle x_a | H | x_i \rangle \neq 0 \textrm{ for some } x_i \in \mathcal{C}^\mu\}$ \Comment{(5)}
    \State $B^{\mu + 1} = \{x_b \in (\mathcal{A}^\mu \cup \mathcal{C}^\mu) \; | \; f(x_b, |\psi^\mu\rangle, H, \varepsilon, D) = 1 \}$ \Comment{(6)}
    \State $\mu \gets \mu + 1$
\EndWhile
\end{algorithmic}
\end{algorithm}

Various configuration selection functions $f$ are commonly used in the electronic-structure literature, giving rise to different SCI heuristics. In this study we consider the following:
\begin{itemize}
    \item ``Configuration interaction by Perturbatively Selecting Iteratively'' (CIPSI)~\cite{Huron_1973_CIPSI}. This heuristic retains all configurations in $B^\mu$ and ranks configurations in $\A^\mu$ according to the perturbative correction to the $|\psi^\mu\rangle$ estimation of the ground state:
    \begin{equation}
    \label{cipsi_cf}
        f(x_b, |\psi^\mu\rangle, H, \varepsilon) = 
        \begin{cases}
            1 & \textrm{ if } x_b \in B^\mu \\
            1 & \textrm{ if } \left|\frac{\langle x_b | H | \psi^\mu\rangle}{E_b - E_0^\mu} \right|> \varepsilon \\
            0 & \textrm{ otherwise},
        \end{cases}
    \end{equation}
    with $E_b = \langle x_b | H | x_b\rangle$. It is also worth noting that the the core set $\C^{\mu}$ of configurations is typically allowed to grow without restriction, i.e., $\C^\mu=B^\mu$. While the algorithm is strictly memory-bound due to the limit on maximum number of iterations, the memory may grow at each iteration. The only hyperparameter controlling the performance of the method is $\varepsilon$. Not bounding $\C^\mu$ yields monotonically-decreasing energies iteration-by-iteration. 
    
    \item ``Heat Bath SCI'' (HCI)~\cite{holmes2016heat_SCI, sharma2017semistochastic_SCI, holmes2016efficient_SCI}. Similar to CIPSI, this heuristic retains all configurations in $B^\mu$ and ranks configurations in $\A^\mu$ according to a proxy for the perturbative correction to the $|\psi^\mu\rangle$ estimation of the ground state that is less computationally intensive to estimate:
    \begin{equation}
        f(x_b, |\psi^\mu\rangle, H, \varepsilon) = 
        \begin{cases}
            1 & \textrm{ if } x_b \in B^\mu \\
            1 & \textrm{ if } \max_{x_i \in B^\mu}\left|\langle x_b| H |x_i \rangle c_{x_i}^\mu \right|> \varepsilon \\
            0 & \textrm{ otherwise}.
        \end{cases}
    \end{equation}
    Like in CIPSI, the size of the core set $\C^\mu$ is typically allowed to grow without restriction, and the same comments on memory, hyperparameter, and monotonic decrease of energy apply.
    
    \item ``Adaptive Sampling Configuration Interaction'' (ASCI)~\cite{tubman2016deterministic_SCI}. ASCI may be regarded as a iteration-by-iteration memory-bound modification of CIPSI. The ASCI selection function, however, treats on equal footing the approximate wave function amplitudes $c^\mu_x$ and the perturbative estimate of the amplitudes for the configurations $x_a$ in $\A^\mu$:
    \begin{equation}
    \label{asci_cf}
        \tilde{c}^\mu_{x_a} = \frac{\langle x_a | H | \psi^\mu\rangle}{\langle x_a | H | x_a\rangle-E^\mu_0}.
    \end{equation}
    A combined set of amplitudes is formed:
    \begin{equation}
    \begin{split}
        \P^\mu &= \{d^\mu_{x_b} \: | \: x_b \in (\C^\mu\cup \A^\mu) \} \\
        &= \{|c^\mu_x| \: | \: x \in \C^\mu \} \cup \{ |\tilde{c}^\mu_{x_a}| \: | \: x_a \in \A^\mu \},
    \end{split}
    \end{equation}
    and the permutation $\sigma$ that sorts the elements in $\P^\mu$ is obtained:
    \begin{equation}
        d^\mu_{x_\sigma(1)} > \hdots > d^\mu_{x_{\sigma(|\P^\mu|))}}
    \end{equation}
    The selection function is then defined as
    \begin{equation}
        f(x_b, |\psi^\mu\rangle, H, D) = 
        \begin{cases}
            1 & \textrm{ if } b = \sigma(i) \; \forall \; i< D \\
            0 & \textrm{ otherwise}.
        \end{cases}
    \end{equation}
    Since this generates a set $B^\mu$ with at most $D$ elements in it, ASCI is an iteration-by-iteration memory-bound algorithm, unlike CIPSI and HCI. In typical ASCI implementations, the size $C$ of the core set is fixed to a fraction of $D$. Thus ASCI requires the tuning of two hyperparamters, $C$ and $D$. Since $C$ is chosen to be smaller than $D$, and since configurations in $\C^\mu$ and $\A^\mu$ are ranked on equal footing, this algorithm is not guaranteed to produce a monotonically-decreasing estimate of the ground-state energy iteration-by-iteration. 

    \item ``TrimCI''~\cite{zhang2025TrimCI_SCI}. This is a newly introduced heuristic whose selection function involves performing diagonalizations in randomized subsets of $(\C^\mu\cup \A^\mu)$. While the cost of evaluation of the selection function increases as compared to its predecessors (CISI, HCI and ASCI), Ref.~\cite{zhang2025TrimCI_SCI} reports a more efficient exploration of the space of electronic configurations of chemical systems as compared to HCI.

    The selection criterion comprises two phases. In the first phase, configurations $x_a \in \A^\mu$ are filtered according to a selection rule similar to that of CIPSI or HCI, yielding a set of configurations $\widetilde{\A}^\mu$:
    \begin{equation}
        \widetilde{\A}^\mu = \left\{x_a \in \A^\mu \; | \; \left|\frac{\langle x_b | H | \psi^\mu\rangle}{E_b - E_0^\mu} \right|> \varepsilon \right\},
    \end{equation}
    or, alternatively:
    \begin{equation}
        \widetilde{\A}^\mu = \{x_a \in \A^\mu \; | \;  \max_{x_i \in C^\mu}\left|\langle x_b| H |x_i \rangle c_{x_i}^\mu \right|> \varepsilon \}.
    \end{equation}
    While Ref.~\cite{zhang2025TrimCI_SCI} employs the latter, we obtained marginally better results for the Hamiltonians considered in this work with the former, which was therefore used for the numerical results shown later on. Ref.~\cite{zhang2025TrimCI_SCI} also suggests a dynamic scheduling of $\varepsilon$ such that $|\C^\mu| + |\widetilde{\A}^\mu| = F|\C^\mu| $, where $F>1$ is a user-defined hyperparameter.

    The configurations in $\C^\mu$ are combined with those in $\widetilde{\A}^\mu$, forming the set $\P^\mu = \C^\mu \cup \widetilde{\A}^\mu$. $\P^\mu$ is then randomly partitioned into $N_S$ disjoint, uniformly-sized subsets that we label $\P^\mu_j$ for $j = 1, \hdots, N_S$. $N_S$ is a user-defined hyperparameter. Then, the lowest energy states within the subspaces spanned by the $\P^\mu_j$ are computed, along with permutations $\sigma^j$ that sort the amplitudes of these states in decreasing order of magnitude:
    \begin{equation}
    \begin{split}
        |\psi(\P^\mu_j)\rangle & = \sum_{x \in \P^\mu_j} c^j_x|x\rangle. \\
        |c^j_{x_{\sigma^j(1)}}| & \ge \hdots \ge |c^j_{x_{\sigma^j(|\P^\mu_j|)}}|
        \end{split}
    \end{equation}
    From each subset $\P^\mu_j$ we keep the $N_K$ configurations with largest amplitude, and pool all of those to form the set $B^{\mu+1}$. $N_K$ is another user-defined hyperparameter.
    
    The size of the diagonalization set is $|B^{\mu + 1}| = N_S \cdot N_K$, making this an iteration-by-iteration memory-bound method. Similarly to ASCI, it is typical to set the size $C$ of the core set to be a fraction of $N_S \cdot N_K$. Accordingly, TrimCI requires the tuning of four hyperparameters, $F$ or $\varepsilon$, $C$, $N_S$, and $N_K$. When $C$ and $F$ are explicitly set, the size of the subsets for selection diagonalization can be obtained as $|\P^\mu_j| = F \cdot C/N_S$. This algorithm is not guaranteed to produce a monotonically-decreasing estimate of the ground-state energy iteration-by-iteration. 
\end{itemize}

\subsubsection{Non-diagonalization-based sparse iterative solvers}
\label{ssapp:non_diagonalization_based}

\begin{itemize}
    \item Diagonal ranking (configuration energy-based selection): we introduce a heuristic iterative method to construct a pool of configurations $B$ for subspace diagonalization. No diagonalization is required in the iteration, as shown in~\Cref{alg:diag ranking}. The algorithm ranks configurations according to their expected energy. At each iteration $\mu$, two sets of configurations are kept. The first is the so-called working set of configurations $B^\mu$, which is the set used for generation of new configurations by Hamiltonian application, and that will eventually be used for subspace diagonalization. The second set is the so-called reservoir ser $\mathcal{R}^\mu$, which is used to keep track of a larger sit of configurations and their corresponding energies, without the need to explicitly perform operations on the larger set. The algorithm requires the specification of the maximum size of the working and reservoir sets,  $D$ and $R$ respectively. We choose $R>>D$. Additionally, an inital working set of configurations is needed to start the iteration. Each iteration is decomposed into six steps (see~\Cref{alg:diag ranking}):
    \begin{enumerate}
        \item Find the set configurations $\A^{\mu}$ not in $\C^\mu$ connected by a non-zero Hamiltonian matrix element to any configuration in $\C^\mu$.
        \item The reservoir set is updated by adding the configurations in $\A^{\mu}$.
        \item Computation of the energies of the configurations in the updated reservoir.
        \item Construction of the permutation that sorts the configuration energies in ascending order.
        \item Update the working set $B^{\mu+1}$ by keeping the $D$ configurations with lowest energy in the updated reservoir set $\mathcal{R}^{\mu + 1}$. 
        \item The reservoir is trimmed by keeping only the $R$ configurations with lowest energy in $\mathcal{R}^{\mu + 1}$. This last step is needed to avoid the growth of the memory requirement of the algorithm.
    \end{enumerate}
    Once the iterations are completed, the Hamiltonian is projected and diagonalized in the subspace spanned by the latest working set $B^T$.
    
    Fixing $D$ guarantees that the number of operations has a fixed ceiling. Likewise, fixing $R$ guarantees that the memory required has a fixed ceiling. The diagonal ranking algorithm requires the tuning of two hyperparameters. This algorithm is not guaranteed to produce a monotonically-decreasing estimation of the ground-state energy.

    \begin{algorithm}[H]
    \caption{Diagonal ranking}\label{alg:diag ranking}
    \begin{algorithmic}[1]
        \Require $D$ \Comment{Cutoff on number of working configurations}
        \Require $R>D$ \Comment{Cutoff on number of reservoir configurations}
        \Require $T$ \Comment{Number of iterations}
        \Require $B^0 = \{ x_1, ..., x_{N_0}\}$ \Comment{Initial set of working configurations}
        \State $\mathcal{R}^0 = B^0$ \Comment{Initialize reservoir}
        \State $\mu \gets 0$ \Comment{Iteration index}
        \While {$\mu < T$}
            \State  $\mathcal{A}^\mu = \{x_a \notin \mathcal{C}^\mu \; | \; \langle x_a | H | x_i \rangle \neq 0 \textrm{ for some } x_i \in B^\mu\}$ \Comment{(1)}
            \State $\mathcal{R}^{\mu+1} = \mathcal{R}^\mu \cup \A^\mu$ \Comment{(2)}
            \State $\E_{\mathcal{R}^{\mu+1}} = \{E_{x_b} = \langle x_b|H|x_b\rangle  \: | \: x_b \in \mathcal{R}^{\mu+1} \}$  \Comment{(3)}
            \State $\sigma \gets$ permutation s.t. $E_{x_{\sigma(1)}}\le\hdots \le E_{x_{\sigma(|\mathcal{R}^{\mu+1}|)}}$ \Comment{(4)}
            \State $B^{\mu + 1} = \{x_{\sigma(1)}, ...,x_{\sigma(D))} \; | \; x_b \in \mathcal{R}^{\mu+1}  \}$ \Comment{(5)}
            \State $\mathcal{R}^{\mu+1} \gets  \{x_{\sigma(1)}, ...,x_{\sigma(R))} \; | \; x_b \in \mathcal{R}^{\mu+1}  \}$ \Comment{(6)}
            \State $\mu \gets \mu+1$ 
        \EndWhile
    \State \textbf{Return} $B^{\mu + 1}$
    \end{algorithmic}
    \end{algorithm}
    
    \item Truncated Arnoldi Method: Lanczos method~\cite{lanczos1950iteration} and its generalization to non-Hermitian matrices, Arnoldi's method, have been standard tools in classical exact diagonalization for decades~\cite{saad2011numerical}.
For example, the sparse matrix eigensolvers in \textsc{scipy}~\cite{virtanen2020scipy} run ARPACK solvers under the hood, which in turn employ variants of Lanczos and Arnoldi's methods~\footnote{https://docs.scipy.org/doc/scipy/reference/generated /scipy.sparse.linalg.eigs.html}.
The problem with applying these algorithms to large quantum systems is that, even in matrix-free sparse vector implementations, the sparsity grows exponentially with the number of iterations, limiting the convergence that can be reached before running out of memory.

In order to construct a classical adversary within this family of methods that is capable of handling large quantum systems with exponential Hilbert space dimension, we propose a new variant of Arnoldi's method that we call \emph{truncated Arnoldi's method}.
The basic idea is to use a matrix-free, sparse vector implementation, but additionally truncate the vector to its most significant entries at each iteration, thus curtailing the exponential blow-up of memory.
Pseudocode for the iteration is given in \Cref{alg:tr_arn_iter}.
To our knowledge, this exact algorithm has not been previously proposed, although very similar ideas have been~\cite{moro1981calculation,kuprov2008polynomially,ma2013sparse}.

\begin{algorithm}[H]
\caption{Truncated Arnoldi iteration}\label{alg:tr_arn_iter}
\begin{algorithmic}[1]
    \Require $M$ \Comment{Cutoff on number of new configurations}
    \Require $T$ \Comment{Number of iterations}
    \Require $v_0$ \Comment{Initial sparse vector}
    \For {$i=0,1,...,T-1$}
        \State $u_{i+1}\gets Hv_i$ \Comment{Expand}
        \For {$j=0,1,...,i$}
            \State $u_{i+1}\gets u_{i+1} - (v_j, u_{i+1})v_j$ \Comment{Orthogonalize}
        \EndFor
        \State $u_{i+1}\gets u_{i+1}|_M$ \Comment{Truncate}
        \State $v_{i+1}\gets u_{i+1}/\|u_{i+1}\|$ \Comment{Normalize}
    \EndFor
\end{algorithmic}
\end{algorithm}

The notation in the truncation step, $u_{i+1}|_M$, means setting all but the largest $M$ amplitudes in $u_{i+1}$ to zero.
The crucial point in the construction is that orthogonalization is carried out before truncation: this ensures that the truncation preserves as much of the new direction that $u_{i+1}$ contributes to the Krylov space as possible.
Since Krylov bases are known to be exponentially ill-conditioned~\cite{beckerman2017singularvalues}, if truncation is carried out before orthogonalization, then after a small number of iterations the new truncated vectors will be contained within the existing Krylov space.

Truncating after orthogonalization has two downsides.
First, while we would normally use Lanczos method for Hermitian $H$, the truncation means that the standard Lanczos trick of only orthogonalizing with respect to the previous two basis vectors is no longer sufficient, and instead we have to orthogonalize with respect to all previous basis vectors as we would in Arnoldi's method (hence ``truncated Arnoldi'' rather than ``truncated Lanczos'').
Second, even so, the resulting basis vectors are no longer exactly orthogonal (although the condition number of the basis no longer grows rapidly), since the truncation happens after orthogonalization.
Therefore, projecting and diagonalizing within this Krylov space means forming and then solving a generalized eigenvalue problem that includes a Gram matrix for the basis; this subroutine is given in \Cref{alg:tr_arn_pd}.

\begin{algorithm}[H]
\caption{Truncated Arnoldi projection and diagonalization (not configuration basis)}\label{alg:tr_arn_pd}
\begin{algorithmic}[1]
    \Require $V\coloneqq[v_0,v_1,...,v_T]$ \Comment{Matrix of basis vectors from \Cref{alg:tr_arn_iter}}
    \State $\tilde{H}\gets V^\dagger H V$ \Comment{Projection matrix}
    \State $\tilde{S}\gets V^\dagger V$ \Comment{Gram matrix}
    \State Solve $\tilde{H}w=\lambda\tilde{S}w$ \Comment{Diagonalize}
\end{algorithmic}
\end{algorithm}

However, we can do much better than this.
Since we are explicitly maintaining the Krylov basis vectors as sparse vectors, projecting onto them is no more costly than projecting onto the union of their supports.
This union of supports forms a configuration basis for a subspace of the Hilbert space, just like in SCI, but no diagonalization is required during the iteration, only at the end.
This algorithm is what we use as a classical adversary, and we give it explicitly in \Cref{alg:tr_arn_pd_bts}.
When we refer to truncated Arnoldi's method throughout the remainder of the paper, we will be refering to this variant.

\begin{algorithm}[H]
\caption{Truncated Arnoldi projection and diagonalization (configuration basis)}\label{alg:tr_arn_pd_bts}
\begin{algorithmic}[1]
    \Require $\{v_0,v_1,...,v_T\}$ \Comment{Basis vectors from \Cref{alg:tr_arn_iter}}
    \State $B\gets\bigcup_{i=0}^T\mathrm{supp}(v_i)$ \Comment{Union of supports}
    \State $\tilde{H}\gets\Pi_BH\Pi_B$ \Comment{Project ($\Pi_B=$ projector onto $B$}
    \State Diagonalize $\tilde{H}$.
\end{algorithmic}
\end{algorithm}

Truncated Arnoldi's method is a heuristic, since lowest eigenvalue of the resulting projected matrix no longer has the guaranteed convergence of Arnoldi's method, but in practice it converges efficiently in many cases.
In particular if, as in the construction in this manuscript, the ground state is sparse and a good initial vector is given, truncated Arnoldi is well-adapted for success.
The only way it can fail is if the truncation at intermediate steps cuts off paths in configuration space that would combine at a later iteration to provide amplitude to one of the configurations in the support of the ground state.
Therefore, in order to cause it to fail, we would need to engineer this situation.

\end{itemize}

\subsection{DMRG}
\label{sapp:dmrg}

DMRG~\cite{white1992dmrg} is a classical tensor-network method that, rather than simulating quantum circuits, targets ground states directly.
It was originally developed for condensed matter physics~\cite{white1992dmrg} but has more recently found great success in strongly-correlated chemistry~\cite{chan2002dmrg}.

More precisely, DMRG is a low-entanglement wave-function approximation~\cite{white1992dmrg}. It targets wavefunctions of matrix product state (MPS) type, which are defined as
\begin{equation}
| \Psi_{\mathrm{MPS}} \rangle = \sum_{x_1 \dots x_n}  \left[ A[1]^{x_1} A[2]^{x_2} \dots A[n]^{x_n} \right] | x_1 \dots x_n \rangle
\;,
\end{equation} 
where $M$ denotes the number of qubits in the system and each $A[i]$ is a $\chi\times\chi$ matrix of real numbers, except for the boundary terms $A[1]$ and $A[i]$, which are length-$\chi$ vectors.
$\chi$ is called the bond dimension of $\Psi_{\mathrm{MPS}}$ and controls the accuracy and memory cost of the simulation: as $\chi$ increases the wave function may converge to the exact correlated ground state.
Given an MPS state $\Psi_{\mathrm{MPS}}$, the energy $E=\langle \Psi_{\mathrm{MPS}} | H | \Psi_{\mathrm{MPS}} \rangle$ can be stably computed by representing the Hamiltonian $H$ as a matrix-product operator and performing a contraction operation~\cite{schollwock2011density}.

The energy can be variationally optimized through the DMRG sweep algorithm~\cite{white1993density}, where in the individual tensors $A[i]$ are update one at a time while sweeping back and forth along a one-dimensional ordering of sites.
The update requires solving an effective eigenvalue problem for $A[i]$, then truncating the solution to keep only the most important states in a singular-value decomposition of such tensor.
Sweeping repeatedly lowers the energy and converges $\Psi_{\mathrm{MPS}}$ toward a stationary point of the energy within the chosen bond dimension.

In practice, the sweep algorithm employs a user-defined schedule in which the bond dimension is gradually increased toward a maximum $\chi$ over successive sweeps, and decreasing amounts of noise are added to the tensors $A[i]$ to discourage convergence to metastable states.
As the calculation progresses, the noise is reduced and the bond dimension increased, allowing the algorithm to refine $\Psi_{\mathrm{MPS}}$ and converge to the variational optimum more reliably. Nevertheless, convergence to a global minimum of the energy is not guaranteed \emph{a priori} and may be challenging in practice.

In principle, MPS and other tensor-network states need not be sparse wavefunctions. However, a linear combination of $K$ basis states $x_1 \dots x_n$ can be represented by an MPS of bond dimension $K$.
Conversely, given an MPS state, one can identify the basis states over which $\Psi_{\mathrm{MPS}}$ is approximately supported by repeatedly sampling the distribution $p(x_1 \dots x_n) = | A[1]^{x_1} A[2]^{x_2} \dots A[n]^{x_n} |^2$, for example with the ``direct sampling'' method~\cite{ferris2012perfect}.
While direct sampling is formally efficient -- and preferable to Markov-chain Monte Carlo approaches as it does not require the choice/design of transition moves and is free from the overhead of equilibration and decorrelation -- its practical usefulness depends on the amount of time required to draw a sample and the number of required samples. For a uniform distribution over $K$ basis states, the average number of required samples is $M = O(K \log K)$ by the coupon collector theorem, whereas for a non-uniform distribution the average number of samples is determined by the tails of the probability distribution.

Unlike sparse iterative solvers, DMRG does not rely on sparsity of the target state.
However, a $K$-sparse state in the computational basis can be represented exactly by an MPS of bond dimension $K$ regardless of the geometry, meaning that DMRG does always have the expressivity to solve a sparse ground state problem.
Its limitations are determined by the entanglement of the target states as well as the spectral structure of the Hamiltonian as a whole, which together determine the rate of convergence.
Therefore, achieving hardness for DMRG for a sparse ground state problem will likely require making the optimization challenging, setting the sparsity large enough that the algorithm is impractical even if the bond dimension is formally polynomial, or some combination of the two.

\subsection{Direct tensor network simulation with belief propagation} \label{sapp:tensornetwork}

If the assumptions that guarantee SKQD convergence are satisfied, one classical strategy is to replace the quantum sampling subroutine in \Cref{alg:SKQD} with a classical simulation of it, thereby replacing the hybrid quantum-classical workflow of SKQD with an entirely classical one. Here, we are focused on problem sizes out-of-reach for exact statevector simulation, and thus turn our attention to approximate methods. One of the most successful approximate methods for simulating quantum circuits relies on tensor networks, which provide an efficient representation for states with low entanglement. 

Due to their conceptual simplicity and the fact that they can be efficiently contracted, matrix product states (MPS) have historically served as the primary method for tensor network-based classical simulations of quantum circuits~\cite{Banuls_SimulationManyqubit_2006, Cirac_MatrixProduct_2021, Zhou_WhatLimits_2020, Napp_EfficientClassical_2022}. However, they are most naturally suited to one-dimensional geometries, and are inherently limited in their ability to efficiently represent even area-law entangled states on higher-dimensional graphs~\cite{Stoudenmire_StudyingTwoDimensional_2012, kim2023evidence, King_BeyondclassicalComputation_2025}, such as the two-dimensional heavy-hex geometries considered here.

In recent years, this limitation has motivated a shift toward more general tensor network representations that faithfully mirror the geometry of the underlying qubit connectivity graph. A key challenge for such techniques is the computational cost of exact contraction, which is in general \#P-hard ~\cite{Schuch_ComputationalComplexity_2007}. Recent progress has addressed this challenge through the use of approximate contraction algorithms based on message-passing and belief propagation (BP) techniques ~\cite{Alkabetz_TensorNetworks_2021, Tindall_GaugingTensor_2023, Wang_TensorNetwork_2024}. In this approach, the tensor network is interpreted as a graphical model, and local ``messages'' are iteratively exchanged along the network edges to approximate marginal distributions or effective environments. While this comes with no guarantee of convergence on loopy graphs, empirical results have demonstrated that it can provide accurate and scalable approximations for a wide class of quantum circuits of practical interest. Recent successes include the simulation of kicked-Ising circuits on heavy-hex geometries~\cite{Tindall_EfficientTensor_2024}, dynamics of spin glasses in two- and three-dimensional lattices~\cite{Tindall_DynamicsDisordered_2025}, and sample-based diagonalization routines similar to the method studied in this work~\cite{rudolphSimulatingSamplingQuantum2025}.

In the present context, we are interested in leveraging tensor networks to classically simulate and sample from the quantum circuits underlying SKQD. For a qubit connectivity graph with worst-case coordination number $z$ (i.e., for $z=3$ for heavy-hex and $z=4$ for square lattices), the dominant computational costs are as follows~\cite{rudolphSimulatingSamplingQuantum2025}:
\begin{itemize}
    \item Following the application of a two-qubit gate, the corresponding bond is truncated to a user-specified maximum bond dimension $\chi$. This truncation is performed via singular value decomposition (SVD), incurring a cost of $O(\chi^{z+1})$ per two-qubit gate.
    \item Message tensors are updated using BP with a time complexity of $O(n \chi^{z+1})$ per update, where $n$ is the number of qubits. As discussed in Ref.~\cite{Tindall_EfficientTensor_2024}, more frequent updates improves the optimality of SVD truncation at the expense of increased computational cost. Updating messages after each layer of non-overlapping gates provides a practical compromise.
    \item Sampling: using the boundary MPS method of Ref.~\cite{rudolphSimulatingSamplingQuantum2025}, $M$ samples can be taken with time complexity $O(n \chi^{z+1} R^3) + O(M n \chi^{z+1}R^3)$ for heavy-hex ($z=3$) and square ($z=4$) lattices, where $R$ is the rank of the boundary MPS.
\end{itemize}

In addition to the above temporal costs, the memory requirements for a tensor network scale as $O(n \chi^z)$.

\section{Hamiltonian construction}
\label{sec:construction}

\begin{figure*}[t]
    \centering
    \includegraphics[width=0.6\linewidth]{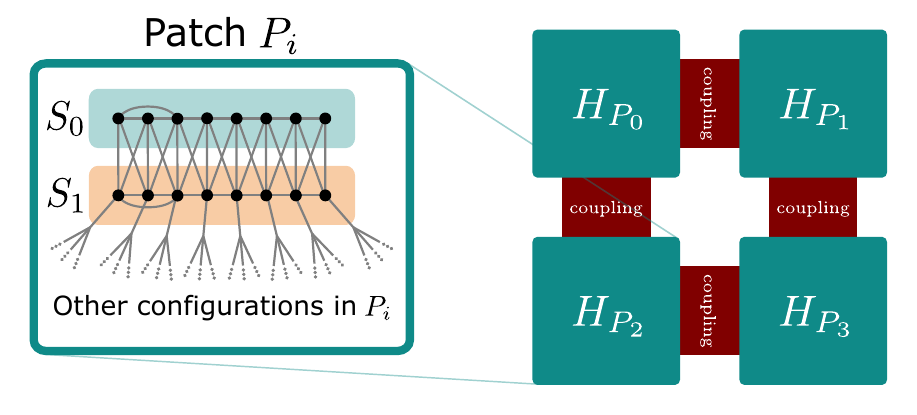}
    \caption{(duplicate of \Cref{fig:hamiltonian_construction_schematic} in the main text) Sketch of the general Hamiltonian construction. The qubit layout is partitioned into disjoint patches. Within each patch $P_i$, the set $S_0$ of support configurations is coupled to a set $S_1$ of first-order configurations, which are then coupled to the remaining configurations of the qubits in the patch. This defines the patch Hamiltonian $H_{P_i}$. The patches are themselves coupled by $H_\mathrm{coupling}$.}
    \label{fig:hamiltonian_construction_schematic_app}
\end{figure*}

The problem we would like to tackle using SKQD and the classical solvers discussed in \Cref{sec:classical_approaches} is the guided sparse ground state problem defined in Problem~\ref{prob:guided_sparse_ground_state_problem}.
In this section, we describe a general set of steps for constructing guided sparse ground state problems, while also making them challenging for classical selected configuration interaction.
We then give two specific examples of their realization.

\subsection{General construction}
\label{sec:general_construction}

A sketch of the steps detailed in this section is given in \Cref{fig:hamiltonian_construction_schematic_app}, which is repeated from the main text for ease of reference.
Our general framework for constructing Hamiltonians is as follows.
\begin{algorithm}[H]
\caption{Hamiltonian construction schematic}\label{alg:gen_const}
\begin{algorithmic}[1]
    \Require $n, G$ \Comment{number of qubits, qubit layout}
    \Require $n_\mathrm{patch}$ \Comment{number of patches}
    \State $\{P_i\}_{i=0,1,...,n_\mathrm{patch}-1} \gets$ partition $G$ into equal size patches (leftover ``padding'' qubits are permitted if $n/n_\mathrm{patch}\notin\mathbb{Z}$)
    \State $H_P \gets$ local Hamiltonian defined on each $P_i$
    \State $|\psi_0\rangle \gets$ nondegenerate lowest eigenvector of $H_P$
    \State $H \gets \sum_{i=0}^{n_\mathrm{patch}-1}H_{P_i}$ \Comment{combine patches}
    \State $|\psi_0\rangle^{\otimes n_\mathrm{patch}}$ is nondegenerate lowest eigenvector of $H$
    \State Define positive semidefinite $H_\mathrm{coupling}$ on $G$, such that $H_\mathrm{coupling}|\psi_0\rangle^{\otimes n_\mathrm{patch}}=0$
    \State $H \gets H + H_\mathrm{coupling}$ \Comment{couple patches}
\end{algorithmic}
\end{algorithm}

In Step 4, the notation $H_{P_i}$ denotes the operator $H_P$ acting on the specific patch $P_i$.
The first component of our Hamiltonian is thus a sum of several local ``patch'' Hamiltonians defined on a partition of the qubits.
Its ground state is therefore the tensor product of the ground states of the local patches.
We then add coupling between the patches that preserves the ground state, so that excited states are interacting between the patches, but the ground state remains a product state.
Thus, if the patch ground state $|\psi_0\rangle$ has sparsity $s_\mathrm{patch}$, the sparsity of the global ground state is $(s_\mathrm{patch})^{n_\mathrm{patch}}$.
Furthermore, for the initial reference state, we will choose an element of the support of the patch ground state, and the global initial reference state will be the tensor product of those local elements.
Thus, if the local initial state overlap is $\gamma_\mathrm{patch}$, then the global initial state overlap is $(\gamma_\mathrm{patch})^{n_\mathrm{patch}}$.
The sparsity and initial state overlap are two of the three conditions required by the guided sparse ground state problem (Problem~\ref{prob:guided_sparse_ground_state_problem}).
The third, the spectral gap, will be controlled by both the patch Hamiltonians and the coupling.

Both the numbers and sizes of the patches are parameters that we will discuss further in the examples in \Cref{v36} and \Cref{v57}, but for example in the actual instance we use for our experimental demonstrations, $n_\mathrm{patch}=3$ and each patch contains $\lfloor n/3\rfloor$ of the $n$ qubits.
We next describe how to construct the patch Hamiltonians, then finally describe how to couple them.

\begin{algorithm}[H]
\caption{Patch Hamiltonian construction}\label{alg:patch_const}
\begin{algorithmic}[1]
    \Require $P_i$ \Comment{patch layout}
    \State $S_0 \gets$ set of configurations \Comment{local support of ground state}
    \State $x_0 \gets [S_0]_0$ \Comment{initial configuration (guiding state)}
    \State $\Pi_{S} \gets$ projector onto arbitrary set $S$ of configurations
    \State Choose $H_{S_0}$ \Comment{Hamiltonian restricted to $S_0$}
    \State Fix local terms in $H_P$ such that $\Pi_{S_0}H_P\Pi_{S_0}=H_{S_0}$
    \State $|\psi_0\rangle \gets$ lowest eigenvector of $H_{S_0}$
    \State $S_1 \gets$ set of configurations, $S_1\cap S_0=\emptyset$ \Comment{configurations coupled to $S_0$ at first-order}
    \State Choose $H_{S_0\rightarrow S_1}$ such that $H_{S_0\rightarrow S_1}|\psi_0\rangle=0$ \Comment{transition matrix}
    \State Add local terms to $H_P$ such that $\Pi_{S_1}H_P\Pi_{S_0}=H_{S_0\rightarrow S_1}$
    \State Add local terms acting on $S_1$ such that $|\psi_0\rangle$ is ground state of $H_P$
\end{algorithmic}
\end{algorithm}

A number of these steps require more detailed comment.
First, one of the configurations in $S_0$ (the set of support configurations) is designated as $x_0$, the initial reference state.
Without loss of generality, we fix it to be the zeroth element of $S_0$.
Next, we want to implement Steps 1-5 such that the local terms in $H_P$ up to this point do not couple $S_0$ to any additional configurations, so that $H_{S_0}$ is a block in $H_P$.

\begin{figure*}[t]
    \centering
    \includegraphics[width=\linewidth]{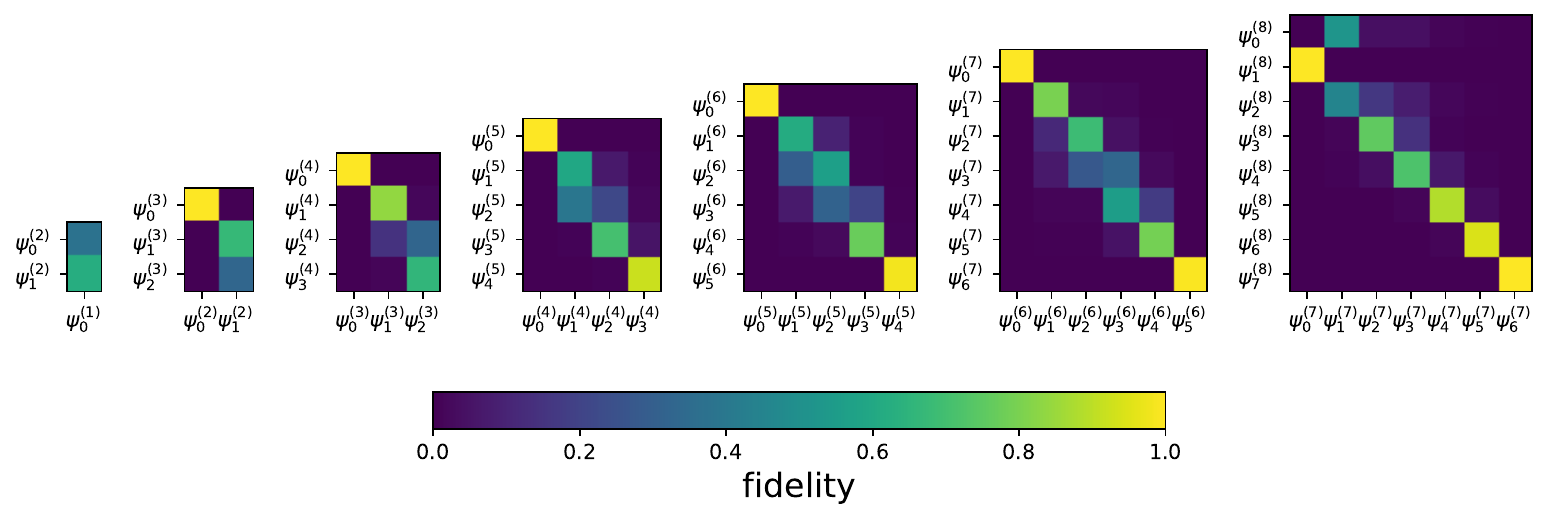}\vspace{-0.1in}\\
    ~\vspace{0.1in}(a)\\
    \includegraphics[width=0.4\linewidth]{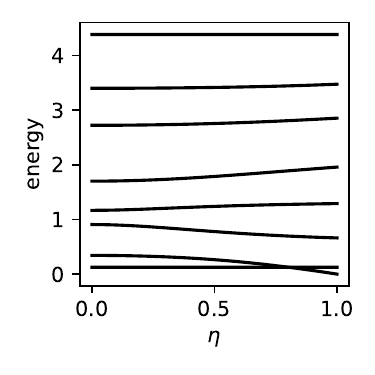}\\
    (b)
    \caption{(a) The sequence of fidelities between all eigenvectors of the $i$th and $(i+1)$th principal submatrices of $H_{S_0}$ \eqref{v36_base}, for each $i=1,2,...,7$. The $j$th eigenvector of the $i$th principal submatrix is labeled $\psi_j^{(i)}$, with $j$ in order of increasing energy. This illustrates that from dimensions 2 to 7, the ground eigenvector remains the same, while from dimension $7$ to $8$ it becomes the first excited state due to the level crossing, which appears visually as displacement of the main nonzero fidelities off of the diagonal in the upper-left of the largest heatmap. The final ground state (whose fidelities are given by the first row in the rightmost heatmap) therefore has no overlap with the previous ground state; it has only $\sim62\%$ fidelity with the entire previous basis (of which $\sim52\%$ is with previous first excited state) so it is significantly different in character from any state that precedes it in iteration. b) The level crossing that is responsible for this new character, made explicit by continuously turning on the final off-diagonal row and column in $H_{S_0}$ via a parameter $\eta$ (i.e., replacing the $1$s in entries $(6,7)$ and $(7,6)$ in \eqref{v36_base} with $\eta$). Each curve is the evolution of one energy eigenvalue with respect to $\eta$. The level crossing between the lowest two energies is apparent around $\eta=0.8$.}
    \label{fig:phase_transition}
\end{figure*}

In Step 5, we choose $H_{S_0}$, which plays the key role of determining how challenging it is for iterative solvers to traverse the space of support configurations.
In order to cause methods based on first-order perturbation theory around the initial configuration (such as CIPSI) to struggle, we want to construct $H_{S_0}$ such that its lowest energy state has a qualitatively different character than that of any of its principal submatrices.

For example, if $|S_0|=8$, the matrix 
\begin{equation}
\label{v36_base}
    H_{S_0}=
    \begin{pmatrix}
        a+\frac{1}{2}&1&b&0&0&0&0&0 \\
        1&a&c&0&0&0&0&0 \\
        b&c&a+2&1&0&0&0&0 \\
        0&0&1&a+2&1&0&0&0 \\
        0&0&0&1&a+1&1&0&0 \\
        0&0&0&0&1&a+1&1&0 \\
        0&0&0&0&0&1&a+1&1 \\
        0&0&0&0&0&0&1&a
    \end{pmatrix}
\end{equation}
with
\begin{equation}
\begin{split}
    &a = 0.90694271, \\
    &b = -0.78820544, \\
    &c = -0.61541221,
\end{split}
\end{equation}
is a good choice.
These parameters were selected in order to yield a particular set of properties that we now discuss.
The ground state energy of $H_{S_0}$ is nondegenerate, and its ground state is supported on all $8$ configurations; $a$ is a diagonal shift that is selected to set the ground state energy to zero for convenience.
The squared overlap of the ground state with the $|0\rangle$ configuration is $3.24\times10^{-4}$.
Recall that $|0\rangle$ will be the initial configuration that forms the guiding state provided in the problem statement.
The size of the spectral range is $4.38$, and the spectral gap is $0.126$ (which will lead to a gap for the global ground state as well once we glue the patches together, as required by Problem~\ref{prob:guided_sparse_ground_state_problem}, Property 1).

Beyond these general properties, $H_{S_0}$ is banded, which means that any iterative solver has to progress through the configurations in order, with the exception of the allowed transition from $|0\rangle$ to $|2\rangle$.
Any iterative solver therefore requires at least $6$ steps to find all $8$ configurations.
The energy penalty for only finding the first $7$ out of $8$ configurations is $0.126$.

More importantly, let $|\psi^{(i)}_0\rangle$ denote the lowest energy state of the $i\times i$ leading principal submatrix of $H_{S_0}$.
We refer to $|\psi^{(i)}_0\rangle$ as a \emph{partial ground state}.
For ${i=2,3,...,7}$, $|\psi^{(i)}_0\rangle$ is supported \emph{only} on the first two configurations: for example
\begin{equation}
    |\psi^{(7)}_0\rangle=\begin{pmatrix}-c\\b\\0\\0\\0\\0\\0\end{pmatrix}.
\end{equation}
The full ground state of $H_{S_0}$, on the other hand, is supported on all $8$ support configurations, as noted above:
\begin{equation}
    |\psi_0\rangle \coloneqq |\psi^{(8)}_0\rangle
    \approx \begin{pmatrix}
        -0.018 \\
        -0.014 \\
        -0.049 \\
        0.119 \\
        -0.298 \\
        0.449 \\
        -0.559 \\
        0.616
    \end{pmatrix}.
\end{equation}
This presents a challenge for SCI heuristics based on first-order perturbation theory: consider for example the numerator of the cost function \eqref{eq:perturbative_cf_app}, which is the magnitude of the transition matrix element between the previous partial ground state and the new proposed configuration.
Since configurations $|3\rangle$ through $|7\rangle$ (indexing from $|0\rangle$) only have nonzero matrix elements with their neighbors, and all partial ground states from $3$ to $7$ end in amplitude $0$, the numerator of \eqref{eq:perturbative_cf_app} is
\begin{equation}
    \langle i| H_P |\psi^{(i)}_0\rangle = 0
\end{equation}
for any $i=3,4,...,7$.
This means that, even though that configuration will be proposed when the pool is expanded at the beginning of the iteration, it will then be discarded in the truncation step of the iteration, for any positive value of the threshold.
Once all $8$ configurations have been found, a level crossing occurs, permitting the full ground state to be supported on all $8$ configurations; this is illustrated in \Cref{fig:phase_transition}.

Having established our choice of $H_{S_0}$, the remaining step that requires comment is 9, in which the only requirement on the local terms we add to enforce $\Pi_{S_1}H_P\Pi_{S_0}=H_{S_0\rightarrow S_1}$ is that they should not couple $S_0$ to any other configurations beyond those in $S_1$.
Let $\lambda_0$ be the lowest eigenvalue of $H_{S_0}$.
After Step 9, we have
\begin{equation}
\begin{split}
    H_{S_0\cup S_1}
    \coloneqq\,&(\Pi_{S_0}+\Pi_{S_1})H_P(\Pi_{S_0}+\Pi_{S_1}) \\
    =\,&
    \begin{bmatrix}
        H_{S_0} & H_{S_0\rightarrow S_1}^\dagger \\
        H_{S_0\rightarrow S_1} & H_{S_1}
    \end{bmatrix},
\end{split}
\end{equation}
where $H_{S_1}$ is some Hermitian matrix determined by the action of the terms we have already chosen within $S_1$.
By construction,
\begin{equation}
    H_{S_0\cup S_1}\begin{bmatrix}
        |\psi_0\rangle \\
        \bf{0}
    \end{bmatrix}
    =
    \begin{bmatrix}
        H_{S_0}|\psi_0\rangle+ 0 \\
        H_{S_0\rightarrow S_1}|\psi_0\rangle + 0
    \end{bmatrix}
    =
    \lambda_0\begin{bmatrix}
        |\psi_0\rangle \\
        \bf{0}
    \end{bmatrix},
\end{equation}
i.e. $\begin{bmatrix}|\psi_0\rangle\\\bf{0}\end{bmatrix}$ is an eigenvector of $H_{S_0\cup S_1}$ with eigenvalue $\lambda_0$.

Furthermore, because we required that the local terms we added up to this point only couple $S_0$ to $S_0\cup S_1$, the above means that $\begin{bmatrix}|\psi_0\rangle\\\bf{0}\end{bmatrix}$ is in fact an eigenvector of $H_P$ with eigenvalue $\lambda_0$.
$S_1$ may indeed be coupled to additional configurations outside of $S_0\cup S_1$ by the existing local terms, but $\begin{bmatrix}|\psi_0\rangle\\\bf{0}\end{bmatrix}$ is only supported on $S_0$, so the submatrix $H_{S_0\cup S_1}$ entirely captures the action of $H_P$ on it.
We then adjust the coupling strengths defined by $H_{S_0\rightarrow S_1}$, and add terms that act only on $S_1$ if necessary, to ensure that $\begin{bmatrix}|\psi_0\rangle\\\bf{0}\end{bmatrix}$ is not only an eigenstate, but the ground state of $H_P$.
Examples of how to accomplish this are given in the following subsections.
This completes the general construction of the patch Hamiltonians.

In Step 4 of \Cref{alg:gen_const}, we build up a larger Hamiltonian as the sum of the local patch Hamiltonians.
Provided there are no additional qubits not included in the patches, this guarantees that
\begin{equation}
\label{eq:global_gs}
    |\Psi_0\rangle \coloneqq |\psi_0\rangle^{\otimes n_\mathrm{patch}}
\end{equation}
is the nondegenerate global ground state.
However, we do not want to stop here, since this would be a noninteracting Hamiltonian that can easily be solved by separately solving uncoupled local patches.
To remedy this, we define a positive semidefinite (p.s.d.) local coupling Hamiltonian $H_\mathrm{coupling}$ such that $H_\mathrm{coupling}|\Psi_0\rangle=0$.
This guarantees that $|\Psi_0\rangle$ remains the global ground state, since
\begin{equation}
    (H+H_\mathrm{coupling})|\Psi_0\rangle = H|\Psi_0\rangle,
\end{equation}
and the requirement that $H_\mathrm{coupling}$ be p.s.d. guarantees that it cannot shift any other eigenenergies below that of $|\Psi_0\rangle$.  

Our other goal in the construction of $H_\mathrm{coupling}$ is to cause the number of nonzero entries in a general sparse vector to grow exponentially with respect to the number of multiplications by $H$.
We will see examples of how to accomplish this below.

One may ask why, since by definition $H_\mathrm{coupling}$ has $|\Psi_0\rangle$ as a ground state, we do not simply use $H_\mathrm{coupling}$ as our full Hamiltonian.
The reason is that in practice, $H_\mathrm{coupling}$ may be highly degenerate, and as we will see in the specific instantiations below, it is easier to satisfy our other desired properties given this flexibility.
We now turn to two instantiations of this construction.

\subsection{Warmup example}
\label{v36}

We begin by giving the construction for $H_P$ as in \Cref{alg:patch_const}.
We let $S_0$ be all eight configurations on three qubits.
This enables us to choose $H_{S_0}$ to be the matrix defined in \eqref{v36_base}, which we can realize by decomposing it directly into (up to) three-qubit Pauli operators.

We then let $S_1$ be the same eight configurations on three qubits, now tensored into $|1\rangle$ on one additional qubit, which we call the \emph{coupling qubit} and label $c_i$ for the $i$th patch.
We choose $H_{S_0\rightarrow S_1}=m_1H_{S_0}$ for $|m_1|\le1$.
Since $|\psi_0\rangle$ is the lowest eigenvector of $H_{S_0}$, and the lowest eigenvalue of $H_{S_0}$ is $0$, $H_{S_0\rightarrow S_1}|\psi_0\rangle=0$ as desired.
We realize $H_{S_0\rightarrow S_1}$ as
\begin{equation}
    \begin{bmatrix}
        0 & H_{S_0\rightarrow S_1}^\dagger \\
        H_{S_0\rightarrow S_1} & 0
    \end{bmatrix}
    =
    \begin{bmatrix}
        0 & m_1H_{S_0} \\
        m_1H_{S_0} & 0
    \end{bmatrix}
    =
    m_1X\otimes H_{S_0},
\end{equation}
where the blocks in the block matrices correspond to the states of the coupling qubit.
Up to this point, $S_0\cup S_1$ is not yet coupled to any additional configurations.

So far, we have
\begin{equation}
    H_P
    =
    \begin{bmatrix}
        H_{S_0} & m_1H_{S_0} \\
        m_1H_{S_0} & H_{S_0}
    \end{bmatrix}
    =(I+m_1X)\otimes H_{S_0}.
\end{equation}
Since $H_{S_0}$ is p.s.d. by construction, and $I+m_1X$ is also p.s.d. as long as $|m_1|\le1$, their tensor product is p.s.d.
Therefore, $\begin{bmatrix}|\psi_0\rangle\\\bf{0}\end{bmatrix}$ is already a ground state since its eigenvalue is $0$.
However, it is degenerate with $\begin{bmatrix}\bf{0}\\|\psi_0\rangle\end{bmatrix}$.
We can fix this by adding $m_2(I-Z)$ (for some $m_2>0$) acting on the coupling qubit to break the degeneracy, i.e., choosing
\begin{equation}
    H_P
    =(I+X)\otimes H_{S_0} + m_2(I-Z)\otimes I^{\otimes 3}.
\end{equation}
This completes the patch Hamiltonian.
The patch size is four qubits, and each contributes a factor of $|S_0|=8$ to the sparsity.

To build up a large Hamiltonian as in \Cref{alg:gen_const}, we fix a lattice of qubits, then choose some local subsets of size $4$ to serve as our patches.
The ground state energy will still be zero, and the ground state will be the tensor product of the local ground states of the patches.
Unfortunately, there is now a degeneracy, since there are no terms acting on all of the remaining qubits so any of their states form a degenerate ground state when tensored into the ground states of the patches.

In order to break this degeneracy, we introduce interactions on all edges in the lattice except except those touching the patches.
We additionally introduce the same interactions on the edges from the remainder of the lattice to the coupling qubits, so that the patches are coupled to the full Hilbert space.
In other words, if we let $E$ denote the edges in the qubit layout $G$,
\begin{equation}
    H_\mathrm{coupling} = \sum_{\substack{e\in E\\\forall i,~e\cap P_i=\emptyset}}[H_\mathrm{int}]_e + \sum_{\substack{e\in E\\\exists i,~e\cap P_i=\{c_i\}}}[H_\mathrm{int}]_e,
\end{equation}
where the notation $[H_\mathrm{int}]_e$ indicates applying $H_\mathrm{int}$ on the pair of qubits $e$.
The two-qubit interaction terms we add are
\begin{equation}
    H_\mathrm{int}
    \coloneqq
    \begin{pmatrix}
        0 & 0 & 0 & 0 \\
        0 & 2.02 & 1+i & 1+i \\
        0 & 1-i & 2.02 & 1+i \\
        0 & 1-i & 1-i & 2.02
    \end{pmatrix},
\end{equation}
which is p.s.d. with eigenvalues $0, 0.02, 1.29, 4.75$.
$|00\rangle$ is the unique lowest eigenvector, so by applying $H_\text{int}$ to every edge not in the patch as well as the edges to the coupling qubit, we ensure that the ground state everywhere other than the patches is $|00...0\rangle$.
Thus, the full, nondegenerate ground state is the tensor product of ground states of the patches, and $|0\rangle$ on all other qubits.
Note that this actually makes the degeneracy-breaking term above unnecessary, so we can set $m_2=0$.

A variety of classical techniques can solve for the ground state of this Hamiltonian.
Since the ground state is a product of the local ground states of the patches and the patches each comprise four qubits, it is frustration-free with respect to a grouping of the Pauli terms into subsets of size at most four.
Therefore, one can identify the patches by inspection, diagonalize their local Hamiltonians, and note that on all the remaining edges the ground state is $|00\rangle$, thus solving the Hamiltonian essentially by hand if the general construction above is known.
DMRG also solves for the ground state, typically with bond dimension 8 since that is the sparsity of the local patch ground states and these are unentangled from each other.

On the other hand, SCI methods can be induced to fail with appropriate choices of parameters in the construction above.
As discussed above, using the matrix \eqref{v36_base} as the Hamiltonian of the patches causes the CIPSI heuristic to fail since the lowest-energy states obtained from projecting onto subsets of the support configurations are qualitatively different from the ground state once all support configurations are found.
Other SCI heuristics are found to perform still worse in practice.
We do not go into further detail since our main focus will be on the construction we describe next in \Cref{v57}.
The motivation for this second construction is to induce longer-range correlations in the sparse ground state and avoid the local frustration-free nature of the ground state.
This will alleviate the concern that SKQD might succeed for the present example simply due to the frustration-free ground state and very short-range correlations in the entire algorithm, which is not a convincing way to achieve an advantage over iterative solvers.

\subsection{Main example}
\label{v57}

We again begin by giving the construction for $H_P$ as in \Cref{alg:patch_const}.
Let $S_0$ be all eight Hamming-weight-$1$ configurations on $8$ qubits.
Assume that they have linear connectivity, but are separated by $2$ along the line, i.e. we leave every other qubit along the line unused, so that the whole line has length $16$.

Once again, choose $H_{S_0}$ to be the matrix defined in \eqref{v36_base}.
Each offdiagonal matrix element corresponds to a partial swap between next-nearest-neighbors along the line of qubits.
The single transition matrix element outside the tridiagonal (the $(0,2)/(2,0)$ matrix element) corresponds to a partial swap between qubits separated by distance $4$ on the line.
All of these partial swaps can be decomposed into Pauli operators.

Let $S_1$ be the eight Hamming-weight-$1$ configurations on the remaining $8$ qubits in the line, which interleave the qubits defining $S_0$.
Like in \Cref{v36}, choose
\begin{equation}
\label{m1_def}
    H_{S_0\rightarrow S_1}=m_1H_{S_0}
\end{equation}
The same comments on the lowest eigenstate apply.
The transition matrix elements between $S_0$ and $S_1$ can be realized by yet more partial swaps.
The most distant partial swap is between qubit $0$ in $S_0$ (which is also qubit $0$ in the line) and qubit $2$ in $S_1$ (which is qubit $5$ in the line).
Therefore, we have this single length $5$ interaction, a few length $4$ interactions, and many of length $3$ and shorter.

 We can break the degeneracy of $H_{S_0\cup S_1}$ in formally the same way as in \Cref{v36}.
In the present construction, this requires adding
\begin{equation}
\label{m2_def}
    \Pi\coloneqq m_2(I-Z)
\end{equation}
on all of the qubits defining $S_1$.

Our patch size is $16$ qubits.
So to build up a large Hamiltonian, we do not need to pad with many additional qubits, if any.
For example, if we want a Hamiltonian with $48$ qubits, we just need to find an embedding of $3$ length-$16$ lines into our desired physical layout, and this will yield a Hamiltonian whose ground state has sparsity $8^3=512$.
The corresponding ground state will be a superposition of all Hamming-weight-$1$ configurations on the $S_0$ qubits, tensored over all patches, and tensored with $|0\rangle$ on all other qubits.

We want to break degeneracies due to any additional padding qubits, and to couple the patches.
To do this, we add the following interaction on all edges in our physical lattice that start with a qubit defined by one of the copies of $S_1$:
\begin{equation}
\label{j1_def}
    H_\mathrm{int}
    \coloneqq
    j_1\times\begin{pmatrix}
        0 & 0 & 0 & 0 \\
        0 & 1.01 & 0 & 1 \\
        0 & 0 & 0 & 0 \\
        0 & 1 & 0 & 1.01
    \end{pmatrix}.
\end{equation}
In other words, we choose
\begin{equation}
    H_\mathrm{coupling} = \sum_{\substack{e\in E\\\exists i,~|e\cap [S_1]_i|>0}}[H_\mathrm{int}]_e,
\end{equation}
where $[S_1]_i$ denotes the set of qubits defined by $S_1$ in the $i$th patch.
We can think of $H_\mathrm{int}$ as a p.s.d. version of a CNOT: controlled on the state of the qubit in $[S_1]_i$, it flips the other qubit.
Since the qubits in $[S_1]_i$ are all $|0\rangle$ in the support configurations, this interaction does not act on the support configurations, but it couples the configurations in $S_1$ to configurations with higher Hamming weight, as well as coupling those configurations to each other.
This yields the desired exponential growth with respect to number of multiplications by the Hamiltonian.

Finally, to obfuscate the structure of the support configurations, we conjugate the Hamiltonian by a layer of single-qubit $X$ gates, applied with independent $50\%$ probability on each qubit.
This maps the support configurations, which all have Hamming weight $n_\mathrm{patch}$ in the unconjugated picture, to configurations with random Hamming weight in the conjugated picture.

Some comments are in order.
First, the Hamiltonian is constructed entirely of one- and two-qubit terms, acting on qubits separated by constant distance within a physical layout.
Second, while the global ground state is still a product state, it is now a product of $8$-sparse states of $8$ qubits each.
So for example, if three patches are used, in principle if one could identify the 24 qubits corresponding to the support configurations in the patches, then one could diagonalize the induced 24-qubit Hamiltonian by brute force and thus solve the problem.
We do not rule out `smart' approaches of this kind, since as we will see in the specific instantiation given below in \Cref{sec:setup}, they are not the only way to solve this Hamiltonian.
However, the concern that the ground state is transparently frustration-free on small local patches is alleviated.

\section{Overview of experiments and results}
\label{sec:results}

\subsection{Instantiating the guided sparse ground state problem}
\label{sec:setup}

As discussed in the main text, we provide three inputs to our classical and quantum solvers:
\begin{enumerate}
    \item A representation of the Hamiltonian $H$ in the Pauli basis.
    \item A single initial configuration $|x_0\rangle$.
    \item A circuit defining a single Trotter step under $H$.
\end{enumerate}
The first two inputs are the only strictly required components, as defined in Problem~\ref{prob:guided_sparse_ground_state_problem}.
However, using advanced quantum circuit transpilation techniques to obtain low-depth Trotter step circuits does reveal some further structure of $H$ beyond $|x_0\rangle$, e.g., a good mapping to a two-dimensional qubit lattice.
We consider it best to share this component of the experimental optimization as well, since it can potentially be used by classical adversaries.

The Hamiltonian we use is the one described in \Cref{v57}.
The specific parameters are as follows:
We build the Hamiltonian on top of a heavy-hex lattice containing 6 hexes in the configuration shown in \Cref{fig:v57_layout}, which is repeated from the main text for ease of reference.
This contains a total of 49 qubits.
Each patch requires 16 qubits and contributes a factor of 8 to the sparsity of the ground state, so we used 3 patches for a sparsity of 512.
This leaves a single padding qubit.
We choose the parameters
\begin{equation}
    m_1=0.1,\quad m_2=0.01,\quad j_1=1,
\end{equation}
defined in \eqref{m1_def}, \eqref{m2_def}, and \eqref{j1_def}, based on a sweep to identify regions of classical hardness.
The Hamiltonian and the corresponding initial configuration are available at~\footnote{\url{https://github.com/quantum-advantage-tracker/quantum-advantage-tracker.github.io/tree/main/data/variational-problems/hamiltonians/guided_sparse_ground_state_problem/49Q_v1}}.
To give some further sense of it here, in \Cref{fig:pauli_details} we give histograms of the weights (numbers of qubits acted upon) and maximum distances between qubits acted upon (in the layout in \Cref{fig:v57_layout}) over all Pauli terms in the Hamiltonian; it contains 419 Pauli terms in total.

\begin{figure}[ht]
    \centering
    \includegraphics[width=\linewidth]{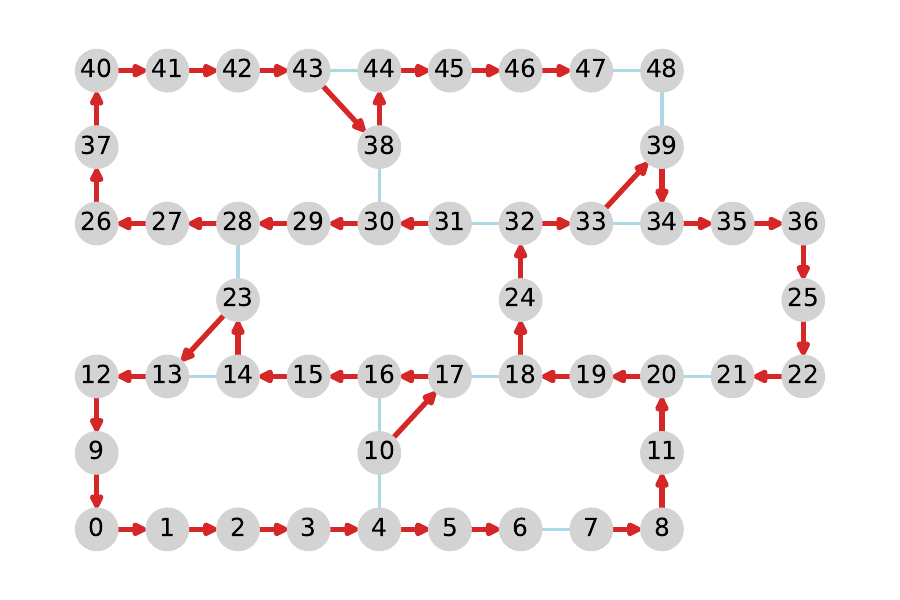}
    \caption{(duplicate of \Cref{fig:v57_layout} in the main text) Layout of the patches defining our Hamiltonian within a heavy-hex graph comprising six hexes. The paths defining each patch are shown by the red arrows. As discussed in \Cref{v57}, patches are defined by paths with even-indexed qubits corresponding to the support configurations: in this figure, the directed edges point in the direction of increasing index. The edges in the physical qubit layout are shown in light blue wherever they do not overlap with the path edges. Qubit $48$ in the upper right corner is a padding qubit, to complete the sixth hex.}
    \label{fig:v57_layout_app}
\end{figure}

\begin{figure}[ht]
    \centering
    \includegraphics[width=\linewidth]{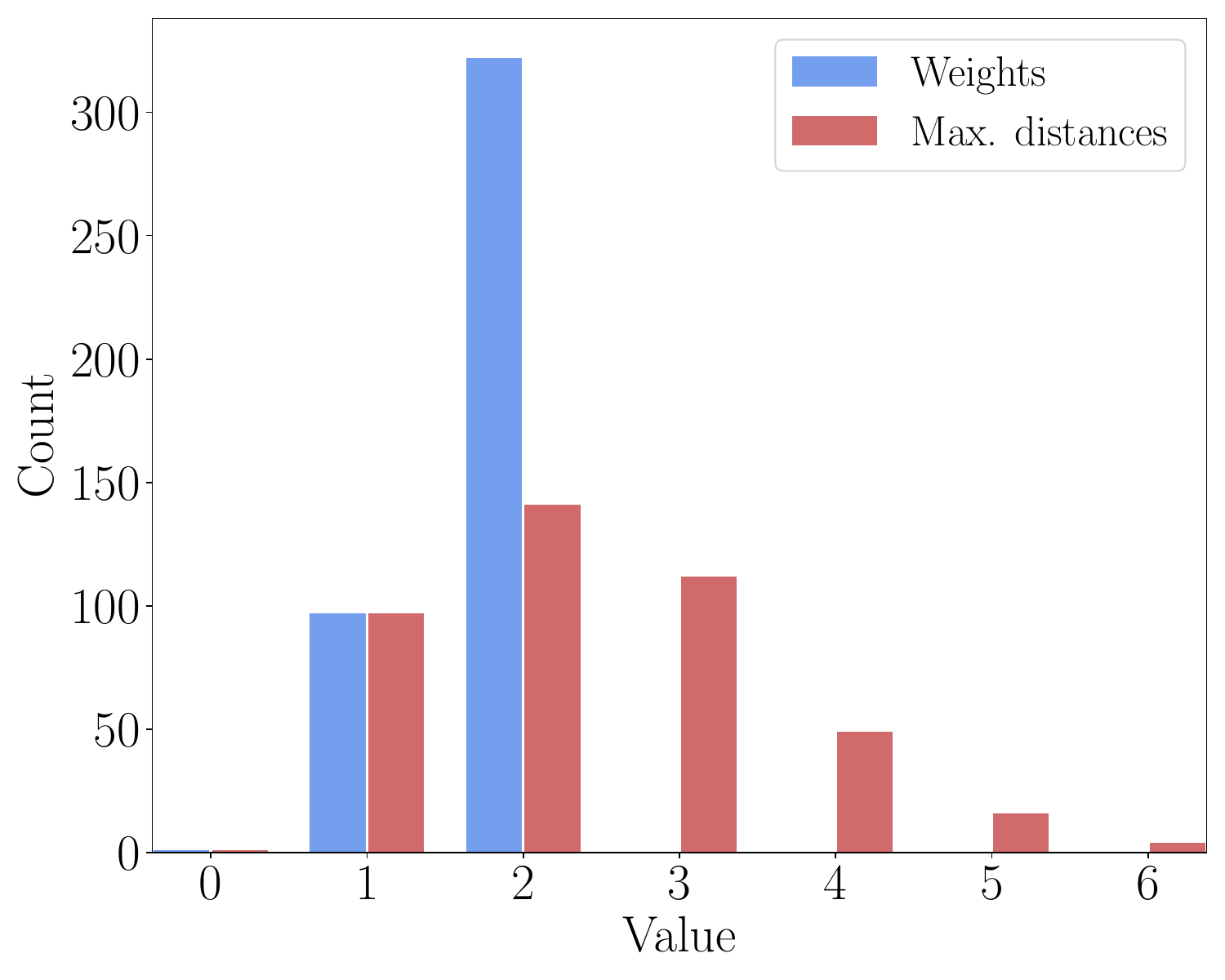}
    \caption{Weights and maximum distances within the qubit layout of Pauli terms in the Hamiltonian. Weight refers to the number of qubits the Pauli term acts upon. Maximum distance refers to the longest length of the minimum path between any pair of qubits the Pauli term acts upon, with respect to the qubit layout shown in \Cref{fig:v57_layout_app}. These distances are given as the number of qubits traversed, including the first and last qubits, so for example a Pauli term acting only on a pair of neighboring qubits would have maximum distance two, while a Pauli acting on three qubits connected along a line would have maximum distance three.}
    \label{fig:pauli_details}
\end{figure}

\subsection{Quantum experiment}
\label{sec:quantum_expt}

We apply the SKQD algorithm as discussed in \Cref{sec:skqd_intro}.
This requires sampling from time-evolution circuits generated the Hamiltonian.
We use two main techniques to approximate and transpile these.

First, we employ Rustiq~\cite{martiel2025rustiq,debrugiere2024faster}, a software package that compiles arbitrary linear combinations of Pauli operators into Trotter steps constructed from local Pauli rotation gates on a specified layout.
See \Cref{app:rustiq_transpilation} for details.
The resulting circuit is a parameterized first-order Trotter step $T_1(\Delta t)$ with a time-independent Clifford correction $C_1$.

To construct a full first-order Trotter step for time $\Delta t$, we implement $T_1(\Delta t)C_1$ for a total CZ count 1443 and depth 111.
However, we can also construct a second-order Trotter step for time $2\Delta t$ as
\begin{equation}
\label{second_order_step}
    T_1(\Delta t)C_1[T_1(-\Delta t)C_1]^\dagger = T_1(\Delta t)[T_1(-\Delta t)]^\dagger.
\end{equation}
The gate costs in terms of CZ count 1568 and depth 113 are closer to a single first-order step than to two of them, since the intermediate Clifford corrections cancel as shown in \eqref{second_order_step}.

Motivated by this observation, we employ the following hybrid Trotterization strategy.
For even numbers of time steps, we use second-order Trotterization with each second-order step being used to implement two timesteps, i.e. we implement
\begin{equation}
    U(2k\Delta t)=\left(T_1(\Delta t)[T_1(-\Delta t)]^\dagger\right)^{k}.
\end{equation}
For odd numbers of timesteps, we use second-order Trotterization for all but the last timestep, as above, followed by a single first-order timestep:
\begin{equation}
    U\big((2k+1)\Delta t\big)=\left(T_1(\Delta t)[T_1(-\Delta t)]^\dagger\right)^{k}T_1(\Delta t)C_1.
\end{equation}

On existing quantum hardware one can execute in practice a few Trotter steps.
From our tensor network simulations of the Trotterized circuits (see ~\Cref{app:tn_sims}), we learned that we would likely require Krylov dimensions (i.e. numbers of timesteps) somewhere in the range of 10-20 in order to sample all of the support configurations within a reasonable total number of shots, and the corresponding circuit depths and gate counts are not currently feasible.
As a way around this, we hybridize our time-evolution constructions a second time by introducing partial approximate quantum compilation.

Approximate quantum compilation (AQC) is a strategy for compressing quantum circuits into lower depth and gate count ans\"atze.
Details are given in \Cref{app:aqc_tensor}, but the results in our case are circuits of fixed CZ depth 24 that approximate timesteps 1 through 20.
These AQC circuits decay exponentially (with respect to number of timesteps) in their fidelity to the ideal states from the Trotter circuits, but they provide a means for reaching higher numbers of timesteps with at least some approximation.

Our full hybrid strategy is to implement each AQC circuit, then append 0, 1, or 2 Trotter steps using the hybrid Trotterization discussed above.
As discussed above, the motivation for this is that it combines two different types of signal decay: from the AQC we have a decay due to coherent error (we are preparing the ``wrong'' state), while from the Trotter approximations we have a decay that is well-modeled by Pauli noise channels, i.e. dominated by incoherent error.
So with this approach we end up with three formally redundant implementations of each number of timesteps: for example, for 16 timesteps we have 16-step AQC, 15-step AQC plus 1 Trotter step, and 14-step AQC plus 2 Trotter steps, which represent different combinations of coherent and incoherent errors.
The total classical computational resource usage for transpilation of the circuits was about 377 core-hours (see \Cref{app:transpilation} for details).

The only hyperparameter for SKQD other than the Krylov dimension is the timestep $\Delta t$.
In the original SKQD paper~\cite{yu2025quantum}, the theoretical analysis suggests that $\Delta t=\pi/\|H\|$ is the correct choice.
However, numerical and experimental experience from both the prior study, and the earlier stages of this project suggest that a significantly larger $\Delta t$ leads to faster convergence in practice.
For our experiment, we chose $\Delta t=25\pi/\|H\|$.
See \Cref{app:tn_sims} for details.

The experiments were run on \texttt{ibm\_boston}, a Heron r3 superconducting quantum processor with 156 qubits laid out in a heavy-hex lattice structure. A $3 \times 2$ heavy-hex layout comprising 49 qubits was selected for our sampling experiment. These 49 qubits were chosen by considering the gate errors as well as the qubit quality. In particular, we compute the estimated fidelity
\begin{equation}
    F = \prod_{g_1 \in \text{1Q gates}} (1-e_{g_2}) \prod_{g \in \text{2Q gates}} (1-p_g) \prod_{q \in \text{qubits}} (1-e_q),
\end{equation}
where $e_{g_1}$,  $e_{g_2}$, and $e_q$ correspond to single-qubit gate errors, two-qubit gate errors, and measurement errors respectively. $e_{g_1}$ and $e_{g_2}$ are estimated using randomized benchmarking (RB), with $e_{g_2}$ computed using infidelities from layer fidelity RB. The estimated fidelities are used to rank and select the best layout.

Due to the unstructured nature of the idle gaps in the circuit, we apply a structure agnostic dynamical decoupling sequence in which the $X_p X_m$ pulse is applied to all the idle spaces in the circuit.
See \Cref{app:experimental_details} for more details about the device.
The shots per Krylov dimension are determined adaptively and iteratively.
In particular, the total evolution for Krylov dimension $k$ corresponds to $k = k_{\textrm{AQC}} + k_{Q}$ where $k_{\textrm{AQC}}$ denotes the number of timesteps approximated by AQC as discussed above, and $k_Q$ denotes the number of Trotter steps that follow.
The first portion of circuit corresponding to $k_{\textrm{AQC}}$ is structured in a heavy-hex native layer comprising 324 two-qubit operations.
Besides its gate angles, this sub-circuit is independent of the value of $k_{\textrm{AQC}}$.
The post-AQC Trotter steps correspond to highly unstructured quantum circuits with odd (even) Trotter steps requiring 1119 and 1244 two-qubit gates respectively.

We began by sampling from the circuits with ${k_{\textrm{AQC}} = 1,...,15}$ and $k_Q = 0,1,2$, as an initial guess; each circuit was run for 1 million shots. Observing that we had not yet reached the exact ground state energy, we then added $k_{\textrm{AQC}} = 16,...,20$ and $k_Q = 0,1,2$, as well as increasing the number of shots (focusing on higher values of $k_{\textrm{AQC}}$ and $k_Q$), in an attempt to improve the energy further.
Ultimately, the circuits for $k_{\textrm{AQC}}=1,...,16$ with $k_Q = 0,1,2$ were adequate to reach the exact energy.
Overall, these circuits concluded with 133 million shots taken over five days (with just 97 shots across the $k_{\textrm{AQC}}=1,...,16$ subset).
The total QPU runtime was around 19 hours.
A summary of allocated shots broken down by circuit is given in \Cref{app:experimental_details}.

We employed a preliminary filter on the configurations, in which we remove any configurations in the pool that are not connected by the Hamiltonian to any other elements of the pool.
Such disconnected configurations would be eigenstates of the projected Hamiltonian, so it is unnecessary to include them in the diagonalization unless they individually have lower energy than the rest of the subspace: we assume that this is not the case.
The majority of configurations fall into this category as it turns out, so this filtering reduces the diagonalization dimension by more than an order of magnitude, leaving a dimension $\sim10.3$ million.
We may interpret this as due to the presence of noise in our measurement outcomes: many of the configurations we sampled were likely affected by at least some bitflip errors.
The Hilbert space dimension of $2^{49}\approx5.6\times10^{14}$ is so much larger than the number of distinct configurations in our original pool that even with just a few random bitflips, configurations become unlikely to be connected by the Hamiltonian to others in the pool.
This filtering therefore acts as a form of sample error-mitigation in addition to reducing the dimension of the diagonalization at no cost in accuracy.

After filtering, we projected the Hamiltonian onto the remaining pool of $\sim10.3$ million configurations, then diagonalized the resulting matrix.
This was performed on a Mac Studio with Apple M3 Ultra chip (32 cores), and required approximately 593 wall-time seconds for projection and diagonalization.
See \Cref{app:diagonalization} for details.

\subsection{Classical approaches}
\label{sec:classical_attacks}

All of the classical methods we tested on the Hamiltonian described in \Cref{sec:setup} are summarized at a high-level in \Cref{sec:classical_approaches}.
The specific strategies and parameters we employed were as follows.

For all of the iterative solvers, we performed fine-grained selection hyperparameter sweeps, with details given in~\Cref{app:classical_simulation_details}. The results reported below in \Cref{ssec:results} correspond to the lowest variational energies achieved by each method with the lowest-possible computational overhead, as measured by the largest diagonalization required (as a proxy). CIPSI, HCI, diagonal ranking, and truncated Arnoldi's method all require tuning one hyperparameter, ASCI requires tuning two hyperparameters, and TrimCi requires tuning four hyperparameters.

As for tensor network methods, four DMRG sweeps with a maximum bond dimension of 200 were enough to reach an error in the ground state energy indistinguishable from machine precision.
We additionally ran tensor network simulations of the quantum circuits, primarily as a guiding and diagnostic tool for the quantum experiments.
See Sec.~\ref{app:tn_sims} for more info.

\subsection{Results}
\label{ssec:results}

\begin{figure*}[ht]
    \centering
    \includegraphics[width=1\linewidth]{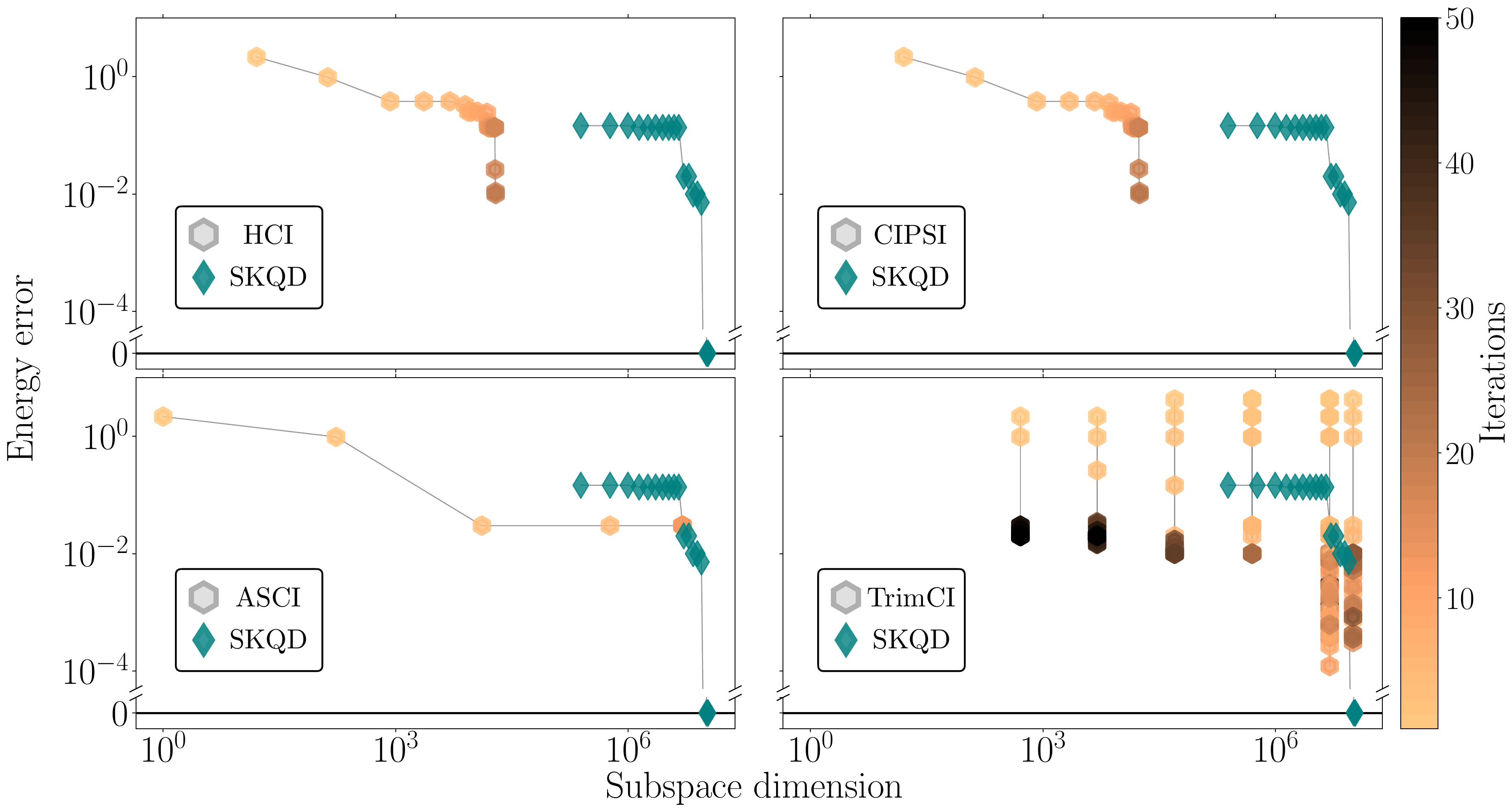}
    \caption{(duplicate of \Cref{fig:results_good} in the main text) \textbf{Improved accuracy of quantum (SKQD) against off-the-shelf classical SCI methods.} Comparison of the ground state energy reached for a given subspace dimension for different SCI methods, and the SKQD experiment executed on quantum hardware. Different SCI methods are shown in different panels by the hexagonal markers. The different colors show the number of iterations. The results obtained with SKQD are shown by blue diamond markers in each panel. The sequence of SQKD points is with respect to Krylov dimension.}
    \label{fig:results_good_app}
\end{figure*}

\begin{figure*}[ht]
    \centering
    \includegraphics[width=1\linewidth]{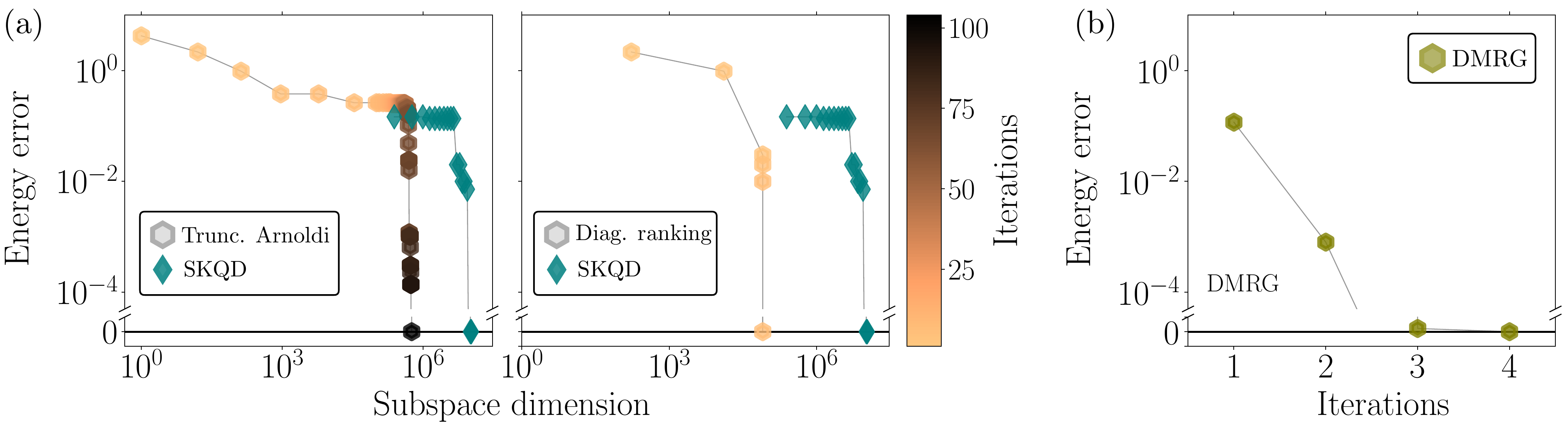}
    \caption{(duplicate of \Cref{fig:results_bad} in the main text) \textbf{Performance of specialized sparse iterative solvers and DMRG.} \textbf{(a)} compares the ground state energy reached for a given subspace dimension for the sparse iterative solvers introduced in this manuscript, as shown by the hexagonal markers. The results obtained with SKQD are shown by blue diamond markers in each panel. \textbf{(b)} Convergence of DMRG as a function of the number of sweeps. }
    \label{fig:results_bad_app}
\end{figure*}

\begin{figure}[ht]
    \centering
    \includegraphics[width=1\linewidth]{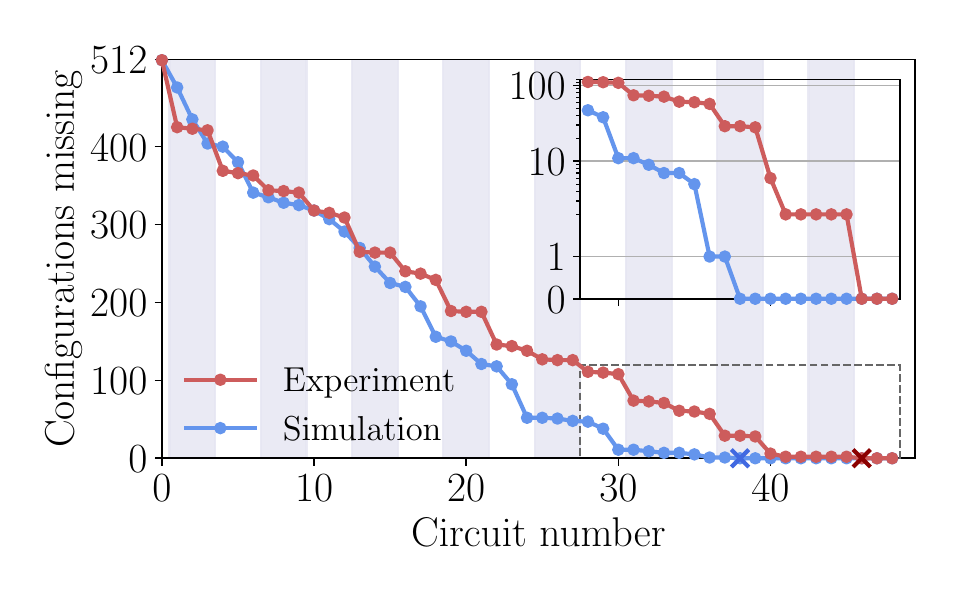}
    \caption{Cumulative count of ``missing'' configurations as a function of circuit number $c$, where $c = 0$ corresponds to $k_{\mathrm{AQC}} = 0$, $k_Q=0$ and $c = 3k_{\mathrm{AQC}} + k_Q - 2$ for $c>0$. Red and blue data show experiment and theoretical expectation based on tensor network simulations, respectively (see Sec.~\ref{app:tn_sims} for details). Vertical bands visually group circuits of constant $k_{\mathrm{AQC}}$, while the inset shows a zoomed-in view of the dotted region. Crosses indicate the point of SKQD convergence, corresponding to observation of all 512 configurations. For both experiment and simulation, we non-uniformly distribute a total of 97 million shots across the various circuits; details are discussed in Sec.~\ref{app:experimental_details}.}
    \label{fig:results_configurations}
\end{figure}

The results of the experiments described in \Cref{sec:quantum_expt,sec:classical_attacks} are shown in \Cref{fig:results_good_app,fig:results_bad_app,fig:results_configurations}.
Specifically, the panels in \Cref{fig:results_good} show the off-the-shelf SCI methods that we applied, compared to the results of SKQD.
All of these SCI methods have hyperparameters, as discussed in \Cref{sec:classical_approaches,sec:classical_attacks}, and in \Cref{fig:results_good_app} we only show the best performing hyperparameter settings: complete results from all parameter settings are given in \Cref{app:classical_simulation_details}.
As the plots show, none of the SCI methods are able to find the exact ground state energy.
SKQD successfully finds it with a Krylov dimension of 17, i.e., the initial configuration plus samples from 1 through 16 timesteps (the quantum data is the same in all panels).

CIPSI and HCI were not able to reach diagonalization dimensions comparable to those required in SKQD, since by design they cannot reach arbitrary dimensions.
Instead, they terminate when no new configurations are added to the pool from one iteration to the next, since once this occurs, no new configurations will be found in any further iterations.
This maximum achievable subspace dimension is partially controlled by the threshold hyperparameter, but decreasing the threshold arbitrarily is not guaranteed to increase the subspace size arbitrarily, since it is possible to reach a pool of configurations such that any further configurations have cost function zero.
As detailed in~\Cref{app:classical_simulation_details}, not even a reduction of the threshold by fourteen orders of magnitude was enough to allow the methods to reach diagonalization dimensions comparable to those in the SKQD experiments. Based on this observation, we conclude that both CIPSI and HCI fail to find the ground state of the system, even for arbitrarily small selection thresholds.
Setting the threshold to zero will of course permit them to explore the entire Hilbert space, but this reduces to (non-selected) configuration interaction, which has exponential cost.
In the cases of ASCI and TrimCI, we have direct control over the diagonalization dimension.
Therefore, we simply swept the dimension until it reached the dimension required by SKQD.
The result is that SKQD obtains an advantage over off-the-shelf SCI strictly in accuracy. 

In addition to the off-the-shelf SCI methods, we tested two iterative solvers that we designed specifically to target Hamiltonians constructed as in \Cref{sec:construction}.
These methods, truncated Arnoldi's method and the diagonal ranking heuristic, were presented in \Cref{sec:classical_approaches}, and both were able to find the exact ground state energy with lower diagonalization dimension than SKQD, as shown in \Cref{fig:results_bad_app}.
Panel (b) illustrates that DMRG was also able to solve this problem.
It is also worth mentioning that diagonalizing in the basis sampled from the AQC circuits alone is able to find the exact ground state energy with a Krylov dimension of 20; however, this is still higher than the dimension required using our hybrid AQC-Trotter evolution circuits.

\section{Circuit transpilation}
\label{app:transpilation}

\subsection{Transpilation with Rustiq}
\label{app:rustiq_transpilation}

\paragraph{Pauli network synthesis} In this work, we used an extension of the Hamiltonian simulation synthesis routine implemented in rustiq \cite{martiel2025rustiq} and described in \cite{debrugiere2024faster}. This synthesis routine greedily synthesizes a circuit that implements a single Trotter step by iteratively growing a Clifford circuit such that all terms of the target Hamiltonian are reduced to a single qubit Pauli somewhere in the circuit. Such a circuit is called an (unordered) Pauli network for some Hamiltonian $H$. Once a Pauli network is generated, it remains to inject single qubit rotations to implement the correct operator.

Consider for instance the following 2-qubit Hamiltonian:

\[H = \theta_1 XX + \theta_2 ZZ + \theta_3 XZ + \theta_4 ZX \]

The following Clifford circuit is a Pauli network for $H$:

\[
\Qcircuit @C=1em @R=.7em {
   & \ctrl{1} & \qw & \ctrl{1} & \qw  \\
   & \gate{X} & \gate{S} & \gate{X} & \qw
}
\]

Indeed, one can implement a Trotter step for $H$ using the following circuit:
\[
\Qcircuit @C=1em @R=.7em {
    &\ctrl{1} & \gate{R_X(\theta_1)} & \qw & \ctrl{1} & \qw& \gate{R_Y(\theta_3)}  &\qw \\
  & \gate{X} & \gate{R_Z(\theta_2)} & \gate{S} & \gate{X} & \qw& \gate{R_Y(\theta_4)} &\qw
}
\]

\paragraph{Synthesis approach} Rustiq's iteration step is to greedily add to the current Pauli network the 2-qubit Clifford gate that makes the most progress toward synthesizing all the terms of $H$. This progress is quantified via a cost metric of the form:
\begin{align}\label{eq:rustiq_cost_function}
    \operatorname{cost}(T_1, ..., T_m) = \sum_{1\leq i \leq m} 2^{-\operatorname{cnot\_cost}(T_i)} 
\end{align} 

In Eq. \ref{eq:rustiq_cost_function}, each $T_i$ is a term of $H$ conjugated by the current Clifford frame. The function $\operatorname{cnot\_cost}$ returns the number of CNOT gates required to \emph{fold} a given Pauli operator onto a single qubit. This cost metric uses a Steiner tree approximation algorithm and takes into account the hardware connectivity constraints (see \cite{Martiel_2022} for a similar cost evaluation).

The algorithm evaluates each possible coupling available in the hardware graph and stores the best patterns consisting of two single qubit Clifford gates followed by a CNOT gate. It then uses a maximal weight matching algorithm to find the best entangling depth-1 circuit and appends the corresponding gate to the Pauli network. If a Pauli operator $T_i$ becomes trivial due to this new layer of gates, it is removed from the pool of operators.

Overall, the algorithm takes as input an Hamiltonian $H = \sum \alpha_i T_i$ and produces a parametrized circuit ${U(t) = C\prod_i e^{i\alpha_it T_i}}$ where $C$ is the Clifford circuit implemented by the bare Pauli network.

\paragraph{Routine used in this work}

For the Hamiltonian experimentally studied in this work, we carried out the above algorithm for $10^6$ different randomizations of the Hamiltonian term ordering. This was implemented in a high throughput calculation, running 1000 batches of 1000 randomizations each as separate batch jobs in a compute cluster. The very large number of randomizations is necessary given the tight tails of the distribution of resulting circuit depths (and gate counts) that the algorithm produces (see Fig.~\ref{fig:circuit_depth-count}). Each individual batch job ran on a single core and lasted for approximately $6.3$ minutes. We employed an optimization that detects early when a randomization will not converge to a circuit within 1000 steps (our self-imposed limit); this only seems to occur in about 17.6\% of the cases, but left unchecked those cases would result in about 10x worse runtime. We then selected the transpilation solution that minimized the circuit depth of a single Trotter step over the various batches (each batch run having produced its own minimum). This approach is embarrassingly parallel apart from the final reduction step. The total classical runtime of this optimization was an aggregate number of 105.5 core-hours on Intel Sapphire Rapids (2.1GHz) hardware.

\begin{figure}
    \centering
    \includegraphics[width=1\linewidth]{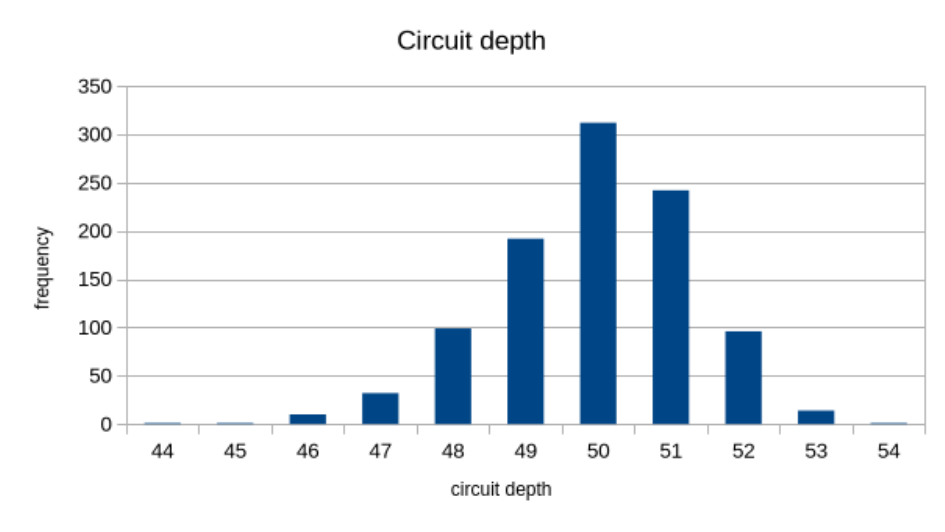}
    \includegraphics[width=1\linewidth]{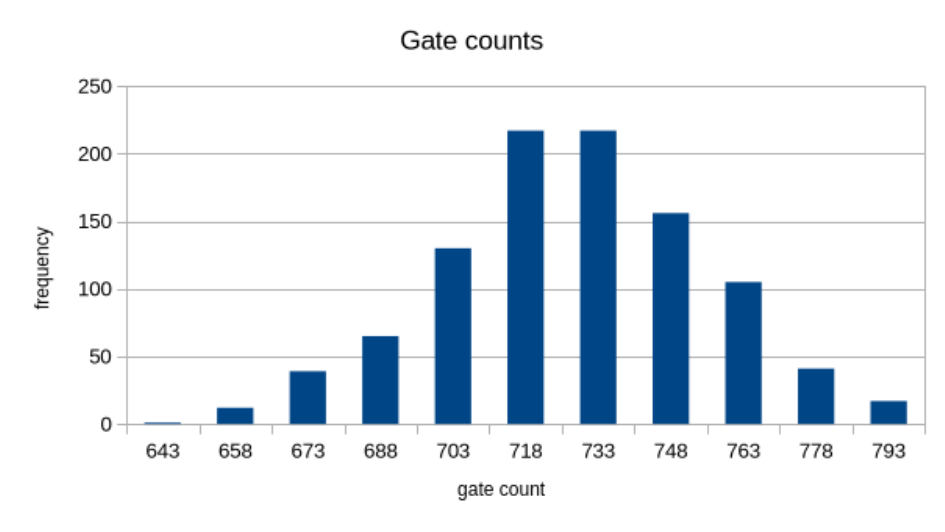}
    \caption{The distribution of circuit depths (top) and gate counts (bottom) achieved during transpilation showing the tight tails to the left.}
    \label{fig:circuit_depth-count}
\end{figure}

\subsection{Approximate quantum compiling}
\label{app:aqc_tensor}

\subsubsection{Background}
Approximate quantum compiling (AQC) describes the task of replacing a target quantum computation by a parameterised circuit of lower resource cost, such that the output is approximately preserved. In general this can be posed either as full-unitary compilation, where a parameterised unitary is trained to approximate a target unitary on all inputs, or as fixed-input state compilation, where only the action on a specified input state is matched. In this work we consider the state compilation setting.

AQC-Tensor~\cite{robertson2025approximate} enables state compilation at scales beyond statevector simulation by representing the target state as a tensor network and evaluating overlaps by tensor-network contraction. Let $\lvert \psi_{\mathrm{t}} \rangle$ denote a target state and let $V(\boldsymbol{\theta})$ denote a parameterised ansatz circuit acting on a fixed reference input $\lvert \psi_{0} \rangle$ (in our case $\lvert 0\rangle^{\otimes n}$). We optimise the global state fidelity
\begin{equation}
F(\boldsymbol{\theta}) = \left|\left\langle \psi_{\mathrm{t}} \right| V(\boldsymbol{\theta}) \left| \psi_{0} \right\rangle\right|^{2},
\qquad
C(\boldsymbol{\theta}) = 1 - F(\boldsymbol{\theta}),
\label{eq:aqc_global_cost}
\end{equation}
and return parameters $\boldsymbol{\theta}^{\star}$ that minimise $C(\boldsymbol{\theta})$. In AQC-Tensor, the overlap in Eq.~\eqref{eq:aqc_global_cost} is evaluated by contracting matrix product state (MPS) representations of $\lvert \psi_{\mathrm{t}} \rangle$ and the ansatz output state $V(\boldsymbol{\theta})\lvert \psi_{0} \rangle$. We use the Qiskit AQC-Tensor addon~\cite{qiskit-addon-aqc-tensor}, which provides the tensor-network evaluation and optimization loop, together with a Quimb ~\cite{gray2018quimb} backend that supports JAX-based automatic differentiation for efficient gradient evaluation.

\subsubsection{Routine used in this work}
In SKQD, the Krylov basis states are $\{\lvert \psi_k\rangle\}_{k=0}^{d-1}$ with $\lvert \psi_k\rangle = \exp(-ik\Delta t H)\lvert \psi_0\rangle$ (Section~\ref{sec:skqd_intro}). We apply AQC-Tensor independently for each Krylov dimension $k \in \{1,\ldots,20\}$, where the compilation target is the state $\lvert \psi_{\mathrm{t}} \rangle \equiv \lvert \psi_k\rangle$. The output of the compilation is a shallow circuit $V_k(\boldsymbol{\theta}_k^\star)$ such that $V_k(\boldsymbol{\theta}_k^\star)\lvert 0\rangle^{\otimes n}$ approximates $\lvert \psi_k\rangle$ in the sense of maximising the global fidelity in Eq.~\eqref{eq:aqc_global_cost}. These compiled circuits are then used for sampling to form the set $B_{d,M}$ defined in Eq.~(2) of Section~\ref{sec:skqd_intro}.

For a fixed $k$, we first obtain an MPS representation of the target state $\lvert \psi_k\rangle$ by simulating the corresponding circuit using the Qiskit Aer MPS simulator~\cite{qiskit2024}. To estimate the accuracy of the MPS, we simulate the circuit with bond dimension increasing in powers of 2 and study the convergence of the overlap with respect to the $\chi=256$ MPS. Given the approximate nature of AQC, some level of truncation error in the compilation target can still lead to useful circuits. In this way, for $k=1$ and $k=2$ we find that |$\langle\psi_k^{(\chi=128)}|\psi_k^{(\chi=256)}\rangle|^2 > 0.99$ and so consider $\chi=128$ to be sufficient. For $k>2$ this is not satisfied, so we use $|\psi_k^{(\chi=256)}\rangle$ as the target. We leave a more rigorous analysis of the accuracy of $|\psi_k^{(\chi=256)}\rangle$ when extrapolating to the limit $\chi\rightarrow\infty$ for future work.

The ansatz we optimize is constructed from two-qubit $\mathrm{SU}(4)$ blocks with 9 real parameters per block~\cite{ibm_aqctensor_docs}. For time-evolution circuits, a known good layout is to replace each interacting term in a Trotter step circuit by a parameterised block~\cite{robertson2025approximate}. However, for the Hamiltonians considered in this work, one layer of such an ansatz would have a CZ depth greater than 100. To reduce this depth, we instead construct an ansatz where one layer consists of a parameterised block on each \emph{unique} interaction edge in the Hamiltonian connectivity graph, which has CZ depth of 15 instead.

Finally, we initialize and optimize the ansatz parameters. Since the performance of AQC optimizations depends strongly on the quality of the initial parameter vector, we warm-start each optimization using a tensor-network compression heuristic. For each target MPS $\lvert \psi_k\rangle$, we compress the MPS to bond dimension $\chi=1$, yielding a product-state approximation $\lvert \psi_k^{(\chi=1)}\rangle$. We then initialize $\boldsymbol{\theta}_k$ such that the ansatz circuit prepares this compressed state at iteration zero, i.e.,
\begin{equation}
V_k(\boldsymbol{\theta}_k^{(0)})\lvert 0\rangle^{\otimes n} = \lvert \psi_k^{(\chi=1)}\rangle,
\end{equation}
and subsequently optimize $\boldsymbol{\theta}_k$ to minimize the global infidelity $C(\boldsymbol{\theta}_k)$ in Eq.~\eqref{eq:aqc_global_cost}. This initialization strategy has been shown to perform significantly better than random initialization \cite{jaderberg2025variational}.

\subsubsection{Compilation results}

Using the routine defined above, we run compiling for each $k$ using a two-layer ansatz. This circuit has 27 CZ depth and contains a total of 1070 trainable parameters. Specifically we use the \texttt{QuimbSimulator} ~\cite{gray2018quimb} backend of AQC-Tensor, which enables automatic gradient calculation via the JAX auto-differentiation engine~\cite{jax2018github}, with a max bond dimension of 256. We use the ADAM optimizer~\cite{kingma2014adam} with learning rate $\alpha=10^{-3}$, $\beta_1=0.9$, $\beta_2=0.999$ and bias $\epsilon=10^{-8}$. Optimization is run on a cluster node with 4 CPU cores and 128gb of RAM for 68 hours, during which between 27000 and 1800 iterations are completed for $k=1$ and $k=20$ respectively.

Table~\ref{tab:aqc_fidelity} shows the final fidelity of the optimized circuits. Here the fidelity decays with increasing $k$, from 0.97 at $k=1$ to 0.08 at $k=20$. This aligns with the expectation that a fixed-depth circuit cannot approximate deeper circuits and longer Hamiltonian evolution to constant accuracy. Nevertheless, as shown in the main text, we find that the non-zero overlap achieved even at large $k$ allows sampling of configurations that are significant to finding the true ground state. Most importantly, the optimized AQC circuits have between a factor of 3 to 30 reduced CZ depth for Krylov dimension $k=1$ to $k=20$ respectively.

\begin{table}[h!]
\centering
\begin{tabular}{cccccccc}
\hline
\textbf{k} & \textbf{$F$} & \textbf{k} & \textbf{$F$} & \textbf{k} & \textbf{$F$} & \textbf{k} & \textbf{$F$} \\
\hline
1  & 0.9738 & 6  & 0.3636 & 11 & 0.2316 & 16 & 0.1609 \\
2  & 0.9063 & 7  & 0.3430 & 12 & 0.1574 & 17 & 0.1353 \\
3  & 0.7519 & 8  & 0.2787 & 13 & 0.1556 & 18 & 0.1042 \\
4  & 0.5321 & 9  & 0.2850 & 14 & 0.1660 & 19 & 0.1019 \\
5  & 0.4491 & 10 & 0.2239 & 15 & 0.1815 & 20 & 0.0808 \\
\hline
\end{tabular}
\caption{AQC fidelity as defined in Eq. (\ref{eq:aqc_global_cost}) after optimization for the time evolution circuits generating each Krylov dimension $k$.}
\label{tab:aqc_fidelity}
\end{table}

\section{Quantum experimental details}
\label{app:experimental_details}

All the reported experiments were on IBM Quantum's \textit{ibm\_boston}, which is a Heron r2 processor with 156 fixed-frequency transmon qubits with tunable couplers where the qubits are laid out on a heavy-hex lattice layout.  49 out of the 156 qubits, corresponding to a 3 by 2 heavy-hex grid were chosen for these experiments. \Cref{fig:backend_info} shows the chosen qubit layout along side the backend specifications from 2025-12-26, which was latest of the data collection days. Each Krylov dimension was sampled with at least $10^6$ shots, with Krylov dimensions near $k=15$ with $k_Q > 0$ assigned a larger number of shots. As explained above, the shot allocation corresponds to an iterative strategy where we heuristically attempted was to minimize energy while also minimizing the total quantum runtime. The experiments were run on different batches on the dates 2025-12-15, 2025-12-18, 2025-12-20, 2025-12-26, 2026-01-28. The table \Cref{tab:backend_details} list the average $T_1$ and $T_2$ times, as well as single-qubit, two-qubit and readout durations and errors. On 2025-12-26, the average values for readout-error was $6.9\times10^{-3}$, single-qubit error was $1.8\times10^{-4}$ and two-qubit gate-error was $1.43\times10^{-3}$, with the latter two characterized by randomized benchmarking. The average relaxation and dephasing times were $290 \ \mu \text{s}$ and $360 \ \mu \text{s}$.  The average values of these errors and relaxation times remained steady over the different dates. A cumulative distribution for the two-qubit gate infidelities acquired using layer fidelity randomized benchmarking collected right before the experiments were initialized for all four dates is shown in \Cref{fig:layerfidelity}. The circuit's two-qubit depth and count, shown in \Cref{tab:2qdepth}, was determined primarily by $k_Q$, i.e., the number of Trotter steps taken after the AQC-compiled circuit was implemented. Odd and even Trotter steps utilized 1119 and 1244 two-qubit gates respectively and all AQC circuits were optimized to have 324 two-qubit operations. 

\begin{figure}[ht]
    \centering
    \includegraphics[width=1\linewidth]{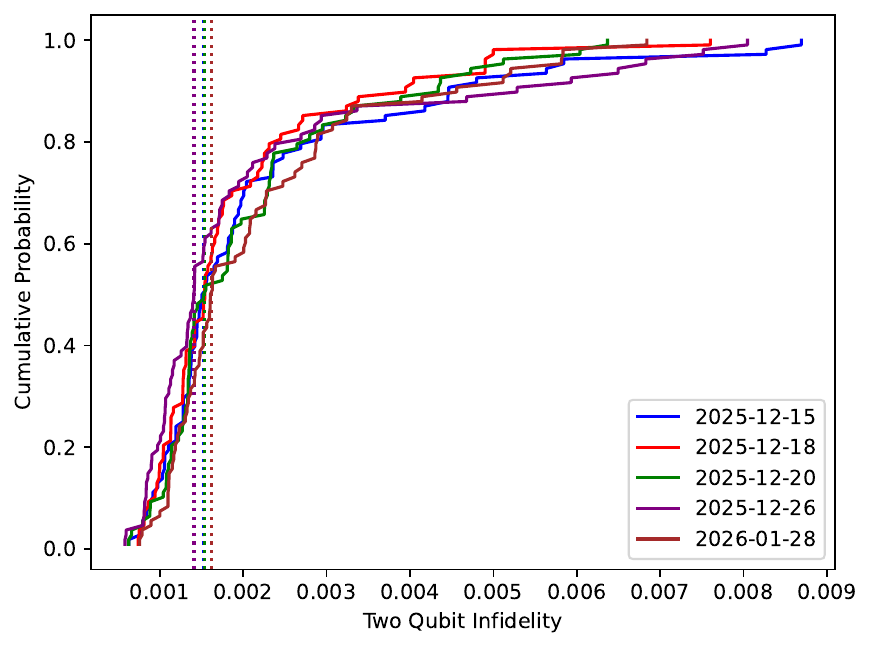}
    \caption{Cumulative distribution of two-qubit infidelity for the 49 chosen qubits on IBM Boston. These are computed before each experiment session using layer fidelity RB.}
    \label{fig:layerfidelity}
\end{figure}

\begin{figure*}[ht]
    \centering
    \includegraphics[width=1\linewidth]{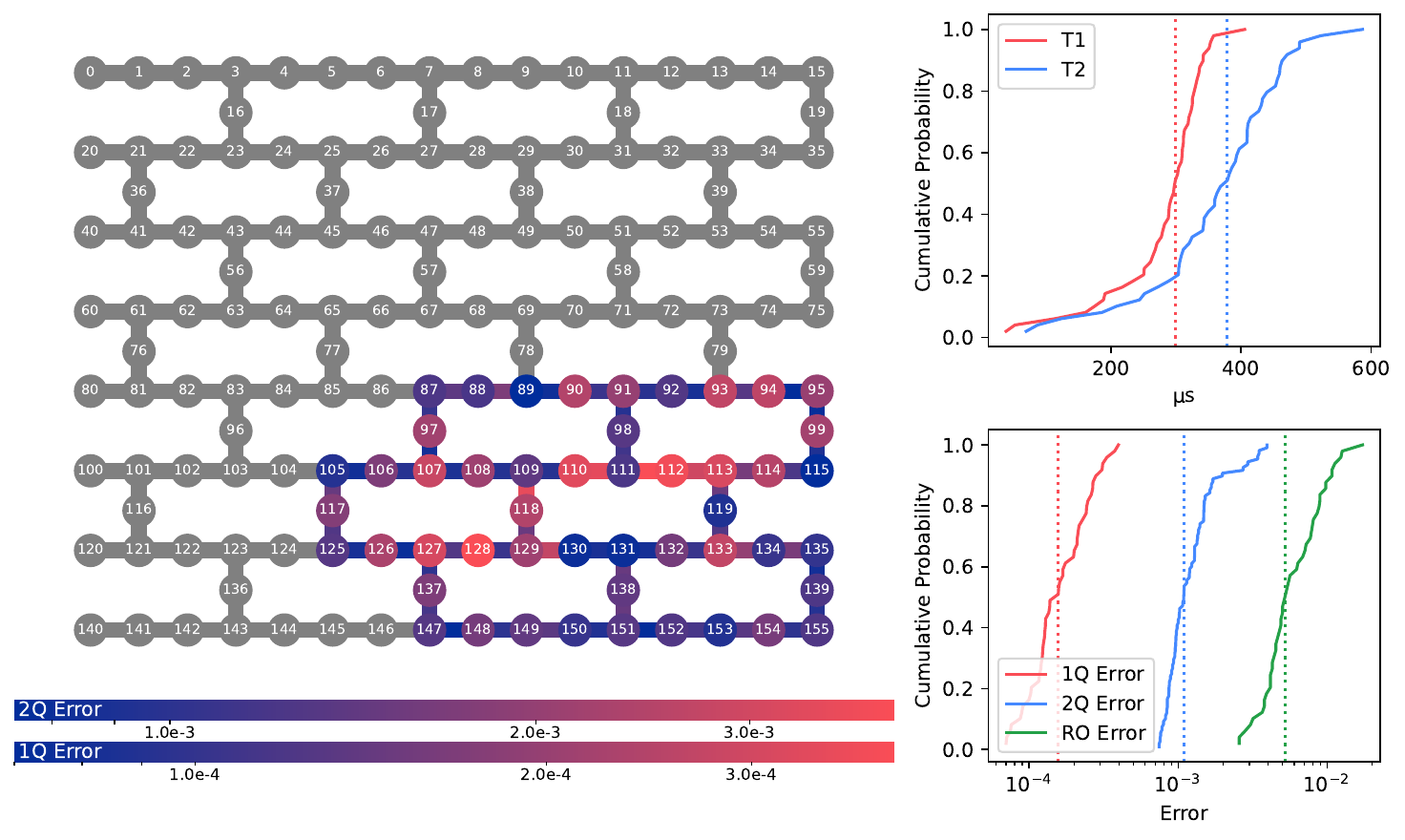}
    \caption{Device connectivity graph for IBM Boston on 2025-12-18 highlighting the 49 chosen qubits in the $3 \times 2$ heavy-hex layout. The $T_1$, $T_2$ times as well as the single-qubit, two-qubit and readout errors for the chosen qubits are shown on the right.}
    \label{fig:backend_info}
\end{figure*}

\begin{figure}[ht]
    \centering
    \includegraphics[width=1\linewidth]{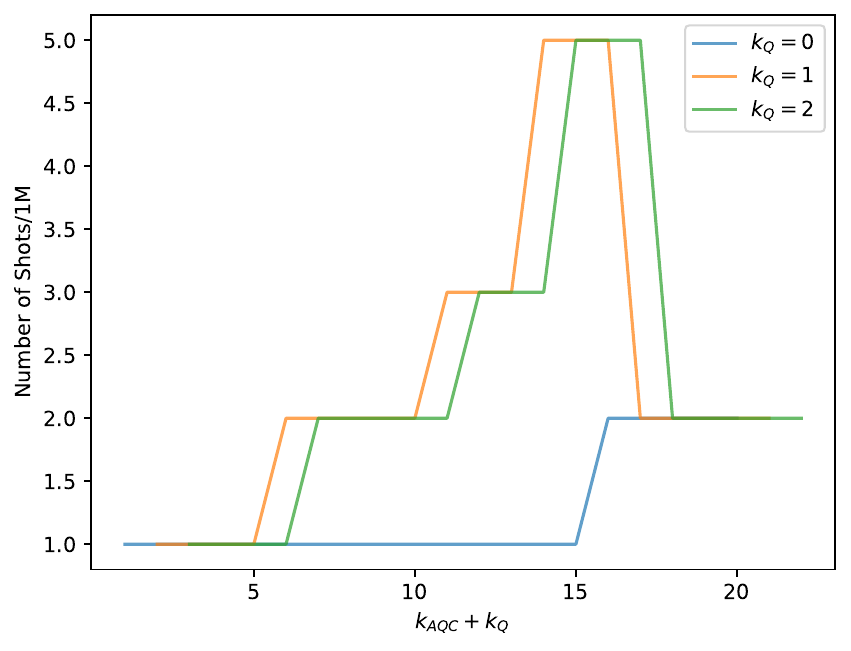}
    \caption{The number of shots per Krylov dimension was chosen adaptively while monitoring the energy, and more shots were assigned to experiments that were able to improve upon the previous energy estimates. The plot shows the final distribution of the shot budget with different curves representing different number of post-AQC Trotter steps.}
    \label{fig:shot_allocation}
\end{figure}

\begin{table*}[]
    \centering
    \begin{tabular}{|c|c|c|c|c|c|c|c|c|}
     \hline
         Date & $T_1$ & $T_2$  & 1Q Duration & 2Q Duration & RO Duration & 1Q Error  & 2Q Error  & RO Error \\
          \hline
2025-12-15 & $2.8  \times  10^{2} \mu$s & $3.5  \times  10^{2} \mu$s & 32 ns & 69 ns & $2.2  \times  10^{3}$ ns & $1.6  \times  10^{-4}$  & $1.5  \times  10^{-3}$  & $6.7  \times  10^{-3}$  \\
2025-12-18 & $2.8  \times  10^{2} \mu$s & $3.6  \times  10^{2} \mu$s & 32 ns & 69 ns & $2.2  \times  10^{3}$ ns & $1.8  \times  10^{-4}$  & $1.3  \times  10^{-3}$  & $6.4  \times  10^{-3}$  \\
2025-12-20 & $2.7  \times  10^{2} \mu$s & $3.5  \times  10^{2} \mu$s & 32 ns & 69 ns & $2.2  \times  10^{3}$ ns & $1.8  \times  10^{-4}$  & $1.3  \times  10^{-3}$  & $5.3  \times  10^{-3}$  \\
 2025-12-26 & $2.9  \times  10^{2}$ $\mu$s & $3.6  \times  10^{2}$ $\mu$s & 32 ns & 69 ns & $2.2  \times  10^{3}$ ns  & $1.8  \times  10^{-4}$  & $1.4  \times  10^{-3}$  & $6.9  \times  10^{-3}$ \\ 
 2026-01-29 & $2.8  \times  10^{2}$  $\mu$s & $3.5  \times  10^{2}$ $\mu$s & 32 ns & 69 ns & $2.2  \times  10^{3}$ ns & $1.6  \times  10^{-4}$  & $1.7  \times  10^{-3}$  & $7.2  \times  10^{-3}$  \\
  \hline
    \end{tabular}
    \caption{Device specifications, in particular the $T_1$ and $T_2$ times as well as the gate and readout errors and durations, for the chosen qubits and gates of IBM Boston for each of the days the experiment was performed are listed.}
    \label{tab:backend_details}
\end{table*}

\begin{table}[]
    \centering
    \begin{tabular}{|c|c|c|c|}
     \hline
& $k_Q = 0$ & $k_Q = 1$ & $k_Q = 2$ \\
\hline
2Q Gates & 324 & 1443 & 1568\\
2Q Depth & 27 & 111 &  113 \\
\hline
    \end{tabular}
    \caption{Number of two-qubit gates and their corresponding two-qubit gate depth for various values of $k_Q$.}
    \label{tab:2qdepth}
\end{table}

\section{Projection and diagonalization}
\label{app:diagonalization}

This section explains how we diagonalize the Hamiltonian $H$ projected onto a subspace spanned by a given set of computational basis states $B$. This task is accomplished in two steps:
\begin{itemize}
    \item Compute an explicit sparse matrix representation (such as compressed sparse row format) of the projected Hamiltonian $H_B$.
    \item Pass the sparse matrix to the Lanczos~\cite{lanczos1950iteration} algorithm to find its eigenvector with lowest eigenvalue.
\end{itemize}
In this work, $H$ is represented as a generic linear combination of Pauli operators, $H = \sum_{k = 1}^L \alpha_k T_k$. The rest of this section explains how to compute the explicit representation of $H_B$.

Computing an explicit matrix representation of $H_B$ requires picking an ordering $\{x_1, \ldots, x_{\abs{B}}\}$ for the elements of $B$ and creating a data structure that can determine whether a given bitstring is contained in $B$, and if so, return its ordering index, or \textit{address}. One possibility is to order the bitstrings by their integer representation and store them in a sorted list. Querying an address then involves searching the list, which can be done in time $O(\log\abs{B}) = O(n)$ using binary search. Another possibility is to store the bitstrings in a hash table, with the ordering chosen arbitrarily. In this case, querying an address takes time $O(1)$ amortized. While the hash table is theoretically superior, the overhead of maintaining a hash table may make the list better for small system sizes.

Once the data structure for storing $B$ has been chosen, the next step is to compute and store the nonzero matrix elements $\bra{x_i} H \ket{x_j}$ for $x_i, x_j \in B$.
A naive way would be to loop through every pair $(x_i, x_j)$ and compute $\bra{x_i} H \ket{x_j}$. Computing $\bra{x_i} H \ket{x_j}$ takes time $O(nL)$, where $n$ is the number of qubits, because there are $L$ Pauli terms in the Hamiltonian, and computing $\bra{x_i} T_k \ket{x_j}$ for a Pauli $T_k$ can be done in time $O(n)$. Since there are $O(\abs{B}^2)$ pairs of elements, the naive way has total cost $O(nL\abs{B}^2)$.

A faster way to compute the matrix elements takes advantage of the fact that applying a Pauli operator to a computational basis state yields another computational basis state, possibly with a complex phase. We first loop over the Pauli terms of the Hamiltonian, and for each Pauli $T_k$, loop over the elements $x_i$ of $B$ and compute $T_k\ket{x_i} = c\ket{y_i}$. Then, we can query the address of $y_i$ and update the appropriate corresponding matrix element. Computing $T_k\ket{x_i}$ and querying $y_i$ can be done in time $O(n)$, so the total cost of this method is $O(nL\abs{B})$.

\section{Details of classical numerical experiments}
\label{app:classical_simulation_details}

\subsection{Truncated Arnoldi's method}
\begin{figure*}
    \centering
    \includegraphics[width=\linewidth]{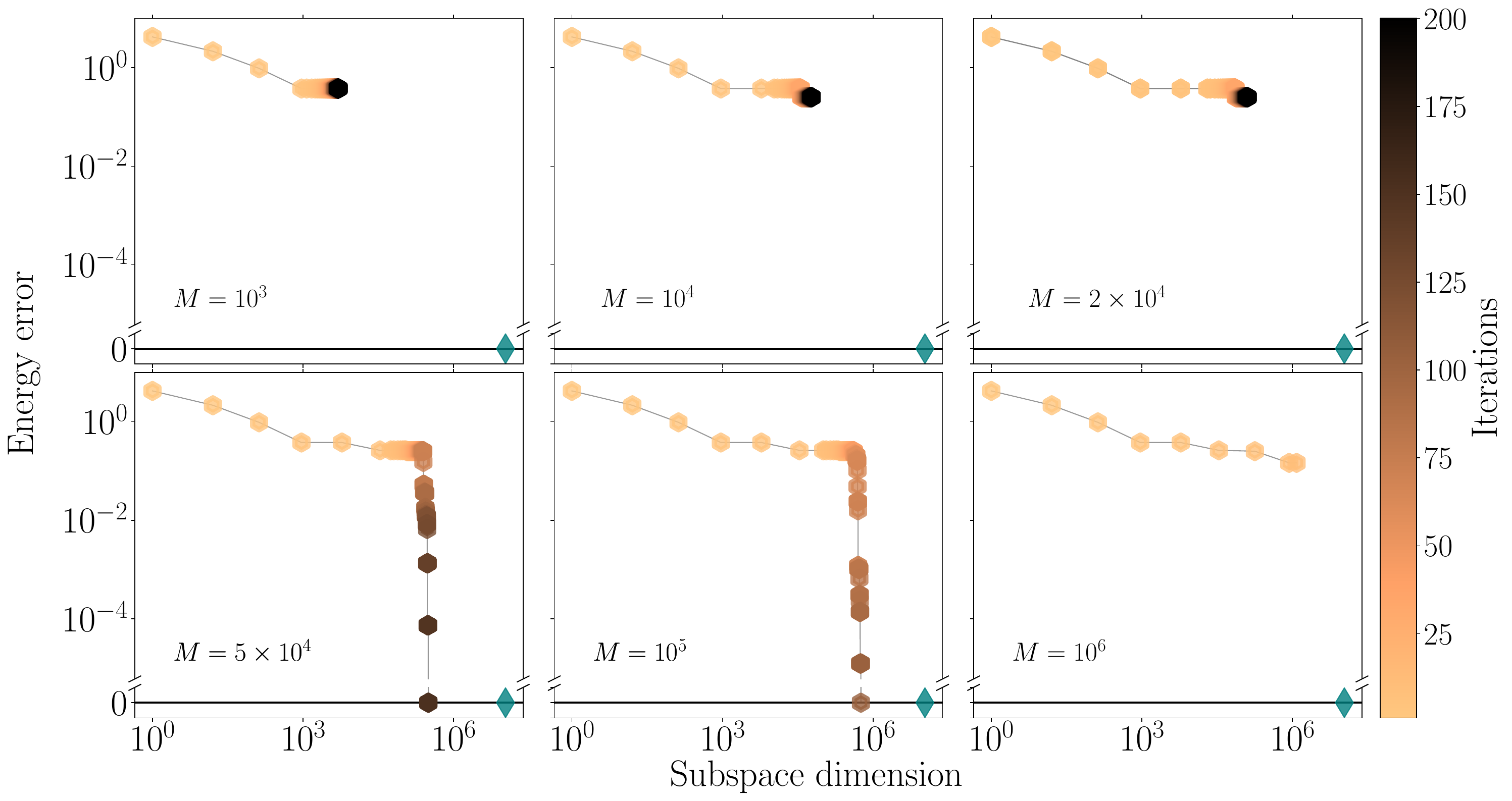}
    \caption{Results from all runs of truncated Arnoldi's method applied to the Hamiltonian described in \Cref{sec:setup}. Each panel corresponds to a different cutoff $M$ on the number of new configurations added in each iteration. In all cases, we capped the number of iterations at $200$, and the $M=10^6$ run was terminated early once it surpassed the subspace dimension of the $M=10^5$ run, which had already found the ground state.}
    \label{fig:trunc_arn_all_runs}
\end{figure*}

For truncated Arnoldi's method, the only hyperparameters are the cutoff $M$ on number of new configurations per iteration, and the number $T$ of iterations. We capped $T$ at $200$ in all of our runs, since this appeared to be past the point where the algorithm was still progressing. We swept $M$ over $10^2, 10^3, 10^4, 10^5, 10^6$. After noting that the method succeeded in finding the exact ground state at $M=10^5$, we attempted to refine the required subspace dimension by testing $M=2\times10^4,5\times10^4$. The former failed to find the exact ground state. The latter succeeded after $148$ iterations and in diagonalization dimension $313303$, yielding the results shown in \Cref{ssec:results}. 

The results from all runs of truncated Arnoldi's method are shown in \Cref{fig:trunc_arn_all_runs}.
As noted in the caption, all of the runs were capped at $200$ iterations, with the exception of the new-configuration cutoff $M=10^6$ run, which was terminated once it reach higher subspace dimension than the $M=10^5$ run, since the latter had already found the ground state.
The $M=2\times10^4$ and $M=5\times10^4$ runs were added after the initial sweep over powers of $10$ in order to further refine the subspace dimension required to find the ground state.
The gap between the final subspace dimensions of $M=2\times10^4$ (the largest cutoff that did not find the ground state) and $M=5\times10^4$ (the smallest cutoff that did find the ground state) was $127935$ versus $313303$, so $127935$ lower bounds the improvement that could be obtained by further refining the cutoff between the above values.
We did not put further resources towards this since the results shown are already sufficient to illustrate the success of the method.
However, if in future work it is necessary to compare resource requirements of truncated Arnoldi's method and SKQD, the required cutoff should be refined as much as possible.

\subsection{Diagonal ranking}\label{sec:diag ranking num}
For these calculations we set the maximum size of the reservoir set to $R = 10^7$ configurations, and perform a sweep in the maximum size of the working set of configurations $D$. A maximum number of iterations $T$ is considered. A node with an Intel Xeon Platinum 8260 CPU with 96 cores (2.40GHz) and 3 TiB of memory is used for these runs. As described in~\Cref{alg:diag ranking}, the method does not rely on performing diagonalizations to carry out the iterations. However, to visualize the progress of the algorithm we perform one diagonalization per iteration.

\Cref{fig:DIAG} shows the progression of the diagonalization energy as a function of the subspace for the different values of $D$ considered here, ranging between $D = 10^3$ to $D = 2\cdot 10^5$. $15$ values are considered in this study. As the size of the working set of configurations is allowed to grow, the energy error decreases monotonically. The smallest value of $D$ for which the method finds the exact ground state is $D = 8\cdot 10^4$, resulting on a diagonalization of the same size. The curve of optimal performance with $D =8\cdot 10^4 $ is the one shown in the main text in~\Cref{fig:results_bad_app}.

\begin{figure}
    \centering
    \includegraphics[width=\linewidth]{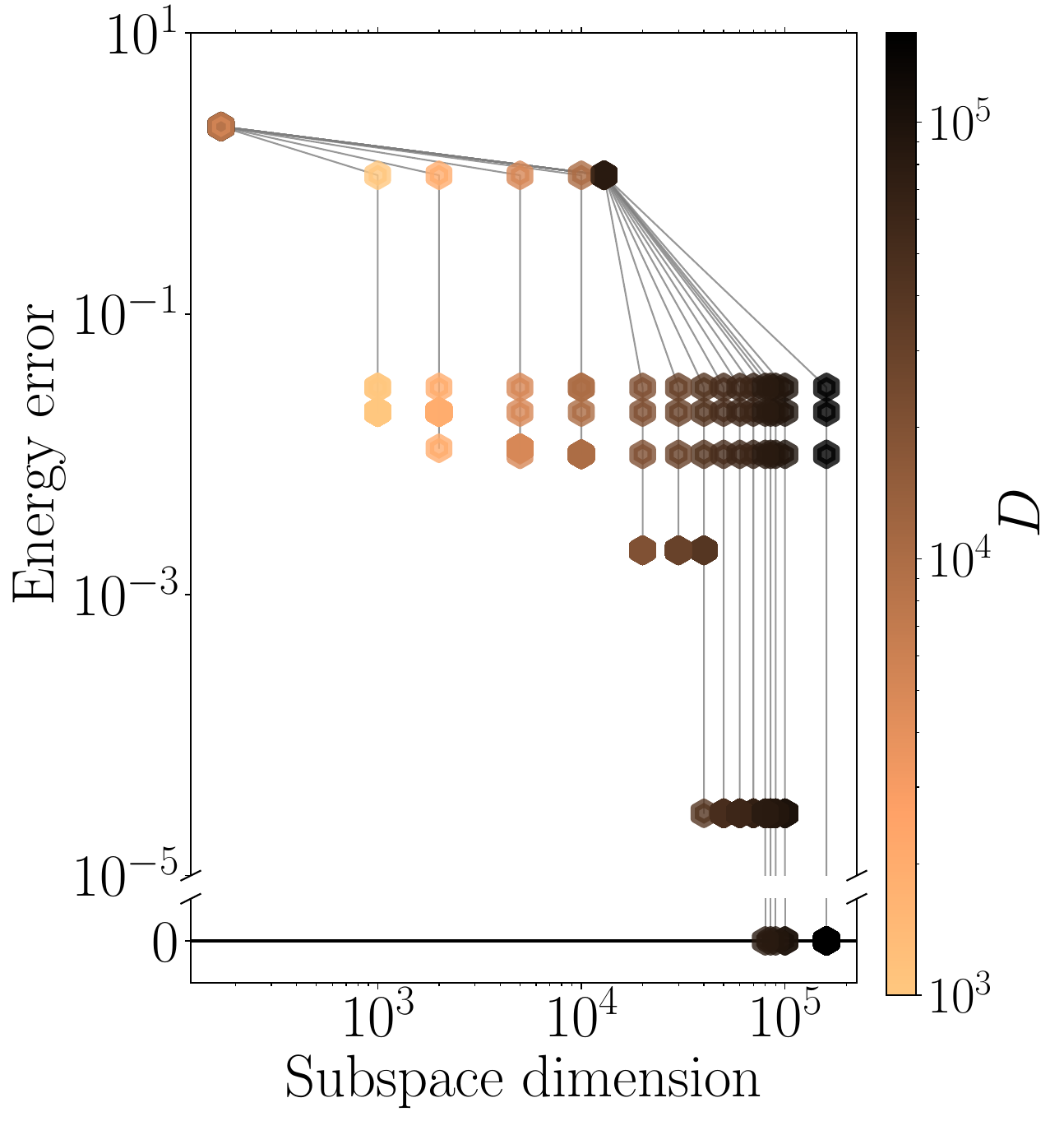}
    \caption{Energy error as a function of the diagonalization subspace dimension obtained by the diagonal ranking method. Each curve shows a different value of the maximum size of the working set of configurations $D$. A diagonalization was performed at each iteration only for illustration purposes, as discussed in the text.}
    \label{fig:DIAG}
\end{figure}

\subsection{CIPSI}\label{sec:CIPSI num}
As described above, the accuracy of CIPSI is controlled by the selection threshold $\varepsilon$. Since the optimal value of $\varepsilon$ is not known \textit{a-priori}, we perform a logarithmic sweep between $\varepsilon = 10^{-4}$ and $\varepsilon = 10^{-17}$. 106 values of $\varepsilon$ have been considered in the analysis. A computer with an Intel Xeon Platinum 8260 CPU with 96 cores (2.40GHz) and 3 TiB of memory is used for these runs. We choose the maximum number of iterations to be $T = 30$. However, we observe that after at most 20 iterations, CIPSI does not find new configurations after their selection, thus terminating the calculation.

Panel (a) in~\Cref{fig:CIPSI} shows the progression of the ground state energy error as a function of the diagonalization subspace dimension for all values of $\varepsilon$. We observe that most values of $\varepsilon$ yield an energy error no-lower than $10^{-2}$.  The early termination of the calculation, as mentioned above, limits the maximum subspace dimension that can be reached, even at extremely small values of the threshold. Panel (b) in~\Cref{fig:CIPSI} shows the progression of the subspace dimension as a function of the calculation iteration for all values of $\varepsilon$. As expected, the diagonalization subspace dimension grows as the calculation progresses, until termination due to the lack of new accepted configurations. The maximum number of iterations increases as $\varepsilon$ is decreased for the larger value of the threshold, until a maximum of approximately $20$ iterations. As the value of $\varepsilon$ is further decreased, the maximum number of iterations before termination decreases again. Panel (c) in~\Cref{fig:CIPSI} shows the energy error as a function of $\varepsilon$. As $\varepsilon$ is decreased, so does the energy error, following a series of jumps and wide plateaus. The plateau-dominated energy error landscape makes it challenging to find automate the search over the $\varepsilon$ parameter, limiting the options to costly sweeps over  $\varepsilon$. The inset hows the saturation of the maximum subspace dimension reached as a function as a function of the selection threshold. 

The saturation of the energy estimate and the maximum diagonalization subspace dimension over $\sim 10$ orders of magnitude makes it likely that further reducing the value of $\varepsilon$ will not allow for the exact estimation of the ground state of the system.

\begin{figure*}
    \centering
    \includegraphics[width=\linewidth]{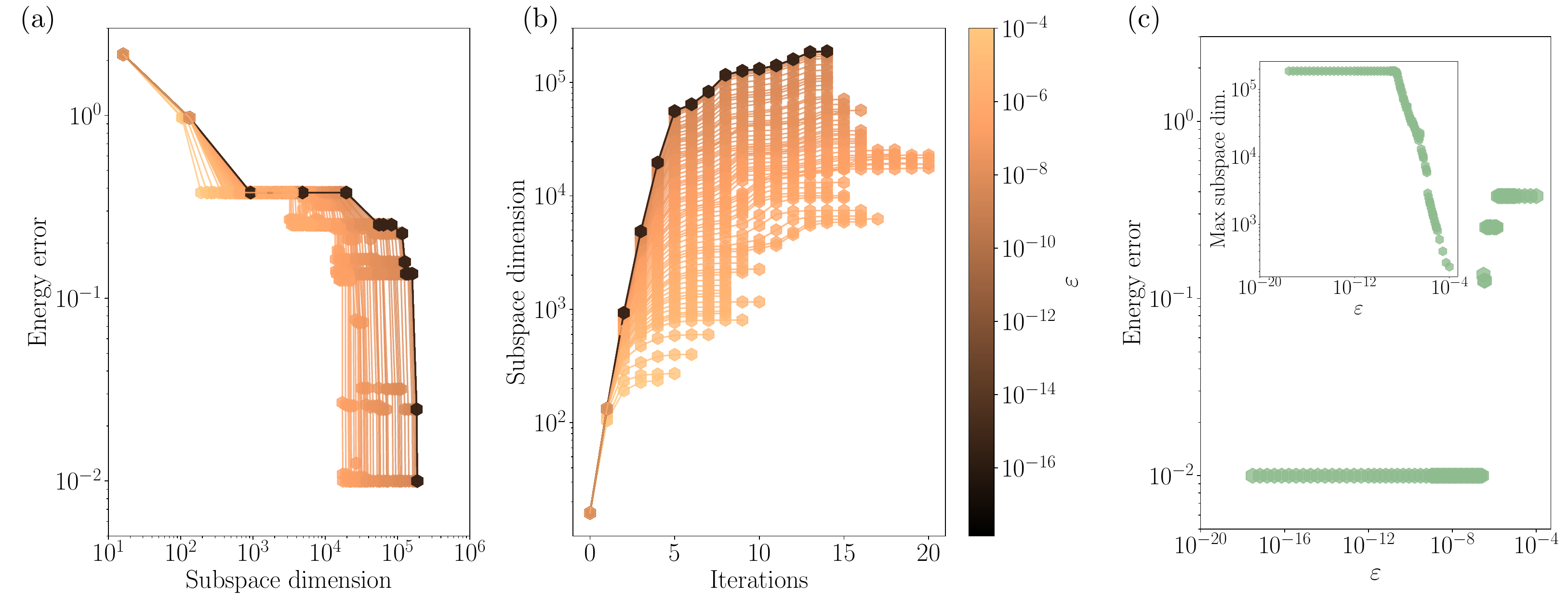}
    \caption{Performance of CIPSI in the proposed Hamiltonian for a sweep of 106 values of $\varepsilon$ between $\varepsilon = 10^{-4}$ and $\varepsilon = 10^{-17}$. \textbf{(a)} Energy error as a function of the subspace dimension. Each curve shows the progression of a CIPSI calculation whose value of $\varepsilon$ is indicated by the colorbar. \textbf{(b)} Evolution of the diagonalization subspace dimension as a function of the calculation iteration. Each curve shows the progression of a CIPSI calculation whose value of $\varepsilon$ is indicated by the colorbar. \textbf{(c)} Lowest energy obtained as a function of the selection threshold. The inset shows the maximum diagonalization subspace dimension as a function of the selection threshold.}
    \label{fig:CIPSI}
\end{figure*}

\subsection{HCI}\label{sec:HCI num}
Similarly to CIPSI, the accuracy of HCI is controlled by the selection threshold $\varepsilon$, and its optimal value is not known \textit{a-priori}. Consequently, we perform a logarithmic sweep between $\varepsilon = 10^{-5}$ and $\varepsilon = 10^{-20}$. 50 values of $\varepsilon$ have been considered in the analysis. A computer with an Intel Xeon Platinum 8260 CPU with 96 cores (2.40GHz) and 3 TiB of memory is used for these runs. We choose the maximum number of iterations to be $T = 30$. However, we observe that after at most 20 iterations, HCI does not find new configurations after the selection step, yielding the termination the calculation.

Panel (a) in~\Cref{fig:HCI} shows the progression of the ground state energy error as a function of the diagonalization subspace dimension for all values of $\varepsilon$. We observe that most values of $\varepsilon$ yield an energy error no-lower than $10^{-2}$.  The early termination of the calculation, as mentioned above, limits the maximum subspace dimension that can be reached, even at extremely small values of the threshold. Panel (b) in~\Cref{fig:HCI} shows the progression of the subspace dimension as a function of the calculation iteration for all values of $\varepsilon$. As expected, the diagonalization subspace dimension grows as the calculation progresses, until termination due to the lack of new accepted configurations. The maximum number of iterations increases as $\varepsilon$ is decreased for the larger value of the threshold, until a maximum of approximately $20$ iterations. As the value of $\varepsilon$ is further decreased, the maximum number of iterations before termination decreases again. Panel (c) in~\Cref{fig:HCI} shows the energy error as a function of $\varepsilon$. As $\varepsilon$ is decreased, so does the energy error, following a series of jumps and wide plateaus. The plateau-dominated energy error landscape makes it challenging to find automate the search over the $\varepsilon$ parameter, limiting the options to costly sweeps over  $\varepsilon$. The inset hows the saturation of the maximum subspace dimension reached as a function as a function of the selection threshold. 

The saturation of the energy estimate and the maximum diagonalization subspace dimension over $\sim 10$ orders of magnitude makes it likely that further reducing the value of $\varepsilon$ will not allow for the exact estimation of the ground state of the system. The qualitative and quantitative behavior of CIPSI an HCI are similar in this Hamiltonian.

\begin{figure*}
    \centering
    \includegraphics[width=\linewidth]{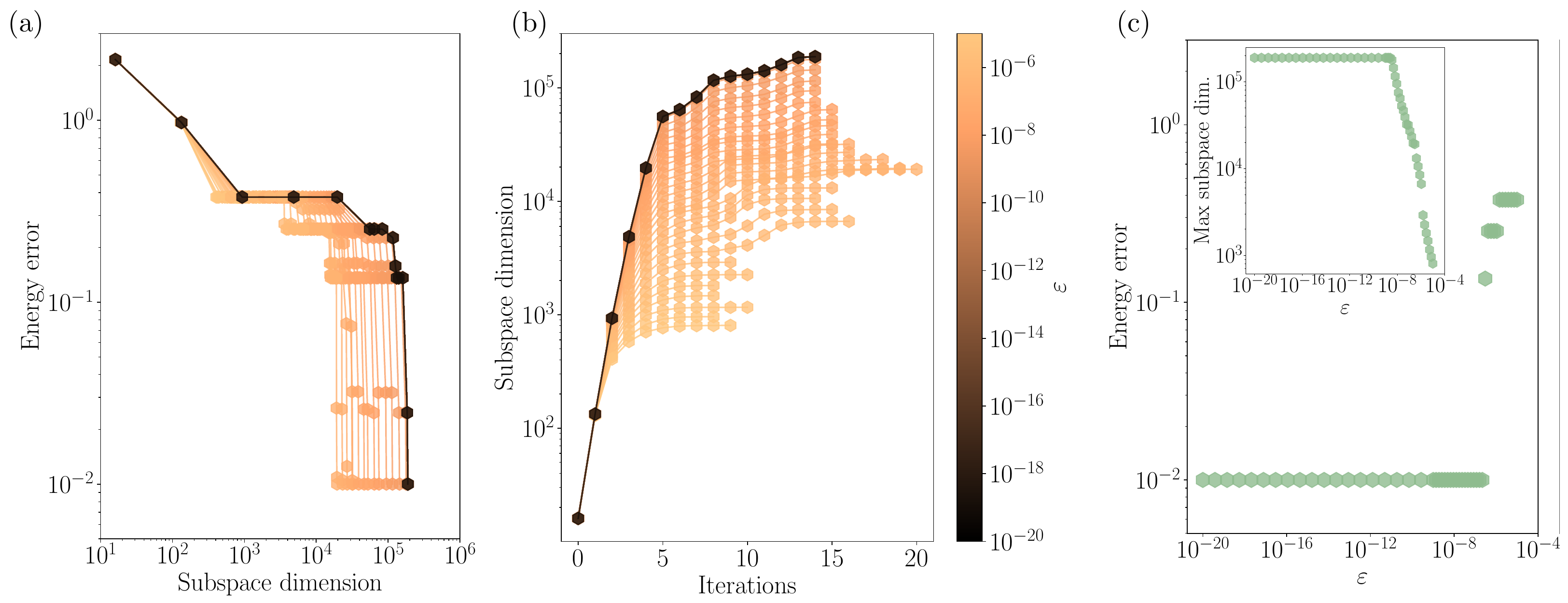}
    \caption{Performance of HCI in the proposed Hamiltonian for a sweep of 106 values of $\varepsilon$ between $\varepsilon = 10^{-4}$ and $\varepsilon = 10^{-17}$. \textbf{(a)} Energy error as a function of the subspace dimension. Each curve shows the progression of a HCI calculation whose value of $\varepsilon$ is indicated by the colorbar. \textbf{(b)} Evolution of the diagonalization subspace dimension as a function of the calculation iteration. Each curve shows the progression of a HCI calculation whose value of $\varepsilon$ is indicated by the colorbar. \textbf{(c)} Lowest energy obtained as a function of the selection threshold. The inset shows the maximum diagonalization subspace dimension as a function of the selection threshold.}
    \label{fig:HCI}
\end{figure*}

\subsection{ASCI}\label{sec:ASCI num}
As described above, the behavior of ASCI is controlled by two hyperparameters. The size of the core set $C$ and the diagonalization subspace dimension $D$. We explore different combinations of $C$ and $D$ values. In particular, we consider $C/D = 3/4, 1/2, 1/4, 1/8$. For each value of the core-to-diagonalization dimension values we consider 12 values of $D = 1,024$, $D = 2,048$, $D = 4,095$, $D = 8,192$, $D  =16,398$, $D = 32,768$, $D = 65,536$, $D  = 131,072$, $D = 262,144$, $D = 524,288$, $D = 1,048,576$, $D = 10,000,000$, for a total of 48 ASCI runs. A computer with an Intel Xeon Platinum 8260 CPU with 96 cores (2.40GHz) and 3 TiB of memory is used for these runs.

Panel (a) in~\Cref{fig:ASCI} shows the ground state energy error obtained from ASCI as a function of the number of iterations performed. The lowest energy error obtained is of $3\cdot 10^{-2}$. Panel (a) in~\Cref{fig:ASCI} also shows the energy error as a function of the diagonalization subspace dimension. For the smallest values of the $C/D$ ratio, the application of the Hamiltonian on the core configurations does not generate enough new configurations for the user define value of $D$ to match with the diagonalization subspace dimension. Panel (b) shows the energy error as a function of $D$, where different panels indicate different values of the $C\D$ ratio. As expected, the energy error decreases as the diagonalization subspace dimension is increased. However, we observe non-monotonic behavior of the energy error with $D$.

\begin{figure*}
    \centering
    \includegraphics[width=\linewidth]{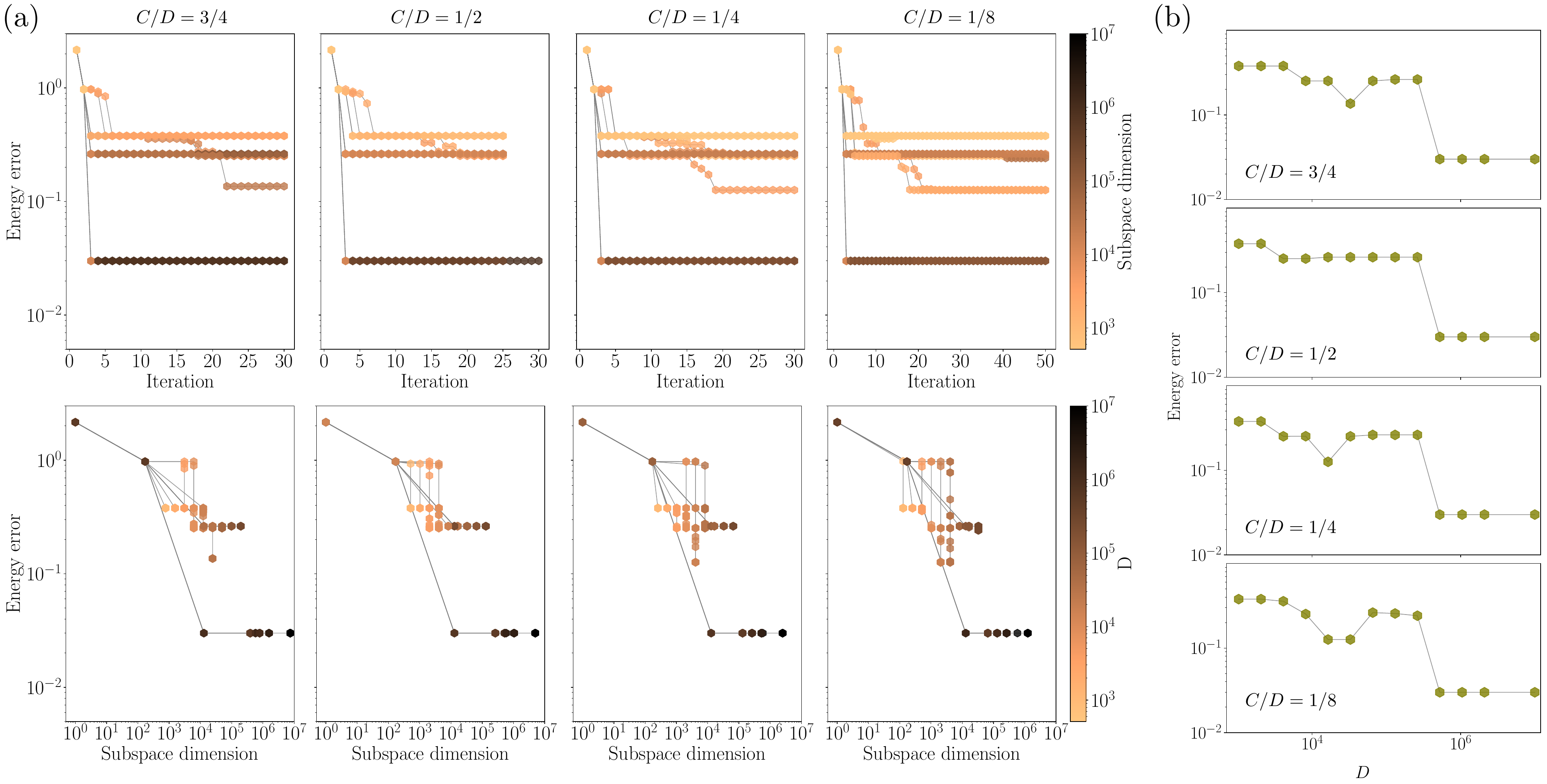}
    \caption{\textbf{Top row (a):} Energy error as a function of the ASCI iteration. Different panels correspond to different values of the core to diagonalization subspace sizes. Different curves correspond to different choices of the diagonalization subspace dimension $D$, as indicated in the colorbar. \textbf{Bottom row (a):} Energy error as a function as a function of the diagonalization subspace dimension. Different curves correspond to different values of $D$ as indicated by colorbar. Different panels correspond to different values of the core to diagonalization subspace sizes. \textbf{(b):} Energy error as a function of $D$ for different values of $C/D$ as shown in each panel.}
    \label{fig:ASCI}
\end{figure*}

\subsection{TrimCI} \label{sec:TrimCI num}
For the TrimCI calculations in this study, we take a value of $F$ (see \Cref{sec:classical_approaches} for its definition) of ${F = 100}$, as suggested by the authors of Ref.~\cite{zhang2025TrimCI_SCI}. We also fix the ratio between the size of the core configurations $C$ and the diagonalization subspace size $D$ to $C/D = 1/10$. The number of random subsets $N_s$ for the diagonalizations involved in the selection of new configurations is set to $N_S = 10^3$ for most runs, with the exception of two runs in which $N_S = 10^4$ was used for benchmarking purposes. We observe that, in practice and for the values of $N_S$ in this study, as $N_S$ increases, so does the runtime of the evaluation function which, in our implementation of the algorithms, takes an overwhelming majority of the runtime. A computer with an Intel Xeon Platinum 8260 CPU with 96 cores (2.40GHz) and 3 TiB of memory is used for these runs.

As shown in~\Cref{fig:TrimCI}, the energy error remains above $10^{-2}$ for the smaller diagonalization subspace sizes. For the larger diagonalization subspace sizes, the energy error is significantly reduced, reaching lower energy errors compared to CIPSI, HCI, and ASCI for comparable subspace dimensions. However, we notice am uncontrolled non-monotonic behavior of the energy error as the algorithm progresses. 

In what follows, we provide a plausible explanation to the increased accuracy of TrimCI as compared to CIPSI, HCI, and ASCI. The selection function in TrimCI involves forming non-overlapping random subsets of the core set of configurations for diagonalization. Since the formation of these random subsets does not contain contain constraints on the existence of substantial numbers on non-zero matrix elements between the configurations in the subsets, it is likely that the configurations in each subset cannot confabulate to created superpositions to lower the energy in the corresponding subspace, yielding a likely effective ranking based mostly on the energy of each configuration, similar to the diagonal ranking heuristic that has been shown to be successful in this problem.

Due to the close practical connection with the diagonal ranking heuristic, which is known to find the ground state in this Hamiltonian, we do not rule-out the possibility that a given combination of the four hyperparameters that control TrimCI may result in a calculation capable of finding the exact ground state.

\begin{figure*}
    \centering
    \includegraphics[width=\linewidth]{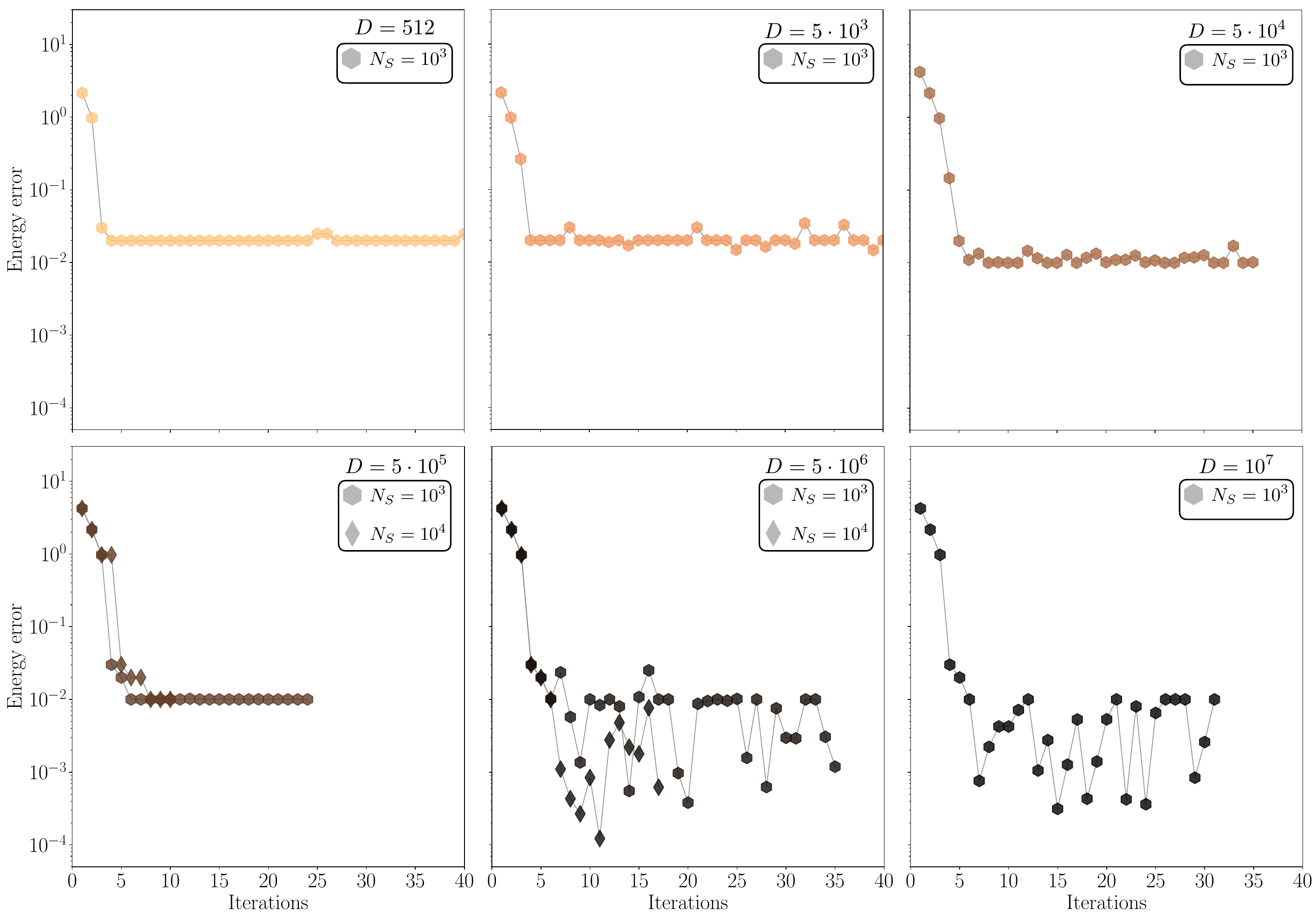}
    \caption{Energy error as a function of the TrimCI iterations for different choices of the diagonalization subspace size $D$, as indicated in each panel. The cases for which $D = 5 \cdot 10^5$ and $D = 5 \cdot 10^6$, contain two curves corresponding to $N_S = 10^3$ and $10^4$ as labeled by different markers.}
    \label{fig:TrimCI}
\end{figure*}

\section{Truncated Power Method}
\label{app:tpower}

\subsection{Definition and proof of convergence}
\label{sapp:tpower_convergence}

Suppose $A$ is a sparse  hermitian positive definite matrix of size $N\times N$.
Let $\chi_A$ be the sparsity of $A$, that is, the maximum number of nonzeros per column.
In the special case when $A$ is proportional to a spin or fermionic Hamiltonian with few-body interactions, the sparsity $\chi_A$
is polynomial in $\log{N}$ (the number of qubits).
Let $\lambda_1\ge \lambda_2 \ge \ldots \ge \lambda_N> 0$ be the eigenvalues of $A$.
We assume that the largest eigenvalue $\lambda_1$ is non-degenerate, that is,
$\lambda_1>\lambda_2$.
Define
the relative spectral gap
\[
\gamma = \frac{\lambda_1-\lambda_2}{\lambda_1} \in (0,1).
\]
Let $\psi\in \CC^N$ be the  principal eigenvector of $A$ such that 
$A|\psi\ra=\lambda_1|\psi\ra$ and $\|\psi\|=1$.
Below  we describe a classical algorithm proposed by Yuan and Zhang~\cite{yuan2013truncated} that takes as input a sparse reference state 
$\psi_{in}$
which has a non-negligible overlap with $\psi$ and
approximates the largest eigenvalue $\lambda_1$.
We make the following assumptions.
\begin{itemize}
\item
The principal eigenvector $\psi$ is sparse. Let $\chi$ be the number of nonzero entries in $\psi$.
\item
We are given a sparse reference vector $\psi_{in}\in \CC^N$  with $\|\psi_{in}\|=1$ and 
$\delta=|\la \psi_{in}|\psi\ra|^2>0$.  Let $\chi_{in}$ be the number of nonzeros in $\psi_{in}$.
\item The matrix $A$ is specified by an oracle that takes as input a column index $x$ and returns the list 
of all nonzero entries of $A$ in the column $x$. The reference vector $\psi_{in}$ is specified by a list of entries $(x,\la x|\psi_{in}\ra)$ with $\la x|\psi_{in}\ra\ne 0$.
\end{itemize}

Given a vector $\phi \in \CC^N$ and an index $j\in \{1,2,\ldots,N\}$, let $x^j(\phi) \in \{1,2,\ldots,N\}$ be the index $x$ (basis vector)
with the $j$-th  largest 
weight $|\la x|\phi\ra|$. If there are  ties, break them arbitrarily.
Thus for any vector $\phi$ the list $x^1(\phi),x^2(\phi),\ldots,x^N(\phi)$ is a permutation of  $\{1,2,\ldots,N\}$ 
such that $|\la x^j(\phi)|\phi\ra|$ is non-increasing with $j$.
For example, suppose $N=3$ and $|\phi\ra =0.5|1\ra - 0.9|2\ra + 0.1|3\ra$. Then $x^1(\phi)=2$, $x^2(\phi)=1$, $x^3(\phi)=3$.
Suppose our goal is to approximate $\lambda_1$ within an additive error $\epsilon$.
Define parameters
\begin{equation}
\label{eq:kstar}
    k_\star  =
     \frac{\chi}{\gamma^2} 
    \max{\left\{
    \frac{64}{\delta},
    \;
    \frac{9\lambda_1}{\epsilon}
    \right\}}
\end{equation}
and
\[
L_\star = \frac1{\gamma}\log{\left( \frac{k_\star  \gamma^2}{\delta \chi} \right)}
\]
\begin{algorithm}[H]
\caption{Truncated Power Method~\cite{yuan2013truncated}} 
\label{alg:tpower}
\begin{algorithmic}[1]  
\State \textbf{Input:} Sparse matrix $A$. Error tolerance $\epsilon$.  Sparse reference vector $\psi_{in}$ with  $|\la \psi|\psi_{in}\ra|^2\ge \delta$. 
\State \textbf{Output:} Real number $\mu$ such that $\lambda_1-\epsilon\le \mu\le \lambda_1$. 
\State $k \gets$ any integer larger than  $k_\star$ \Comment{Sparsity cutoff}
\State $L \gets$ any integer larger than  $L_\star$ \Comment{Number of  iterations}
\State Initialize $|\phi_0\ra \gets |\psi_{in}\ra$
\For{$t = 1$ to  $L$}
\State $|\theta_t\ra \gets A|\phi_{t-1}\ra$ \Comment{Evolve using the power method}
\State $F_t\gets \{x^1(\theta_t), \ldots,x^k(\theta_t)\}$ \Comment{Find $k$ most heavy indices}
\State $|\omega_t\ra \gets \sum_{x\in F_t} \la x|\theta_t\ra |x\ra$  \Comment{Truncate}
\State $|\phi_t\ra \gets |\omega_t\ra/\| \omega_t\|$. \Comment{Normalize}
\EndFor
\State \Return $\mu=\la \phi_L|A|\phi_L\ra$
\end{algorithmic}
\end{algorithm}

\begin{theorem}
\label{thm:tpower}
Algorithm~\ref{alg:tpower} outputs a real number $\mu$ satisfying 
$\lambda_1-\epsilon\le \mu\le\lambda_1$.
The algorithm makes at most $(L-1)k+\chi_{in}$ queries to the oracle specifying $A$
and performs roughly
$Lk\chi_A \log{(k\chi_A)} + \chi_{in}\chi_A \log{(\chi_{in} \chi_A)}$
arithmetic operations with complex numbers.
\end{theorem}

\newpage
\begin{proof}
Our proof closely follows~\cite{yuan2013truncated}. 
Given a  nonzero  vector $v\in\CC^N$,
let $T(v)$ be the tangent of the angle between $v$ and $\psi$ defined as
\[
T(v)= \frac{\|(I-|\psi\ra\la\psi|)|v\ra\|}{|\la\psi|v\ra|}.
\]
Clearly, $T(v)$ is scale-invariant, that is, $T(c v)=T(v)$ for any non-zero scalar $c$.
If $v$ is a unit vector then $T(v)=\sqrt{1-|\la\psi|v\ra|^2}/|\la\psi|v\ra|$.

Let $\phi_L$ be the final vector reached by Algorithm~\ref{alg:tpower}. Choose the phase of $\phi_L$ such that $\la \psi|\phi_L\ra$ is real non-negative.
Let $c=\la \psi|\phi_L\ra$ and $s=\sqrt{1-c^2}$. Then $T(\phi_L)=s/c$ and 
$|\phi_L\ra = c|\psi\ra + s|\psi^\perp\ra$ for some unit vector $\psi^\perp\in \CC^N$ orthogonal to $\psi$.
Since $\lambda_1$ is the largest eigenvalue of $A$ and $\phi_L$ is a unit vector, one must have
$\mu=\la \phi_L|A|\phi_L\ra\le \lambda_1$. From $\la \psi^\perp|A|\psi^\perp\ra\ge 0$ one gets
\begin{align*}
\lambda_1-\mu & = \lambda_1 - c^2 \lambda_1 - s^2 \la \psi^\perp|A|\psi^\perp\ra \le \lambda_1 - c^2 \lambda_1 \\
& = \lambda_1 s^2 \le \lambda_1 T(\phi_L)^2.
\end{align*}
Let
\[
\rho=\sqrt{\frac{\chi}{k}}.
\]
We will show that the sequence of vectors $\{\phi_t\}_{t\ge 0}$ generated by Algorithm~\ref{alg:tpower} obeys
$\la \psi|\phi_t\ra\ne 0$ and 
\be
\label{claim1}
T(\phi_t)\le \xi_t,
\ee
where the sequence $\{\xi_t\}_{t\ge 0}$ is defined recursively as
\[
\xi_0 = T(\psi_{in}), \qquad \xi_t = f(\xi_{t-1}), \qquad f(x)= \frac{(1-\gamma)x+\rho}{1-\rho(1-\gamma)x}.
\]
We will show that for any $k\ge 64\chi\gamma^{-2}\delta^{-1}$
 the sequence $\xi_t$ converges exponentially fast to the limiting point
$\xi_\star = \lim_{t\to \infty}\xi_t$ such that 
\be
\label{claim2}
|\xi_t - \xi_\star|\le e^{-t \gamma/2} \delta^{-1/2} 
\ee
for all $t\ge 0$ and the limiting point obeys
\be
\label{claim3}
\xi_\star \le  \frac{2\rho}{\gamma}.
\ee
Our choice of the number of iterations $L$  and Eq.~(\ref{claim2}) give
$|\xi_L-\xi_\star|\le \rho/\gamma$. From Eqs.~(\ref{claim1},\ref{claim3}) one gets
\[
T(\phi_L) \le \xi_L \le \xi_\star + |\xi_L-\xi_\star| \le  \frac{3\rho}{\gamma} = 
\frac{3 \chi^{1/2}}{\gamma k^{1/2}}.
\]
Plugging  this into the above upper bound on $\lambda_1-\mu$  gives
$\lambda_1-\mu\le \epsilon$ provided that $k\ge 9\lambda_1 \chi \gamma^{-2} \epsilon^{-1}$.
This proves the first part of Theorem~\ref{thm:tpower}.

{\em Comment:} in practice one could choose the sparsity cutoff $k$ and the number of iterations $L$
by  numerically computing the sequence $\xi_t$
and finding the smallest $L$ such that $\lambda_1' \xi_L^2 \le\epsilon$
for a given $k$, where $\lambda_1'\ge \lambda_1$ is an efficiently computable
upper bound on $\lambda_1$.
 The same argument as above then gives $\lambda_1-\mu\le \epsilon$.

Let us estimate the runtime of Algorithm~\ref{alg:tpower}. We shall say that a vector $v\in \CC^N$ has sparsity $m$ if $v$ has at most $m$ nonzeros.
Each vector $v$ that appears in Algorithm~\ref{alg:tpower} is stored
as a list of pairs $(x,\la x|v\ra)$ with $\la x|v\ra\ne 0$. If $v$ is $m$-sparse then one can compute $A|v\ra$ by making  $m$ calls to the oracle specifying $A$
and performing roughly $m\chi_A$ arithmetic operations. By construction, $\phi_t$ has sparsity
$\chi_{in}$  if $t=0$ and sparsity $k$ if $t\ge 1$. Vectors $\theta_t$ have sparsity $\chi_{in} \chi_A$ if $t=0$ and $k \chi_A$ if $t\ge 1$.
Vectors $\omega_t$ have sparsity $k$.  Step~8 can be implemented by sorting non-zero coordinates
of $\theta_t$ according to their weight $|\la x|\theta_t\ra|$. Sorting requires $O(m\log{m})$ operations, where $m$ is the sparsity of $\theta_t$.
This implies the claimed query complexity and the number of operations.

It remains to prove that the sequence of vectors $\phi_t$ generated by
Algorithm~\ref{alg:tpower} obeys  Eqs.~(\ref{claim1},\ref{claim2},\ref{claim3}).
We divide the proof into several lemmas.
First, let us  show that a single iteration of the standard (untruncated)  power method 
shrinks the tangent of the angle between $\psi$ and the evolved vector at least by the factor of $1-\gamma$.
\begin{lemma}
\label{lemma:1}
Consider any nonzero vector $\phi\in \CC^N$ such that $\la \psi|\phi\ra\ne 0$.
Then 
\[
T(A\phi) \le (1-\gamma) T(\phi).
\]
\end{lemma}
\begin{proof}
Since $T(\phi)$ is scale-invariant, we can assume wlog
 that $\phi$ is a unit vector. Choose the phase of $\phi$ such that $\la \psi|\phi\ra$ is real
 and $\la \psi|\phi\ra>0$. 
Let $c = \la \psi|\phi\ra \in (0,1]$ and $s=\sqrt{1-c^2}$. 
Then 
\[
T(\phi)=\frac{s}{c}.
\]
Write $\phi = c \psi  +s \psi^\perp$
for some vector $\psi^\perp\in \CC^N$ such that $\|\psi^\perp\|=1$ and $\la \psi|\psi^\perp\ra=0$. Then
\[
\| A\phi\|^2 = c^2 \lambda_1^2 + s^2 \| A \psi^\perp\|^2 \le c^2 \lambda_1^2 + s^2 \lambda_2^2.
\]
Let $\theta=A\phi/\|A\phi\|$. Then
\begin{align}
\label{a_lower}
|\la \psi|\theta\ra| & = \frac{|\la \psi|A|\phi\ra|}{\|A \phi\|} \ge \frac{c\lambda_1}{\sqrt{c^2 \lambda_1^2 + s^2 \lambda_2^2}} \nonumber \\
& =c \cdot \frac1{\sqrt{c^2 + s^2 (1-\gamma)^2}}.
\end{align}
Here we used the identity $\lambda_2/\lambda_1=1-\gamma$.
Let $a=|\la \psi|\theta\ra|$ and $b=\sqrt{1-a^2}$ so that 
\be
\label{T(Aphi)}
T(A\phi)=\frac{b}{a}.
\ee
From Eq.~(\ref{a_lower}) one gets
\[
a^2 \ge \frac{c^2}{c^2 + s^2 (1-\gamma)^2} 
\]
and
\[
b^2=1-a^2 \le \frac{(1-\gamma)^2 s^2 }{c^2 + (1-\gamma)^2 s^2}.
\]
It follows that $b^2/a^2 \le (1-\gamma)^2 s^2/c^2$. Substituting this into Eq.~(\ref{T(Aphi)}) proves the lemma.
\end{proof}

The next lemma shows that the truncation  performed at Steps~8,9 of Algorithm~\ref{alg:tpower} 
cannot reduce the overlap between $\psi$ and the evolved vector too much.
\begin{lemma}
\label{lemma:2}
Consider any vector $\theta \in \CC^N$ with $\|\theta\|=1$. 
The following is true for any integer
$k\ge \chi$.
Let $F=\{x^1(\theta),\ldots,x^k(\theta)\}$ be the set of $k$ most heavy indices of $\theta$ and 
\[
|\omega\ra = \sum_{x\in F} \la x|\theta\ra |x\ra.
\]
Then
\[
|\la \psi|\omega\ra| \ge |\la \psi|\theta\ra| - \rho \sqrt{1-|\la \psi|\theta\ra|^2}.
\]
Recall that $\rho = \sqrt{\chi/k}$.
\end{lemma}
\begin{proof}
Choose the phase of $\theta$ such that $\la \psi|\theta\ra$ is real and non-negative.
Let $c = \la \psi|\theta\ra \in [0,1]$ and $s=\sqrt{1-c^2}$. 
Write $|\theta\ra = c|\psi\ra + s|\psi^\perp\ra$ for some vector
$\psi^\perp\in \CC^N$ such that $\|\psi^\perp\|=1$ and $\la \psi|\psi^\perp\ra=0$.
Let $M=\mathrm{supp}(\psi)\subseteq \{1,\ldots,N\}$ be the support of $\psi$.
By assumption, $|M|=\chi$. 

Given a subset of indices $K\subseteq \{1,2,\ldots,N\}$, define a diagonal projector
\[
P_K = \sum_{x\in K} |x\ra\la x|.
\]
Then $|\omega\ra = P_F|\theta\ra$ and $|\psi\ra = P_M|\psi\ra$.
Therefore
\[
\la \psi|\omega\ra = \la \psi|P_F|\theta\ra = \la \psi|\theta\ra - \la \psi|(I-P_F)|\theta\ra.
\]
The triangle inequality gives
\[
|\la \psi|\omega\ra| \ge |\la \psi|\theta\ra| - |\la \psi|(I-P_F)|\theta\ra|.
\]
Moreover, since $|\psi\ra$ is supported on $M$, one has
\begin{align*}
|\la \psi|(I-P_F)|\theta\ra|
& = |\la \psi|P_M(I-P_F)|\theta\ra|\\ 
& \le \|P_M(I-P_F)|\theta\ra\|
= \|P_{M\setminus F}|\theta\ra\|.
\end{align*}
Thus
\be
\label{eq:tri}
|\la \psi|\omega\ra| \ge |\la \psi|\theta\ra| - \|P_{M\setminus F}|\theta\ra\|.
\ee
It remains to bound $\|P_{M\setminus F}|\theta\ra\|$ in terms of $s=\sqrt{1-|\la \psi|\theta\ra|^2}$.
Let 
\[
m= |M\setminus F|.
\]
If $m=0$ then $\|P_{M\setminus F}|\theta\ra\|=0$ and the lemma follows  from Eq.~(\ref{eq:tri}).
From now on assume $m\ge 1$.
Let
\[
\tau= \min_{x\in F} |\la x|\theta\ra|.
\]
Since $F$ contains $k$ most heavy indices of $\theta$, one must have 
$|\la x|\theta\ra|\le \tau$ for all $x\notin F$.
It follows that
\be\label{eq:missed}
\|P_{M\setminus F}|\theta\ra\|^2 =\sum_{x\in M\setminus F} |\la x|\theta\ra|^2 \le \tau^2 |M\setminus F| \le  m \tau^2.
\ee
On the other hand, 
\begin{align*}
|F\setminus M| & = |F|-|F\cap M| = |F| -|M|+|M\setminus F| \\
& = k - \chi + m.
\end{align*}
For each  $x\in F \setminus M$ one has $|\la x|\theta\ra|\ge \tau$. Thus
\be
\label{eq:outside}
\|(I-P_M)|\theta\ra\|^2 \ge \sum_{x\in F\setminus M} |\la x|\theta\ra|^2 \ge (k-\chi+m)\tau^2.
\ee
Recall that $|\theta\ra = c|\psi\ra + s|\psi^\perp\ra$.
Since $\psi$ is supported on $M$, one has 
$(I-P_M)|\theta\ra = s (I-P_M)|\psi^\perp\ra$.
Thus
\[
\|(I-P_M)|\theta\ra\|^2 \le s^2.
\]
Combining with Eq.~(\ref{eq:outside}) gives
\[
\tau^2 \le \frac{s^2}{k-\chi+m}.
\]
Substituting into Eq.~(\ref{eq:missed}) gives
\[
\|P_{M\setminus F}|\theta\ra\|^2 \le \frac{m}{k-\chi+m}s^2.
\]
Recall that $m\ge 1$. Furthermore, $m=|M\setminus F|\le |M|=\chi$.
Consider a function $f(m)=m/(k-\chi+m)$ with $m\in [1,\chi]$.
The assumption $k\ge \chi$ implies that $f(m)$ is non-decreasing for $m\in [1,\chi]$,
that is,
$f(m)\le f(\chi)=\chi/k$. This gives
\[
\|P_{M\setminus F}|\theta\ra\|^2 \le\frac{\chi}{k} s^2.
\]
Therefore
\[
\|P_{M\setminus F}|\theta\ra\| \le \sqrt{\frac{\chi}{k}}\,s
=\rho \sqrt{1-|\la \psi|\theta\ra|^2}.
\]
Plugging this into Eq.~(\ref{eq:tri}) proves the lemma.
\end{proof}

The next lemma converts the lower bound on the overlap with $\psi$
from Lemma~\ref{lemma:2} to an upper bound on the tangent of the angle with $\psi$.
\begin{lemma}
\label{lemma:2tangent}
Let $\theta,\omega\in \CC^N$ be the vectors
defined in Lemma~\ref{lemma:2}.
If $\rho T(\theta)<1$ then $\la\psi|\omega\ra\neq 0$ and
\[
T(\omega)\ \le\ \frac{T(\theta)+\rho}{1-\rho\,T(\theta)}.
\]
\end{lemma}
\begin{proof}
By assumption, $\theta$ is a unit vector.
The assumption $\rho T(\theta)<1$ implies $T(\theta)<\infty$, that is, $\la \psi|\theta\ra\ne 0$.
Let $a=|\la\psi|\theta\ra|\in(0,1]$ and $b=\sqrt{1-a^2}$. Then $T(\theta)=b/a$.
Let 
\[
\widehat \omega =\frac{\omega}{\|\omega\|}
\]
be the normalized version of $\omega$.
Since $\omega$ is obtained from a unit vector $\theta$ by zeroing out some coordinates,
one has $\|\omega\|\le 1$. As a consequence, normalization of $\omega$
can only increase its overlap with $\psi$. 
Thus
\be\label{eq:overlapphi}
|\la\psi|\widehat \omega\ra| \ge\ |\la\psi|\omega\ra|
\ \ge\ a-\rho b.
\ee
Here the second inequality uses Lemma~\ref{lemma:2}.
By assumption, $\rho T(\theta)=\rho b/a<1$.
Thus $a-\rho b=a(1-\rho b/a)>0$ and  the right-hand side of Eq.~(\ref{eq:overlapphi})
 is positive. In particular
$\la\psi|\widehat \omega\ra\neq 0$ which implies $\la \psi|\omega\ra\ne 0$.
Using Eq.~(\ref{eq:overlapphi}) 
one gets
\be
\label{Tomega_upper1}
T(\omega)=T(\widehat \omega)=\frac{\sqrt{1-|\la\psi|\widehat \omega\ra|^2}}{|\la\psi|\widehat \omega\ra|}
\le \frac{\sqrt{1-(a-\rho b)^2}}{a-\rho b}.
\ee
We have
\[
a-\rho b = a(1-\rho T(\theta))
\]
and
\begin{align*}
1-(a-\rho b)^2
& =1-a^2(1-\rho T(\theta))^2\\
& 
=\frac{1+T(\theta)^2-(1-\rho T(\theta))^2}{1+T(\theta)^2}.
\end{align*}
Here the second equality uses the identity  $a^2=1/(1+T(\theta)^2)$.
Plugging this into Eq.~(\ref{Tomega_upper1}) gives
\be
\label{Tomega_upper2}
T(\omega)\le
\frac{\sqrt{\,1+T(\theta)^2-(1-\rho T(\theta))^2\,}}{1-\rho T(\theta)}.
\ee
Finally, note that
\[
1+T^2-(1-\rho T)^2 = T^2 + 2\rho T - \rho^2T^2 \le (T+\rho)^2
\]
where $T\equiv T(\theta)$.
Plugging this into Eq.~(\ref{Tomega_upper2}) gives
\[
T(\omega)\le \frac{T(\theta)+\rho}{1-\rho T(\theta)},
\]
as claimed.
\end{proof}

\begin{lemma}
\label{lemma:3}
Suppose $k\ge k_\star$.
The sequence of vectors $\phi_t$ generated by
Algorithm~\ref{alg:tpower} obeys  Eqs.~(\ref{claim1},\ref{claim2},\ref{claim3}).
\end{lemma}
\begin{proof}
Recall that we consider a sequence $\{\xi_t\}_{t\ge 0}$ defined as
\[
\xi_0 = T(\psi_{in}), \qquad \xi_t = f(\xi_{t-1}), \qquad f(x)=\frac{(1-\gamma)x +\rho}{1-\rho(1-\gamma) x}
\]
for $t\ge 1$. Let us first establish properties of this sequence. 
\begin{prop}
\label{prop:Lip}
Suppose $\rho^2 \le \gamma/4$.
Consider an interval
\[
I= [0,\eta], \qquad \eta=\frac{\gamma}{8\rho}.
\]
For any  $x,y \in I$ with $x\le y$ one has
\be
\label{shrinking_factor}
0\le f(y)-f(x) \le  (1-\gamma/2) (y-x).
\ee
\end{prop}
\begin{proof}
A simple algebra gives
\[
f'(x)=\frac{(1-\gamma)(1+\rho^2)}{(1-\rho(1-\gamma)x)^2}.
\]
For all $x\in [0,\eta]$ one has $\rho(1-\gamma)x\le \rho x \le \gamma/8$.
Since $\rho^2\le \gamma/4$, one gets
\[
0\le f'(x)\le  \frac{(1-\gamma)(1+\gamma/4)}{(1-\gamma/8)^2}\le 1-\gamma/2.
\]
Now Eq.~(\ref{shrinking_factor}) follows from the mean value theorem.
\end{proof}
Note that $k\ge k_\star$ implies $k\ge 64 \chi \gamma^{-2} \delta^{-1}$, which 
is equivalent to 
\be
\label{rho_condition_2}
\rho^2 \le \frac{\gamma^2 \delta}{64}.
\ee
Since $\gamma \in (0,1)$, from
Eq.~(\ref{rho_condition_2})  one gets $\rho^2\le \gamma/4$, as required for Proposition~\ref{prop:Lip},
and  $k\ge \rho^{-2} \chi \ge \chi$, as required for Lemma~\ref{lemma:2}.
By definition, 
\be
\label{xi_0_vs_delta}
\xi_0 = \frac{\sqrt{1-|\la \psi|\psi_{in}\ra|^2}}{|\la \psi|\psi_{in}\ra|} = \frac{\sqrt{1-\delta}}{\sqrt{\delta}} \le \frac1{\sqrt{\delta}}.
\ee
Thus  $\xi_0\in [0,\eta]$ whenever  $\delta^{-1/2}\le \gamma/(8\rho)$, which is equivalent to Eq.~(\ref{rho_condition_2}).

Let $\xi_\star=\lim_{t\to \infty} \xi_t$ be the limiting point.
It must obey  $f(\xi_\star)=\xi_\star$.  Solving the equation gives
\[
\xi_\star=\frac{\gamma-\sqrt{\gamma^2-4(1-\gamma)\rho^2}}{2\rho(1-\gamma)}.
\]
Note that $\gamma^2-4(1-\gamma)\rho^2\ge 0$ due to Eq.~(\ref{rho_condition_2}).
We have
\begin{align}
\label{xi-fixed}
\xi_\star & =\frac{\gamma-\gamma \sqrt{1-4(1-\gamma)\rho^2\gamma^{-2}}}{2\rho(1-\gamma)} \nonumber \\
& 
\le \frac{\gamma - \gamma (1-4(1-\gamma)\rho^2\gamma^{-2})}{2\rho(1-\gamma)}=\frac{2\rho}{\gamma}
\end{align}
proving Eq.~(\ref{claim3}).
Note that $\xi_\star \in [0,\eta]$ provided that $2\rho/\gamma \le \gamma/(8\rho)$, that is, 
$\rho^2 \le \gamma^2/16$, which follows from Eq.~(\ref{rho_condition_2}).
The above shows that $\xi_0,\xi_\star \in [0,\eta]$.
\begin{prop}
\label{prop:monotone}
For all $t\ge 0$ one has $\xi_t \in [0,\eta]$ and 
\be
\label{dist_recursion}
|\xi_{t}-\xi_\star| \le (1-\gamma/2)^t |\xi_0 - \xi_\star|.
\ee
\end{prop}
\begin{proof}
Consider the case $\xi_\star \le \xi_0$.
Suppose $\xi_\star\le \xi_{t-1} \le \eta$ for some $t\ge 1$ (initially $t=1$).
Applying Proposition~\ref{prop:Lip} with $x=\xi_\star$ and $y=\xi_{t-1}$ one gets
\[
0\le f(\xi_{t-1}) - f(\xi_\star) \le (1-\gamma/2)(\xi_{t-1}-\xi_\star)
\]
Since $f(\xi_{t-1})=\xi_t$ and $f(\xi_\star)=\xi_\star$, one gets
\[
0\le \xi_t - \xi_\star \le  (1-\gamma/2)(\xi_{t-1}-\xi_\star).
\]
This implies $\xi_t \in [\xi_\star,\xi_{t-1}]$.
Proceeding inductively shows that   $\xi_t \in [0,\eta]$
and
\[
0\le \xi_t - \xi_\star \le (1-\gamma/2)^t (\xi_0-\xi_\star)
\]
for all $t\ge 0$. 
This proves  the claim 
 in the case $\xi_\star\le \xi_0$.
The  case $\xi_\star \ge \xi_0$ is completely analogous.
\end{proof}
From Eqs.~(\ref{rho_condition_2},\ref{xi-fixed}) one gets $\xi_\star \le \delta^{1/2}/4 \le 1/4$. 
From  Eqs.~(\ref{xi_0_vs_delta}) one gets $\xi_0\le \delta^{-1/2}$. Thus
$|\xi_0-\xi_\star|\le \max{(\xi_0,\xi_\star)}\le \delta^{-1/2}$. Substituting this into Eq.~(\ref{dist_recursion}) gives
\[
|\xi_{t}-\xi_\star|  \le (1-\gamma/2)^t \delta^{-1/2} \le e^{-t \gamma/2} \delta^{-1/2}.
\]
This proves Eq.~(\ref{claim2}).

Finally, let us prove Eq.~(\ref{claim1}), that is, $T(\phi_t)\le \xi_t$ for all $t\ge 0$.
Let $\theta_t$ and $\omega_t$ be the vectors defined  in Algorithm~\ref{alg:tpower}.
By definition, $|\theta_t\ra=A|\phi_{t-1}\ra$.
We shall use induction in $t$ to prove that 
\be
\label{induction_1}
\rho T(\phi_t)<1
\ee
and
\be
\label{induction_2}
T(\phi_t) \le \xi_t.
\ee
The base of induction is $t=0$. 
By definition,  $T(\phi_0)=\xi_0$.
The inclusion $\xi_0\in [0,\eta]$ gives
$\rho T(\phi_0) \le \rho \eta =\gamma/8<1$ since $\gamma \in (0,1)$.
This proves Eq.~(\ref{induction_1},\ref{induction_2}) at $t=0$.

Consider the induction step. Suppose  $t\ge 1$ and we have already proved that 
\be
\label{induction_1'}
\rho T(\phi_{t-1})<1
\ee
and
\be
\label{induction_2'}
T(\phi_{t-1}) \le \xi_{t-1}.
\ee
By definition,  $|\theta_{t}\ra = A|\phi_{t-1}\ra$.
From Eq.~(\ref{induction_1'}) one gets $\la \psi|\phi_{t-1}\ra\ne 0$ 
(since otherwise the tangent $T(\phi_{t-1})$ becomes infinite).
Thus we can apply Lemma~\ref{lemma:1} with 
$\phi=\phi_{t-1}$ obtaining 
\be
T(\theta_t) = T(A\phi_{t-1}) \le (1-\gamma)T(\phi_{t-1}).
\ee
Let us apply Lemma~\ref{lemma:2tangent} with the vectors
$\theta=\theta_t/\|\theta_t\|$ and 
$\omega=\omega_t/\|\theta_t\|$. The condition $\rho T(\theta)<1$
of the lemma is satisfied. Indeed, from  Eq.~(\ref{induction_2'}) 
and the inclusion $\xi_{t-1}\in [0,\eta]$
one gets
\begin{align*}
\rho T(\theta)& =\rho T(\theta_t) \le \rho (1-\gamma) T(\phi_{t-1})\\
& 
\le
 \rho \xi_{t-1}
\le
 \rho \eta = \frac{\gamma}8 <1.
\end{align*}
Thus applying Lemma~\ref{lemma:2tangent} is justified and the lemma gives
\begin{align*}
T(\phi_t)& =T(\omega_t) \le \frac{T(\theta_t) + \rho}{1-\rho T(\theta_t)} \\
&  \le \frac{(1-\gamma) T(\phi_{t-1}) + \rho}{1-\rho (1-\gamma)T(\phi_{t-1})}
=f(T(\phi_{t-1})).
\end{align*}
Here we used scale-invariance of the tangent function.
Recall that the function $f(x)$ is non-decreasing for $x\in [0,\eta]$, see Proposition~\ref{prop:Lip}.
From Eq.~(\ref{induction_2'}) and the inclusion $\xi_{t-1}\in [0,\eta]$ one gets
\[
T(\phi_t) \le f(T(\phi_{t-1}))\le f(\xi_{t-1}) = \xi_t
\]
proving Eq.~(\ref{induction_2}). Finally, $\xi_t \in [0,\eta]$
gives $\rho T(\phi_t) \le \rho \xi_t \le \rho \eta =\gamma/8<1$ proving Eq.~(\ref{induction_1}).
This completes the induction step and the proof of Eq.~(\ref{claim1}).

\end{proof}

\end{proof}

\subsection{Comparison to truncated Arnoldi's method}
\label{sapp:comparison_tarn}

We compare the truncated power method (TPM) to truncated Arnoldi's method (TAM; see \Cref{ssapp:non_diagonalization_based}) for the Hamiltonian defined in \Cref{sec:setup}.
We test two versions of the TPM:
\begin{enumerate}
    \item the original expectation value method as in \Cref{alg:tpower}, to which the proof of convergence in the previous subsection applies.
    \item the modified, diagonalization-based variant described in \Cref{sec:sis_intro}: instead of taking our energy estimate to be the expectation value of the Hamiltonian as in Step 12 of \Cref{alg:tpower}, we project and diagonalize on the support of the final vector $|\phi_L\rangle$. As noted above, this cannot yield a higher energy error than the expectation value method.
\end{enumerate}
We remark that at each iteration, 2. yields either equal of lower energies than 1.

\begin{figure}[t]
    \centering
    \includegraphics[width=\linewidth]{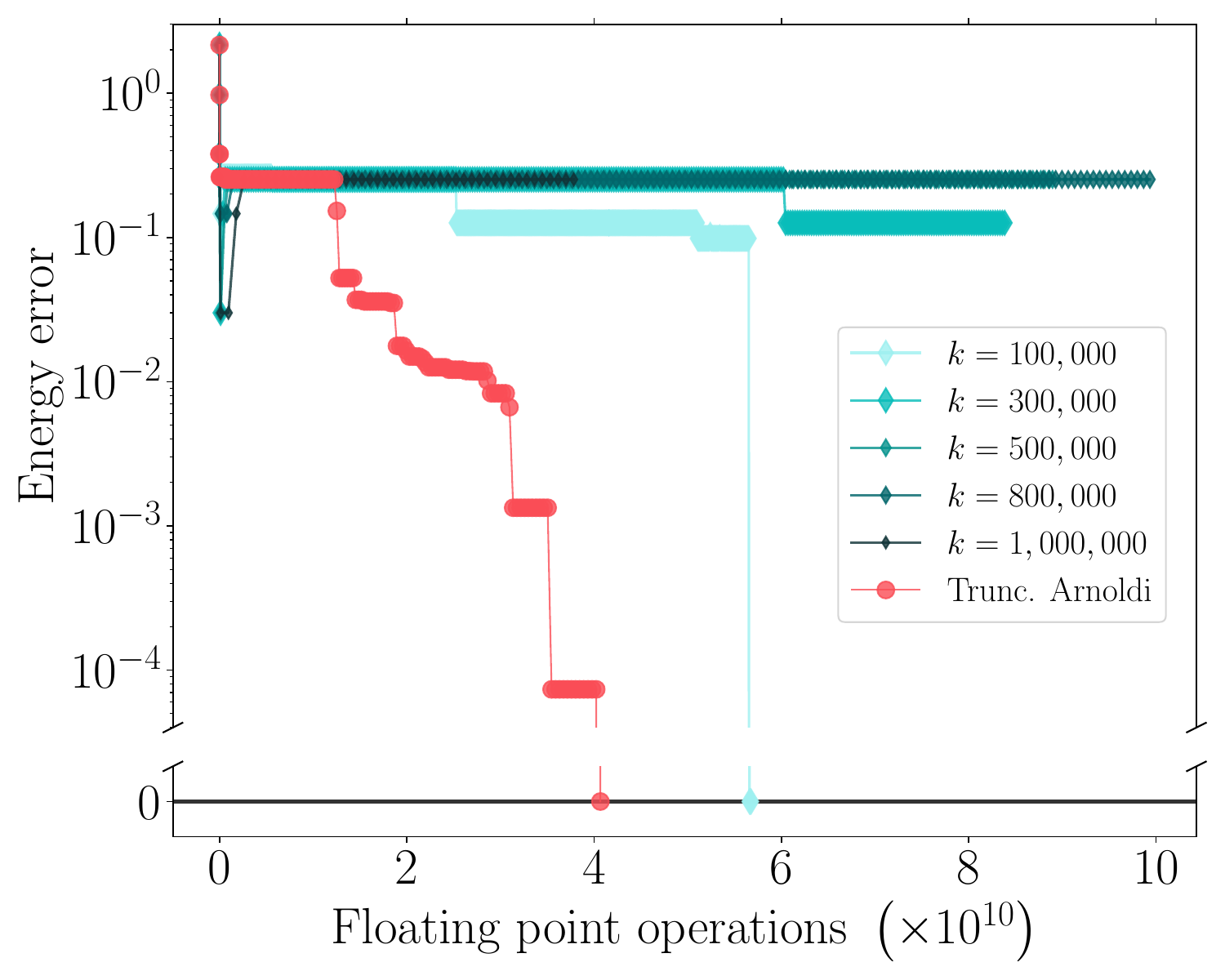}
    \caption{Comparison between truncated Arnoldi's method (TAM) with $M = 5\times10^4$ (see~\cref{fig:trunc_arn_all_runs}) and the diagonalization-based variant of truncated power method (TPM), in terms of energy error as a function of floating point operations. Several values of the cutoff $k$ for TPM are used, while only the most performance run for TAM is shown. For TPM, we show the energy error resulting in the projection and diagonalization of the Hamiltonian to the support obtained at each iteration. These results illustrate that, even though TPM possesses a provable convergence guarantee for sufficiently high $k$, for this problem instance TAM is more performant in practice.}
    \label{fig:tpm_tam_comparison}
\end{figure}

To compare TPM to TAM, we employed the second, diagonalization based variant of TPM, since it yields higher performance at more comparable cost to TAM.
Even so, a direct comparison based on subspace dimension makes less sense in this case than it does for comparing the other sparse iterative solvers, since the diagonalization dimension is always fixed to $k$ for TPM, and does not grow over the course of the iteration. Additionally, TAM requires an additional orthogonalization step not present in TPM. For these reasons, we instead compare the energy errors as function of the number of floating point operations.

For this Hamiltonian, $\chi=512$, $\gamma\approx2.9\times10^{-5}$, and $\delta\approx3.4\times10^{-11}$, yielding a value of the theory cutoff $k_\star$ as given by \eqref{eq:kstar} that is at least $\sim4.6\times10^{23}$. Storing this vector with single-precision float components in memory would require $1.6 \times 10^9 \; \mathrm{PB}$, orders of magnitude beyond the reach of current supercomputers. Instead of relying on the theoretical estimate, we test a range of smaller values of $k$, for which the performance of TPM will be heuristic. The results are shown in \Cref{fig:tpm_tam_comparison}.

As reflected in \Cref{fig:tpm_tam_comparison}, both TAM and TPM solve this guided sparse ground state problem exactly. We ran TPM for several choices of the cutoff $k$, for as many as 700+ iterations. For all choices of $k$, we observe that the cost to solution is greater for TPM as compared to TAM, as measured by the number of floating point operations. 

For completeness, we also tested the original expectation value version of TPM in~\cref{fig:tpm_vs_iter}. We observe the monotonic decrease of the energy error when the energies are estimated directly from the power method vector. The estimation obtained after diagonalizing in the subspace spanned by the support of the power-method vector yields a lower estimate than the direct expectation value, at the expense of not monotonically decreasing with progressing iterations.

\begin{figure}[t]
    \centering
    \includegraphics[width=\linewidth]{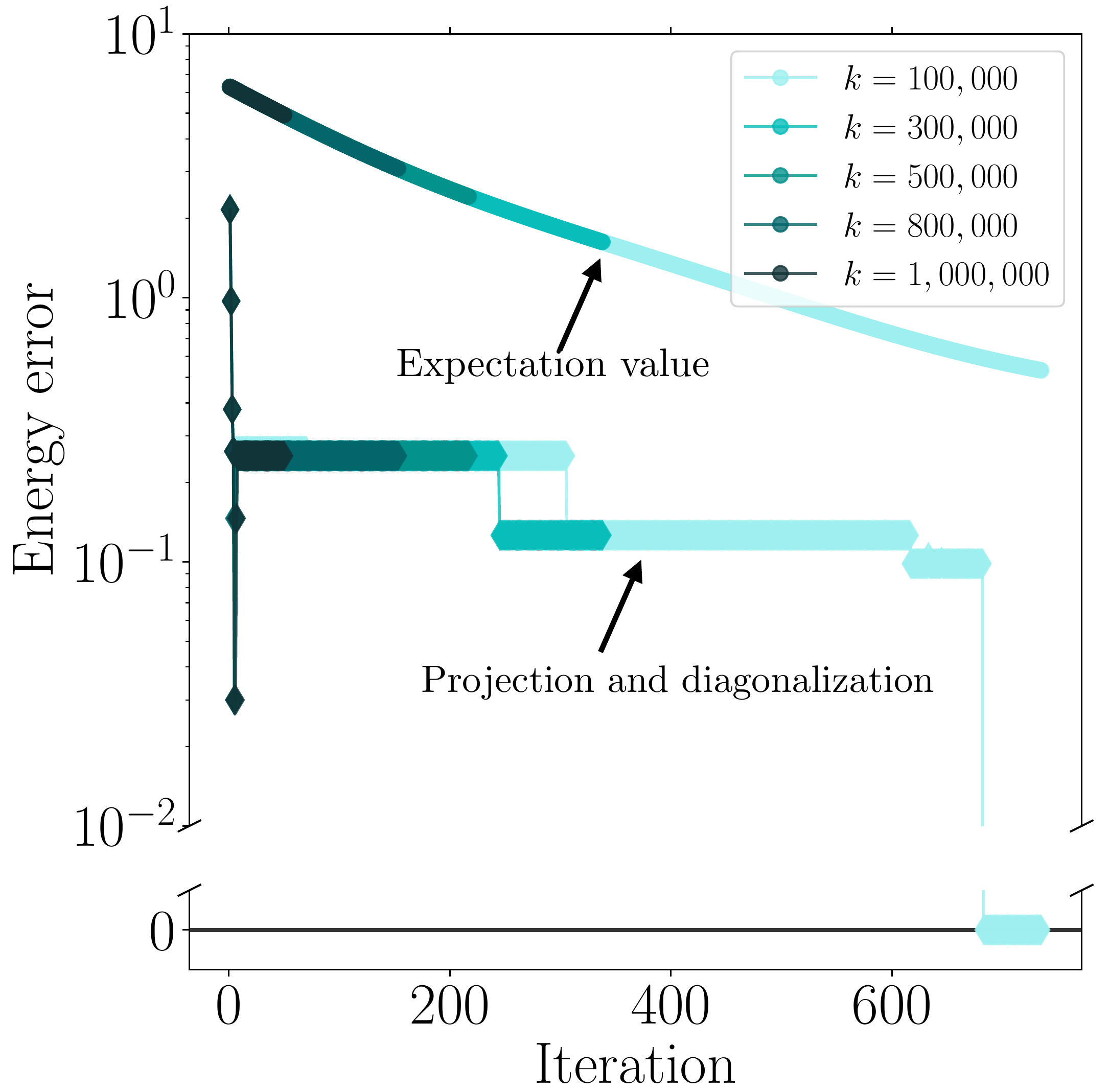}
    \caption{Comparison of the energy error at each iteration obtained from the truncated power method. The energy error obtained from the direct expectation value is compared to the energy obtained by diagonalizing in the subspace reached at each iteration of TPM. Different truncation values $k$ are considered.}
    \label{fig:tpm_vs_iter}
\end{figure}

\section{Tensor network simulation of SKQD}
\label{app:tn_sims}

Throughout this work, tensor network simulations are primarily used as a tool to guide and benchmark quantum experiments. At the same time, they also serve as a classical adversary: for a general quantum advantage, classical simulation of the SKQD subroutine must become prohibitively expensive, a threshold that is not reached for the Hamiltonian studied experimentally in this work. We detail our use of tensor networks in both of these roles below.

\paragraph{As a guiding and diagnostic tool.}
Once a method of approximating time-evolution is determined, the main experimental parameters of SKQD are the Krylov dimension and the timestep $\Delta t$.
The former can be updated adaptively if necessary.
In principle, the same is true of the latter, but it is preferable to start out with a good guess.
The proof of convergence of SKQD provides a theoretical value of $\Delta t = \pi/\|H\|$ that should permit convergence, but in practice, more rapid convergence has been obtained with larger values in the range of $15\times$ to $30\times$ the theoretical value~\cite{yu2025quantum}.
In addition, for high-dimensional problems, the spectral norm $\|H\|$ is classically challenging to calculate, so we instead upper bound it by the $1$-norm of the Pauli coefficients in $H$.

Based on~\cite{yu2025quantum} and earlier test Hamiltonians in the family presented in \Cref{sec:construction}, we selected $\Delta t = 25\pi/\|H\|$.
We then used simulations with belief propagation to validate this choice and ensure that we were not setting the experiment up for failure. In particular, we carry out low bond-dimension ($\chi=16$) simulations of the Rustiq-compiled, pre-AQC quantum circuits and leverage knowledge of the support configurations to query their amplitudes in each Krylov basis state. Following this, we estimate the number of expected support configurations given a particular sampling rate, choosing the optimal value of $\Delta t$ based on qualitative trends -- see Fig.~\ref{fig:tn_dt_sweep}. This value is then used for the circuits input to the approximate quantum compilation routine discussed in Sec.~\ref{app:aqc_tensor}. We also use these simulations to provide a sense for the number of samples needed per basis state for different target energies and Krylov dimensions in the absence of noise (see Fig.~\ref{fig:tn_sampling_rate}).

Although this procedure leverages prior knowledge of the support configurations, it is motivated by efficiency rather than necessity: absent this knowledge, we could alternatively sample from the circuits, diagonalize the projected Hamiltonian, and use the ground state energy as a guide to choose $\Delta t$ to similar effect. Furthermore, this strategy does not require the full SKQD loop to be classically simulable: one could imagine Hamiltonians where basis states $|\psi_k\rangle$ are simulable for only $k\leq k_{\textrm{classical}}$, and $\Delta t$ is chosen based on early trends.

Separate from their role in the experimental design, we also use tensor network simulations to validate experimental performance by comparing the number of support configurations found in experiment to theoretical expectation based on simulations of the same (post-AQC) circuits, as in Figs.~\ref{fig:results_configurations} and \ref{fig:tn_kQ_j}. To produce these plots, we simulated the full suite of quantum circuits  (i.e, $1\leq k_{\textrm{AQC}} \leq 20$ and $k_Q=0,1$ and $2$). We found that a bond dimension of $\chi=64$ was sufficient to achieve estimated fidelities above 0.99. As with the simulations of the pre-AQC circuits, we estimate the expected number of configurations from the probability amplitude of all support configurations for each basis state.

\begin{figure}
    \centering
    \includegraphics[width=1\linewidth]{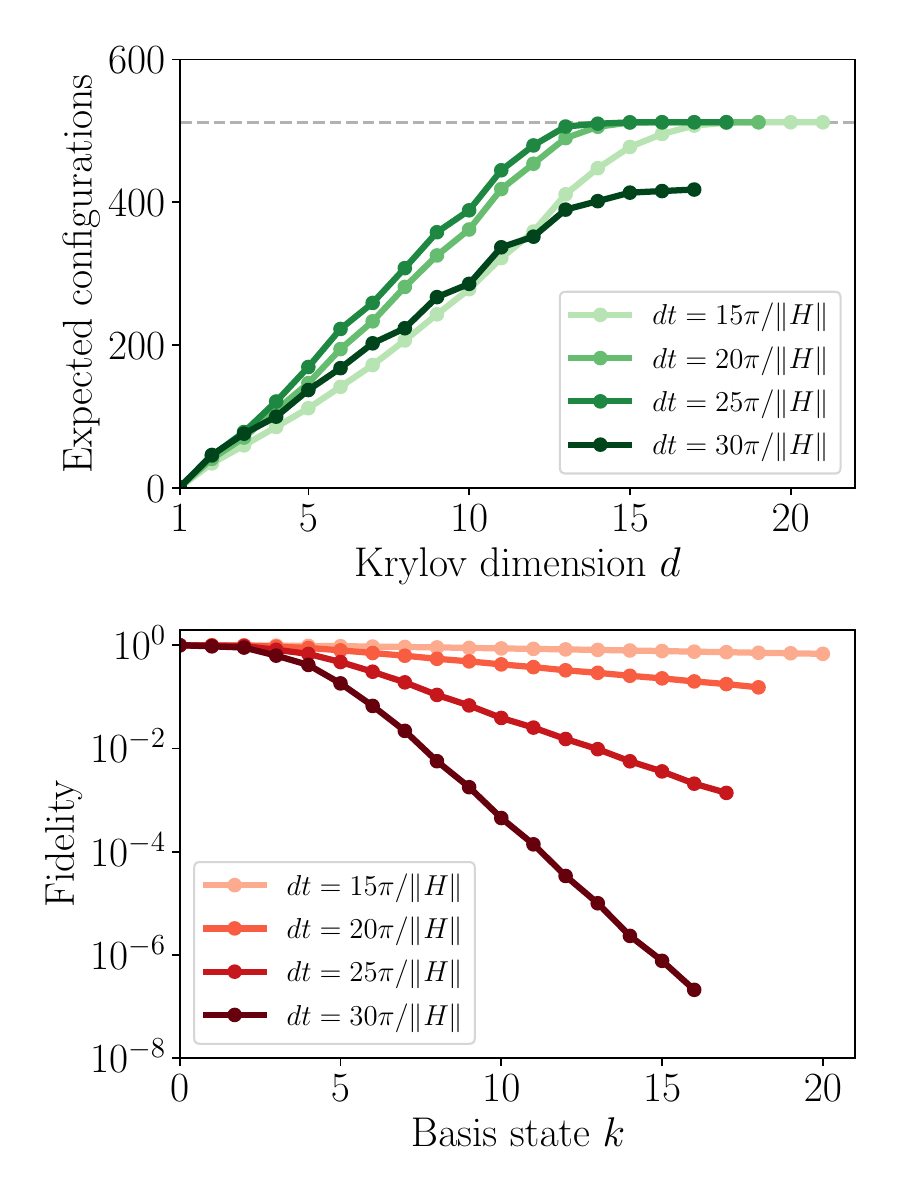}
    \caption{Using tensor network simulations to select the hyperparemeter $dt$. {\bf (Top)} The expected number of support configurations found as a function of Krylov dimension. We use low-fidelity tensor network simulations ($\chi=16$) to quickly estimate the rate of convergence for different choices of $dt$. In particular, we carry out simulations of pre-AQC time evolution circuits and estimate the cumulative expected good configuration count from the probability amplitude of each good configuration for each Krylov basis state $k$. {\bf (Bottom)} The estimated simulated state fidelity for the  basis state $k$ as a function of $dt$.}
    \label{fig:tn_dt_sweep}
\end{figure}

\begin{figure}
    \centering
    \includegraphics[width=1\linewidth]{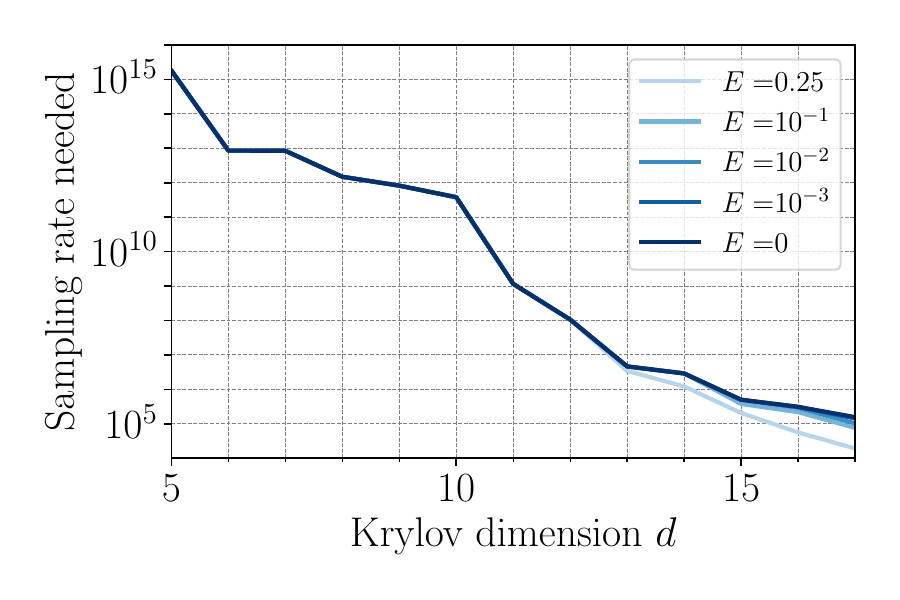}
    \caption{Estimated samples needed per basis state to achieve different target energies for $dt=25\pi/\|H\|$. Using tensor network simulations of the SKQD time-evolution circuits, we can estimate the number of samples needed as a function of Krylov dimension $d$. As these estimates are computed from low fidelity simulations (see Fig.~\ref{fig:tn_sampling_rate}, lower panel), this data provides a rough guide for experiment design rather than an exact expectation. Despite this, we find that increasing the bond dimension does not significantly change these estimates.}
    \label{fig:tn_sampling_rate}
\end{figure}

\begin{figure*}
    \centering
    \includegraphics[width=1\linewidth]{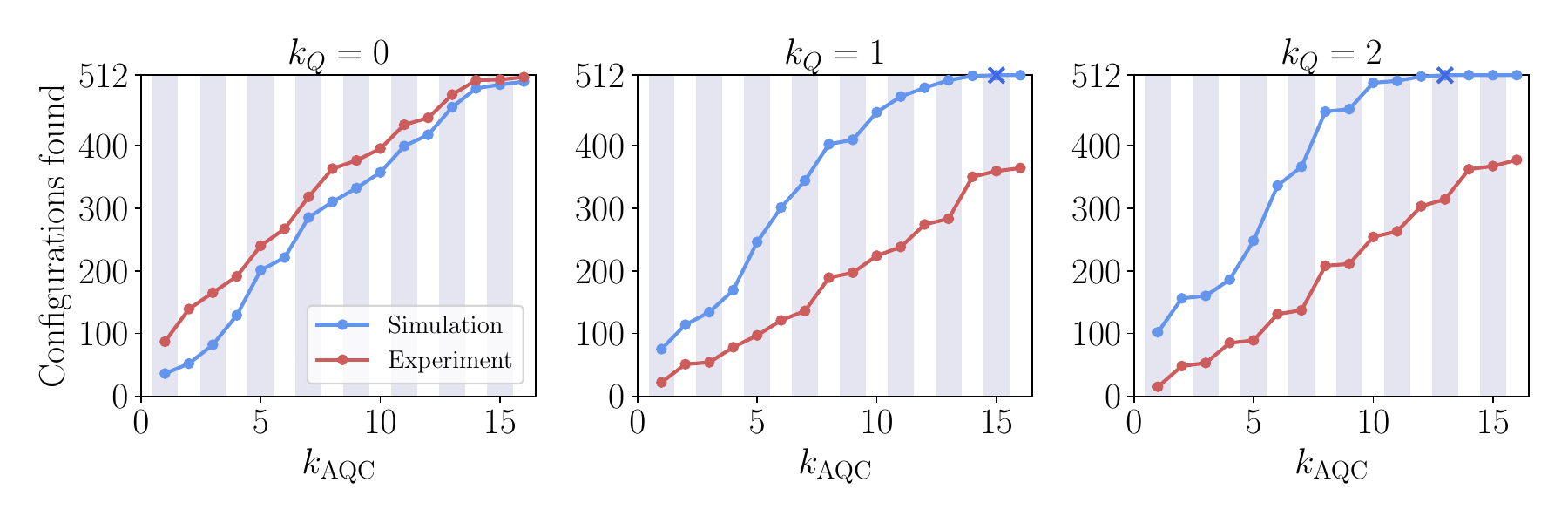}
\caption{Configurations found in experiment vs. expectation from tensor network simulations, broken down by $k_Q$. These results are in complement to Fig.~\ref{fig:results_configurations}, which presents the same data set cumulatively across all 49 time‑evolution circuits. Blue crosses indicate expected convergence. In total, 97 million shots were collected across the various circuits and were non‑uniformly distributed according to Fig.~\ref{fig:shot_allocation}. We fail to find all support configurations when sampling only the AQC circuits $k_{\textrm{AQC}}\leq 16$ without additional Trotter steps (i.e., $k_Q=0$), consistent with expectations from simulation. In contrast, simulations predict that the subsets containing both $k_Q=1$ and $k_Q=2$ are sufficient to observe all support configurations in the absence of noise. However, these latter circuits are comparatively deeper (see Fig.~\ref{tab:2qdepth}), and in practice we find that combining all 49 circuits provides the most effective strategy for convergence.}
    \label{fig:tn_kQ_j}
\end{figure*}

\paragraph{As a classical adversary.} 
As shown in Fig.~\ref{fig:tn_dt_sweep}, we find that tensor network simulation of the pre‑AQC circuits yields Krylov basis states of sufficient quality that all support configurations are discoverable with modest sampling overhead, even when the global fidelity is low. We estimate that roughly $10^{5}$ samples from each of the first $20$ basis states would suffice to recover all support configurations (see Fig.~\ref{fig:tn_sampling_rate}), and observe sampling times on the order of one second per sample on a single CPU, consistent with Ref.~\cite{rudolphSimulatingSamplingQuantum2025}. We expect that this cost could be reduced through further optimization, parallelization, and the use of GPUs. As an alternative strategy, one can leverage direct access to the intermediate AQC target states generated during the compilation workflow and sample from them using MPS methods. Drawing $10^{5}$ samples from $|\psi_k^{(\chi=256)}\rangle$ for each $k \leq 20$ reliably yields all support configurations.

\end{document}